\documentclass[12pt]{article}
\usepackage[utf8]{inputenc}
\usepackage{amsmath}
\usepackage{amsthm}
\usepackage{amssymb}
\usepackage{caption}
\usepackage{subcaption}
\usepackage{algorithm}
\usepackage{algcompatible}
\usepackage{algpseudocode}
\newsavebox{\measurebox}
\usepackage{graphicx,psfrag,epsf}
\usepackage{multirow}
\usepackage{mathtools}
\usepackage{enumerate}
\usepackage{arydshln}
\usepackage{bbm}
\usepackage{enumitem}
\usepackage{natbib}
\usepackage{url}
\usepackage{xcolor}
\usepackage[export]{adjustbox}
\theoremstyle{definition}

\usepackage{geometry}
\newcommand{\blind}{0}
\usepackage{comment}

\usepackage{xr}
\makeatletter
\newcommand*{\addFileDependency}[1]{
  \typeout{(#1)}
  \@addtofilelist{#1}
  \IfFileExists{#1}{}{\typeout{No file #1.}}
}
\makeatother

\newcommand*{\myexternaldocument}[1]{%
    \externaldocument{#1}%
    \addFileDependency{#1.tex}%
    \addFileDependency{#1.aux}%
}
\myexternaldocument{supplimentary}

\newcommand\bm[1]{\boldsymbol #1}

\DeclareMathAlphabet\mathbfcal{OMS}{cmsy}{b}{n}
\DeclareMathOperator\sign{sign}
\DeclareMathOperator*{\argmin}{arg\,min}
\newcommand{\E}{\textnormal{E}}

\newcommand{\cov}{\textnormal{cov}}

\providecommand{\keywords}[1]
{
  \small	
  \textbf{\textit{Keywords---}} #1
}
\newtheorem{theorem}{Theorem}

\newtheorem{lemma}{Lemma}

\begin{document}

\if0\blind
{
  \title{\bf Two Gaussian regularization methods for time-varying networks}
  \author{Jie Jian, Peijun Sang, and Mu Zhu \\
    Department of Statistics and Actuarial Science, University of Waterloo}
  \maketitle
} \fi

\if1\blind
{
  \bigskip
  \bigskip
  \bigskip
  \begin{center}
    {\LARGE\bf Title}
\end{center}
  \medskip
} \fi

\bigskip

\begin{abstract}
We model time-varying network data as realizations from multivariate Gaussian distributions with precision matrices that change over time. To facilitate parameter estimation, we require not only that each precision matrix at any given time point be sparse, but also that precision matrices at neighboring time points be similar. We accomplish this with two different algorithms, by generalizing the elastic net and the fused LASSO, respectively. Our main focuses are efficient computational algorithms and convenient degree-of-freedom formulae for choosing tuning parameters. We illustrate our methods with two simulation studies. By applying them to an fMRI data set, we also detect some interesting differences in brain connectivity between healthy individuals and ADHD patients.
\end{abstract}\hspace{10pt}


\section{Introduction}
\label{sec:intro}
In many applications, we need to identify and estimate associations and interactions among a set of random variables to uncover their latent topological structures, such as protein-protein interaction networks \citep{sato2006partial}, ecological networks in the microbial interactome \citep{dohlman2019mapping}, and gene regulatory networks \citep{emmert2012statistical}. 
Modeling an undirected static network has been studied since last century \cite[see e.g.][]{whittaker1990, lauritzen1996graphical}. More often than not, however, network structures evolve over time in response to both endogenous and exogenous factors; therefore, the assumption of the relational structure being fixed is too restrictive. A motivating example is the scientifically important task of detecting functional connectivity of human brains. The human brain is a complex dynamical system composed of many interacting regions, and brain connectivity can be represented as a network composed of a set of random variables (nodes) interconnected by a set of interactions (edges). As the brain is actively yielding thoughts and ideas, along with changes in arousal, awareness, and vigilance, modelling brain connectivity as a static network where a single snapshot of the network is observed can be misleading. This has catalyzed emerging interest in estimating time-varying networks where the data consist of serial snapshots of the networks that evolve over time. 

The literature concerning dynamic networks has appeared only recently \citep{zhou2010, kolar2010, monti2014estimating,xue2020time,bartlett2021two}, while there is rich literature on estimating a static network \cite[see e.g.][]{meinshausen2006high,yuan2007model,friedman2008sparse,peng2009partial}, among which the Gaussian Graphical Model (GGM) is particularly useful. Consider a $p$-dimensional GGM $(x_1, \cdots, x_p)^{\top} \sim N( \bm{0,\bm{\Sigma}} )$. Denote the precision matrix as $\bm{\Sigma}^{-1}=(\sigma^{ij})_{p \times p}$. Under multivariate normality, 
zero $\sigma^{ij}$ indicates conditional independence between $x_i$ and $x_j$ given the remaining ones. Therefore, a network can be encoded by conditional dependencies in a GGM, where nodes represent random variables and the edge connecting the $i^{th}$ and $j^{th}$ nodes is decided by whether $\sigma^{ij}$ is zero. The partial correlation between $x_i$ and $x_j$, denoted as $\rho_{ij}$, can be expressed as $-\sigma^{ij}/\sqrt{\sigma^{ii}\sigma^{jj}}$, and can be estimated by performing a multiple regression $x_{i} = \sum_{j\neq i} \beta_{ij} x_{j} +\epsilon_{i}$ sequentially for $i=1,\cdots,p$, where the prediction error variance $Var(\epsilon_i)=1/\sigma^{ii}$, and the regression coefficient $\beta_{ij} = -\sigma^{ij}/\sigma^{ii}=\rho_{ij} \sqrt{\sigma^{jj}/\sigma^{ii}}$ \citep{meinshausen2006high}. Accordingly, estimating the partial correlations as regression coefficients can characterize the graphical structures, however, the general regression estimators are never exactly zero due to the high-dimension-low-sample-size setting of the problem and the sampling variation in the data. To obtain a sparse network and make the regression problem well-posed, regularization technique is employed \citep{meinshausen2006high}. \cite{peng2009partial} proposed an efficient method, referred to as the Sparse PArtial Correlation Estimation or \textit{SPACE}. Suppose that $\bm{X}_i$ is the $i$th column of $\bm{X}_{n \times p}$ that consists of $n$ i.i.d. observations from the GGM, $\bm{\sigma}_{p\times 1}=(\sigma^{11},\cdots,\sigma^{pp})^{\top}$, and $\bm{\theta}_{p(p-1)/2\times 1}=(\rho_{12},\cdots,\rho_{p-1,p})^{\top}$. \textit{SPACE} is to estimate the partial correlations by minimizing a penalized likelihood function with LASSO penalty:
\begin{align}
 L_{LASSO}(\bm{X} ,\bm{\theta},\bm{\sigma},n,\lambda_1) 
& =  \frac{1}{n} \sum\limits_{i=1}^{p} \left\| \bm{X}_i - \sum\limits_{j < i } \rho_{ji} \sqrt{\dfrac{\sigma^{jj}}{\sigma^{ii}}} \bm{X}_j - \sum\limits_{j > i } \rho_{ij} \sqrt{\dfrac{\sigma^{jj}}{\sigma^{ii}}} \bm{X}_j  \right\|^2 +\lambda_1 \left\|\bm{\theta}\right\|_1, \label{rhoRegression2}
\end{align}
where $\lambda_1$ is a prespecified regularization parameter to control the strength of shrinkage and variable selection of $\bm{\theta}$. By stacking all columns in $\bm{X}$ into a long response vector $\bm{Y}$ and filling nonzero blocks in an $np\times p(p-1)/2$ sparse matrix $\tilde{\bm{X}}$  by $\sqrt{\sigma^{jj}/\sigma^{ii}} \bm{X}_j$, \eqref{rhoRegression2} is converted to a standard LASSO problem 
\begin{align}
 L_{LASSO}(\bm{X} ,\bm{\theta},\bm{\sigma},n,\lambda_1) 
 =  L_{LASSO}(\bm{Y},\tilde{\bm{X}},\bm{\theta},n,\lambda_1) 
 = \frac{1}{n}  \| \bm{Y} -  \tilde{\bm{X}} \bm{\theta} \|^2 +\lambda_1 \left\|\bm{\theta}\right\|_1, \label{standardLASSO}
\end{align}
which can be solved efficiently \citep{peng2009partial}. 

In this article, we develop statistical methods to identify the associations and their dynamic changes in discrete time-varying networks based on \textit{SPACE}. This goal is achieved under the assumption that the changes in the temporal network from one time point to the next are smooth, which encourages the regularization on the difference of partial correlations between adjacent time points. The regularization techniques we use in this paper include both $l_1$  \citep{Tibshirani1996lasso, Tibshirani2005fuse} and $l_2$ regularization \citep{zou2005regularization}, leading to two different algorithms. Our approach of modelling time-varying networks differs from those proposed by \cite{xue2020time} and by \cite{bartlett2021two}. Though partial correlations are also employed by \cite{xue2020time} to encode network structure at each time point, they are treated as functions of time via regression splines to capture time-varying network structures. Due to B-spline bases having local support, sparse networks are obtained by imposing a group LASSO penalty on the coefficient vectors in the regression splines. In contrast to \cite{xue2020time} and our method, \cite{bartlett2021two} propose a Bayesian framework to separate two types of sparsity---sparsity across time and sparsity across variables---when modelling time-varying networks.

We structure the remainder of this paper as follows. In Section \ref{sec:meth}, we present our time-varying network models with two different penalties on temporal dissimilarity. In Section \ref{sec:algo}, we describe high-level computational details. Our main contributions are: first, we use the alternating direction method of multipliers (ADMM) to come up with computationally efficient algorithms by parameterizing $l_1$ and $l_2$ penalties differently (Section~\ref{subsec:ADMM}); second, by generalizing existing results in the literature, we derive approximate degree-of-freedom formulae to characterize the effective complexity of our solutions and to facilitate the selection of tuning parameters (Section~\ref{sec:ModelSelection}); third, both when implementing the ADMM iterations and when computing the degrees of freedom, we use a few specific tricks to handle the inversion of some potentially large matrices. In Section \ref{sec:simulation}, we illustrate the performance of our methods in two different simulation scenarios. In Section \ref{sec:application}, we apply our methods to fMRI data of human brains in both healthy individuals and those suffering from the attention deficit hyperactivity disorder (ADHD); the most marked differences between the two groups are noted. 

\section{Methodology}
\label{sec:meth}

We consider a time-varying GGM defined on a set of $T$ equidistant discrete time points indexed by $\{ t_1,\cdots, t_T \} $: $(x_{1} (t_k),\cdots,x_p (t_k))^{\top}\sim N(\bm{\mu}(t_k),\bm{\Sigma}(t_k) ),\ k\in\{ 1,\cdots,T\}$. Without loss of generality, we assume $\bm{\mu}(t_k)=\bm{0}$, which can be achieved in practice by centering the data set at each time point. 
The notation in Section \ref{sec:intro} is inherited at every discrete time point. Then we have $T$ temporal datasets $\bm{X}(t_1),\cdots,\bm{X}(t_T)$ and diagonals in temporal precision matrices $\bm{\sigma}^{\top}(t_1),\cdots,\bm{\sigma}^{\top}(t_T)$ by stacking which leads to a vector $\bm{\sigma}$ with length $pT$. 
We have the temporal response vector $\bm{Y}(t_k)$ and temporal predictor matrix $\tilde{\bm{X}}(t_k)$ formed as in Section \ref{sec:intro}. A vector $\mathbfcal{Y}$ of length $npT$ is formed by stacking all temporal reponse vector $\bm{Y}(t_k)$. Let $\mathbfcal{X}$ denote a $Tnp\times Tp(p-1)/2$ block diagonal matrix, where each diagonal block is the temporal predictor matrix $\tilde{\bm{X}}(t_k), \ k\in\{ 1,\cdots,T\} $. 
Our objective is to estimate a vector $\bm{\theta}$ of length $Tp(p-1)/2$ composed of all temporal partial correlations, i.e., $\bm{\theta}=(\bm{\theta}^{\top}(t_1),\cdots,\bm{\theta}^{\top}(t_T))^{\top}$. Throughout the rest of the paper, we use $\bm{\theta}$ and $\bm{\sigma}$ to denote these two long vectors containing temporal parameters over all time points.

Na\"{i}vely, one can minimize \eqref{standardLASSO} at each time point independently to estimate the temporal partial correlations, and this baseline method is referred to as \textit{LASSO in the time-varying network}, where the loss function can be written in the matrix form as:
\begin{align}
\mathcal{L}_{TVN\_LASSO}( \mathbfcal{Y},\mathbfcal{X} ,\bm{\theta},n,\lambda_1) =
\sum\limits_{k=1}^{T} L_{LASSO}(\bm{Y}(t_k) ,\tilde{\bm{X}}(t_k) ,\bm{\theta},n,\lambda_1)=\frac{1}{n}  \| \mathbfcal{Y}-\mathbfcal{X}\bm{\theta}\| ^2  + \lambda_1 \|\bm{\theta} \|_1 .
\label{TVlasso}
\end{align}
But it is natural to assume that the covariance matrix $\bm{\Sigma}(t)$ are element-wise smooth over $t$. Then, by Cramer's rule, the entries in the precision matrix and thus the corresponding partial correlations should also be smooth over $t$. Therefore, we propose a regularization method with an extra penalty term $\lambda_2 \cdot P (\bm{\theta})$ to encourage the partial correlations to be similar for neighbouring time points. The objective function $\mathcal{L}_{TVN}$ of our time-varying network problem is
\begin{align}
\mathcal{L}_{TVN}( \mathbfcal{Y},\mathbfcal{X} ,\bm{\theta},n,\lambda_1,\lambda_2) =\mathcal{L}_{TVN\_LASSO}(\mathbfcal{Y},\mathbfcal{X} ,\bm{\theta},n,\lambda_1) + \lambda_2 \cdot P (\bm{\theta}).  \label{tvnet}
\end{align}
Here $P(\bm{\theta})$ denotes a penalty function measuring the total distance between neighbouring coefficients and $\lambda_2$ is another tuning parameter. We consider two  different penalty functions for $P(\bm{\theta})$. 
\paragraph{Generalized elastic net (GEN)}
Our first penalty generalizes the work of  \cite{zou2005regularization}. To achieve smoothness of partial correlations over time, the GEN applies $l_2$ penalties to the differences of partial correlations along the time sequences in $P (\bm{\theta})$, taking the form
\begin{align}
P(\bm{\theta}) = \sum_{k=2}^{T} \sum_{1\leq i<j\leq p } [ \rho_{ij} (t_k)-\rho_{ij} (t_{k-1}) ]^2 .  \label{gen1}
\end{align}
\paragraph{Generalized fused lasso (GFL)}
Our second penalty generalizes the work of \cite{Tibshirani2005fuse} by penalizing the absolute difference of the partial correlations at adjacent time points. In particular, the penalty function takes the following form:
 \begin{align}
P(\bm{\theta}) = \sum_{k=2}^{T} \sum_{1\leq i<j\leq p } | \rho_{ij} (t_k)-\rho_{ij} (t_{k-1})  | .  \label{gfl1}
\end{align}
A large tuning parameter $\lambda_2$ in the GFL not only yields smoothness in the changes between neighboring coefficients, but also shrinks some of those changes to be exactly zero.

When $\lambda_2=0$, both the GEN and the GFL are reduced to the na\"{i}ve LASSO problem \eqref{TVlasso}, where there is no regularization on the changes in coefficients at adjacent time points. The key difference between the two penalties is that the GFL is able to force the partial correlations at adjacent time points to be identical if their difference is sufficiently small, while GEN cannot.

\section{Algorithm for the Time-varying Network Estimation}
 \label{sec:algo}

In this section, we describe some high-level computational details, while specific technicalities are described in the supplementary materials. 
 
The loss function \eqref{tvnet} with  GEN \eqref{gen1} and GFL \eqref{gfl1} penalties can be respectively written as:
\begin{align}
\mathcal{L}_{GEN}(\mathbfcal{Y}, \mathbfcal{X} , \bm{\theta}, n,\lambda_1,\lambda_2) & =\frac{1}{n}  \| \mathbfcal{Y}-\mathbfcal{X}\bm{\theta}\| ^2  + \lambda_1 \|\bm{\theta}\|_1 + \lambda_2 \| \bm{D} \bm{\theta} \|^2   \label{gen2} \quad \text{and} \\
\mathcal{L}_{GFL}(\mathbfcal{Y}, \mathbfcal{X} , \bm{\theta}, n,\lambda_1,\lambda_2) & =\frac{1}{n}  \| \mathbfcal{Y}-\mathbfcal{X}\bm{\theta}\| ^2  + \lambda_1 \|\bm{\theta} \|_1 + \lambda_2 \| \bm{D} \bm{\theta} \|_1 ,  \label{gfl2}
\end{align}
where $\bm{D}$ is a block difference matrix composed of many $p(p-1)/2-$by$-p(p-1)/2$ square matrices:
\begin{align}
\bm{D}_{(T-1)p(p-1)/2 \times T p(p-1)/2}=
\begin{bmatrix}
I & -I & 0 & 0 & 0  \\
0 & I & -I & 0 & 0   \\
 &  & \ddots  & \ddots &    \\
0 & 0 & 0 & I & -I \\
\end{bmatrix}.\label{Dmat}
\end{align} 
If $\mathbfcal{X}$ is completely known, the minimizations of \eqref{gen2} and \eqref{gfl2} over $\bm{\theta}$ are standard GEN and GFL optimization problems.
However, the predictor matrix $\mathbfcal{X}$ involves an unknown parameter $\bm{\sigma}$, which requires us to leverage coordinate descent techniques and update $\bm{\sigma}$ and $\bm{\theta}$ iteratively. We extend the two-step iterative procedure developed by \cite{peng2009partial}. 
The detailed algorithm is summarized in Algorithm~\ref{alg:space}. Step 2 and 3 are iterated to update $\bm{\sigma}$ and $\bm{\theta}$ until convergence. Given $\bm{\sigma}$, both minimization problems---of \eqref{gen2} and of \eqref{gfl2}---are convex; details are given in Section \ref{subsec:ADMM}. 
\begin{algorithm}[ht]
\caption{Two-step iterative procedure}  
\label{alg:space}
    \begin{algorithmic}[1]
        \Statex \textbf{Input:} The centered original data $\{ \bm{X} \}$.
        \Statex \textbf{Output}: Estimated $ \bm{\sigma} $ and $ \bm{\theta} $
        \Statex  \textbf{Initialization}:
        \Statex Start with the initial estimate $(\sigma^{ii})^{(0)}(t_k) =1/\hat{\sigma}_{ii}(t_k)$, where $\hat{\sigma}_{ii} (t_k) = (n-1)^{-1} \sum_{j=1}^n [x_i^j (t_k) -\bar{x}_i (t_k)]^2.$  
        Form the initial predictor matrix $\mathbfcal{X}^{(0)}$ with the data and $\bm{\sigma}^{(0)}$.
        \While {$\| \bm{\sigma}^{(l)} - \bm{\sigma}^{(l-1)} \|_2 > tol_1 $ and $ \| \bm{\theta}^{(l)} - \bm{\theta}^{(l-1)} \|_2 >tol_2 $}
        
        \State Estimate $\bm{\theta}^{(l+1)}$ by solving \eqref{gen2} or \eqref{gfl2} with the given $\bm{\sigma}^{(l)} $.
        \State Update $\bm{\sigma}^{(l+1)}$, where $1/\hat{\sigma}^{ii} (t_k)= n^{-1} \| X_i  (t_k) - \sum_{j \neq i} \hat{\beta}_{ij} (t_k) X_j (t_k) \|^2 $ and $\hat{\beta}_{ij}(t_k) = (\rho_{ij})^{(l+1)}(t_k)\sqrt{(\sigma^{jj})^{(l)} (t_k)/ (\sigma^{ii})^{(l)} (t_k)}$.
        \State Update the predictor matrix $\mathbfcal{X}^{(l+1)}$ with the data and $ \bm{\sigma}^{(l+1)} $.
          \EndWhile
    \end{algorithmic}
\end{algorithm}

\subsection{Fast ADMM algorithms for GEN and GFL}
\label{subsec:ADMM}
Given $\bm{\sigma}$, 
we minimize \eqref{gen2} and \eqref{gfl2} using the ADMM. The ADMM algorithm and its convergence properties are illustrated in detail in \cite{boyd2011distributed}. The key trick to use the technique in our context is that, by adding a new constraint $\bm{\theta}-\bm{z}=\bm{0},$ we can freely re-express our objective functions \eqref{gen2} and \eqref{gfl2} in either of the following ways, 
\begin{equation*}
\mathcal{L}(\bm{\theta,\bm{z}}) = 
\begin{cases}
\frac{1}{n}  \| \mathbfcal{Y} -    \mathbfcal{X}  \bm{\theta} \| ^2 + \lambda_1 \| \bm{z} \|_1 + \lambda_2 P( \bm{\theta} ),\\[1mm]
\frac{1}{n}  \| \mathbfcal{Y} -    \mathbfcal{X}  \bm{\theta} \| ^2 + \lambda_1 \| \bm{z} \|_1 + \lambda_2 P( \bm{z} ), 
\end{cases}
\end{equation*}
depending on the specific form of the penalty function $P(\cdot)$.

Define $\bm{u} \in \mathbb{R}^{Tp(p-1)/2 \times 1}$ as the dual variable; and let $a \in \mathbb{R}^+$ be a penalty parameter. The augmented Lagrangian $L_a$ for minimizing $\mathcal{L}(\bm{\theta},\bm{z})$, subject to $\bm{\theta}-\bm{z}=\bm{0}$, is
 \begin{align}
 L_a (\bm{\theta} ,\bm{z} ,\bm{u} ) = \mathcal{L}(\bm{\theta},\bm{z}) + a \cdot \bm{u}^{\top} (\bm{\theta}-\bm{z})+\frac{a}{2} \| \bm{\theta}-\bm{z} \|^2, \label{LagGEN}
\end{align}
where we have scaled the dual variable $\bm{u}$ by $a$ itself, so that
the ADMM algorithm iterates over the following three steps:
\begin{flalign}
\text{(i) } & \bm{\theta}_{(l)}=\arg\min\limits_{\bm{\theta}} L_a (\bm{\theta} ,\bm{z}_{(l-1)} ,\bm{u}_{(l-1)} ) \label{rhoADMM}\\
\text{(ii) } & \bm{z}_{(l)}=\arg\min\limits_{\bm{z}} L_a (\bm{\theta}_{(l)} ,\bm{z} ,\bm{u}_{(l-1)} )\label{zADMM} \\
\text{(iii) } & \bm{u}_{(l)}=\bm{u}_{(l-1)} + {\bm{\theta}_{(l)} - \bm{z}_{(l)}} \label{uADMM}
\end{flalign}
over $l=0,1,2,\dots$ until convergence, with typical initialization $\bm{\theta}_{(0)}=\bm{0}$, $\bm{z}_{(0)}=\bm{0}$ and $\bm{u}_{(0)}=\bm{0}$. 

The separation of the underlying optimization problem into two subproblems---namely, \eqref{rhoADMM} and \eqref{zADMM}---allows us to obtain the key insight that, for the GEN penality $P(\bm{\theta})=\|\bm{D}\bm{\theta}\|^2$, it is more advantageous to parameterize the penalty as $P(\bm{\theta})$; whereas, for the GFL penalty $P(\bm{\theta})=\|\bm{D}\bm{\theta}\|_1$, it is more advantageous to parameterize it as $P(\bm{z})$. More details are presented below, where, for clarity, we shall suppress the step index $l$ in all formulae. 


\subsubsection{ADMM for GEN}
\label{subsubsec:admm_gen}

As stated above, for the GEN penalty we parameterize it as $P(\bm{\theta})$ in the ADMM iterations, so that \eqref{rhoADMM} merely minimizes over a quadratic function of $\bm{\theta}$,
\begin{align}
\bm{\theta} & =\arg\min\limits_{\bm{\theta}} \frac{1}{n}  \| \mathbfcal{Y} -    \mathbfcal{X}  \bm{\theta} \| ^2 + \lambda_2 \| \bm{D} \bm{\theta} \|^2  + a \cdot \bm{u}^{\top} (\bm{\theta}-\bm{z})+\frac{a}{2} \| \bm{\theta}-\bm{z} \|^2, \label{rhoGEN}
\end{align}
which has closed form solution,
\begin{align}
\bm{\theta} =\left( \dfrac{2}{n}  \mathbfcal{X}^{\top}  \mathbfcal{X} + 2\lambda_2 \bm{D}^{\top} \bm{D} + a I \right)^{-1} \left( \frac{2}{n} \mathbfcal{Y}^{\top} \mathbfcal{X} + a (\bm{z}- \bm{u} ) \right). \label{rhoClose}
\end{align}
While this may appear easy, it is worth emphasizing that, for us, the matrix that must be inverted in \eqref{rhoClose} can be very large. Fortunately, the inversion can be pre-calculated outside the ADMM iterations, for  it remains constant from one iteration to another. Furthermore, \eqref{rhoClose} can be computed efficiently by exploiting the fact that $\left( 2 \mathbfcal{X}^{\top}  \mathbfcal{X} /n + 2\lambda_2 \bm{D}^{\top} \bm{D} + a I \right)$ is a symmetric block tri-diagonal matrix whose off-diagonal blocks are $-2\lambda_2 I$. Technical details for efficiently inverting such a matrix are given in \ref{subsec:supplement1} in the supplementary materials.

The minimization \eqref{zADMM} over $\bm{z}$,
\begin{align*}
\bm{z} & =\arg\min\limits_{\bm{z}} \lambda_1 \| \bm{z}  \|_1 + a \cdot \bm{u}^{\top} (\bm{\theta}-\bm{z})+\frac{a}{2} \| \bm{\theta}-\bm{z} \|^2, 
\end{align*}
is simply a LASSO-type problem. 
As $ \| \bm{z}  \|_1$ is not differentiable everywhere, we leverage its sub-differential and obtain the solution as
\begin{equation}\label{GEN_zupdate}
\bm{z} =
\begin{cases}  
      \bm{u}+\bm{\theta}-\dfrac{\lambda_1}{a}, & \text{if $\bm{u}+\bm{\theta}>\dfrac{\lambda_1}{a}$}, \\[2mm]
      \bm{u}+\bm{\theta}+\dfrac{\lambda_1}{a}, & \text{if $\bm{u}+\bm{\theta}<-\dfrac{\lambda_1}{a}$}, \\
      \bm{0}, & \text{otherwise}.
   \end{cases}
\end{equation}

\subsubsection{ADMM for GFL}
\label{subsubsec:admm_gfl}
The GFL problem is in itself important for a wide range of scientific procedures including signal processing and machine learning, especially when the matrix $\bm{D}$ in \eqref{gfl2} takes on more general forms. Even though \eqref{gfl2} is convex and there exists a global optimal solution, minimizing it is still computationally challenging. A large body of literature exists on solving the GFL problem \citep[e.g.,][]{tibshirani2011solution,ye2011split,xin2016efficient}, but many methods still suffer from high computational cost or have difficulties with achieving sparsity in both $\bm{\theta}$ and $\bm{D}\bm{\theta}$ simultaneously. To get around these bottlenecks, we design a specific ADMM algorithm by exploiting the special block structure in our problem. 

Again, as stated earlier, for the GFL penalty we parameterize it as $P(\bm{z})$ in the ADMM iterations, so \eqref{rhoADMM} still merely minimizes over a quadratic function of $\bm{\theta}$,
\begin{align}
\bm{\theta} & =\arg\min\limits_{\bm{\theta}} \frac{1}{n}  \| \mathbfcal{Y} -  \mathbfcal{X}  \bm{\theta} \| ^2  + a \cdot \bm{u}^{\top} (\bm{\theta}-\bm{z})+\frac{a}{2} \| \bm{\theta}-\bm{z} \|^2, \nonumber
\end{align}
with closed-form solution,
\begin{align}
\bm{\theta} =\left( \dfrac{2}{n}  \mathbfcal{X}^{\top}  \mathbfcal{X}  + a I \right)^{-1} \left( \frac{2}{n} \mathbfcal{Y}^{\top} \mathbfcal{X} + a (\bm{z}- \bm{u} ) \right). \label{rhoupdateGFL2}
\end{align}

The minimization \eqref{zADMM} over $\bm{z}$ now becomes
\begin{align}
\bm{z} & =\underset{\bm{z}}{\arg\min}  \ \lambda_1 \| \bm{z}  \|_1 + \lambda_2 \| \bm{D} \bm{z } \|_1 + a \cdot \bm{u}^{\top} (\bm{\theta}-\bm{z})+\frac{a}{2} \| \bm{\theta}-\bm{z} \|^2. \nonumber \\
& = \underset{\bm{z}}{\arg\min} \  \frac{1}{2} \| \bm{\theta}+\bm{u}-\bm{z} \|^2 +\dfrac{ \lambda_1}{a} \| \bm{z}  \|_1 + \dfrac{\lambda_2}{a} \| \bm{D} \bm{z } \|_1 . \label{zupdateGFL}
\end{align}
Let $\bm{z}(t_k)=(z_{12}(t_k),\cdots,z_{p-1,p}(t_k))^{\top}$, and $\bm{u}$ be defined in a similar way. Then \eqref{zupdateGFL} can be decomposed into $p(p-1)/2$ independent optimization problems,
\begin{multline*}
 \{ z_{ij} (t_k)\}_{k=1}^{T}   = 
 \underset{\{ z_{ij}(t_k)\}_{k=1}^{T}}{\arg\min} \ \sum\limits_{k=1}^{T} \left[ z_{ij} (t_k) -\rho_{ij}(t_k) -u_{ij}(t_k) \right ]^2 + \\
 \dfrac{\lambda_1}{a} \sum\limits_{k=1}^{T} | z_{ij} (t_k) |  + \dfrac{\lambda_2}{a} \sum\limits_{k=2}^{T} | z_{ij} (t_k)- z_{ij} (t_{k-1}) |,
 \quad 1\leq i < j \leq p,
\end{multline*}
a collection of fused LASSO signal approximator (FLSA) problems, and we use the algorithm in \cite{hoefling2010path} to solve them.


\subsection{Tuning parameter selection}
\label{sec:ModelSelection}

In this section, we discuss how to choose the tuning parameters $(\lambda_1,\lambda_2)$.
We adopt the Bayesian Information Criterion (BIC), 
\begin{equation}
BIC(\lambda_1,\lambda_2) = n \times \sum\limits_{k=1}^{T} \left[ -\log|\hat{\bm{\Sigma}}^{-1} (t_k) | + \text{tr}\left( \hat{\bm{\Sigma}}^{-1} (t_k) \cdot \bm{S} (t_k) \right) \right] + \log (n) \times \hat{df}(\lambda_1,\lambda_2), \label{bic}
\end{equation}
where $\hat{\bm{\Sigma}}^{-1}$ is the estimated precision matrix based on the estimated partial correlations, and $\bm{S}$ is the sample covariance matrix. For the degree of freedom $\hat{df}(\lambda_1,\lambda_2)$ in \eqref{bic}, we use existing results in the literature to derive specific formulae for both \eqref{gen2} and \eqref{gfl2}. 

\cite{zou2007degrees} derived an explicit  degree-of-freedom formula for the LASSO. Their approach can easily be adopted to derive the degree of freedom for \eqref{gen2}---by combing the $l_2$ penalty $\|\bm{D}\bm{\theta}\|^2$ with the least-squares objective $\|\mathbfcal{Y}-\mathbfcal{X}\bm{\theta}\|^2$; see \ref{subsec:supplement2} in the supplementary materials. We obtain
\begin{align}
\hat{df}_{GEN} & = \text{tr}  \left[  \left(  \mathbfcal{X}_{\mathcal{A}}^{\top} \mathbfcal{X}_{\mathcal{A}}  + n\lambda_2 \bm{D}_{\mathcal{A}}^{\top} \bm{D}_{\mathcal{A}}   \right)^{-1} \mathbfcal{X}_{\mathcal{A}}^{\top} \mathbfcal{X}_{\mathcal{A}} \right] \label{GENdfH} \\
&\approx \text{tr}  \left[  \left(  I + n\lambda_2 \left( \mathbfcal{X}_{\mathcal{A}}^{\top} \mathbfcal{X}_{\mathcal{A}} +\eta I \right)^{-1}  \bm{D}_{\mathcal{A}}^{\top} \bm{D}_{\mathcal{A}}   \right)^{-1}  \right] ,
\label{GENdfSimplify}
\end{align}
where $\mathcal{A}=\{ i: \hat{\bm{\theta}}_i \neq 0 \} $ denotes the active set, and $\mathbfcal{X}_{\mathcal{A}}$ (or $\bm{D}_{\mathcal{A}}$) denotes the corresponding submatrix containing only the columns indexed by $\mathcal{A}$. Note that the matrix $\mathbfcal{D}$, having fewer rows than columns, is not full-rank, and that,
for large $p$, the matrix $\mathbfcal{X}$ is often not full-rank, either. Therefore, to compute \eqref{GENdfH} for any $\mathcal{A}$, it is necessary to first add a small perturbation matrix $\eta \bm{I}$ to $\mathbfcal{X}_{\mathcal{A}}^{\top} \mathbfcal{X}_{\mathcal{A}}$---we set $\eta=10^{-5}$. The final step \eqref{GENdfSimplify} is due to the identity $(\bm{A}+\bm{B})^{-1}=(I+\bm{A}^{-1}\bm{B})^{-1}\bm{A}^{-1}$; it has the additional advantage over \eqref{GENdfH} that the trace of such a matrix inverse can be approximated by Chebyshev interpolation  \citep{han2017approximating}. 


\cite{tibshirani2012degrees} worked out how to compute the degree of freedom for a generalized LASSO problem, into which \eqref{gfl2} can be transformed---see \ref{subsec:supplement3} in the supplementary materials. Applying their result, we conclude that the degree of freedom for the GFL problem \eqref{gfl2} is equal to dimension of the null space of $[\bm{D}^{\top},\bm{I}]^{\top}_{-\mathcal{A}}$, where $ \mathcal{A}=\{ i: [\bm{D}^{\top},\bm{I}]^{\top} \hat{\bm{\theta}}_i \neq 0 \} $. It turns out this somewhat abstract conclusion can be further characterized (again, see \ref{subsec:supplement3}) by something more interpretable---namely,
\begin{align}
\hat{df}_{GFL}=  \sum\limits_{1\leq i < j \leq p}  \left(   \mathbbm{1}\{\hat{\rho}_{ij}(1)\neq 0\} + \sum\limits_{k=2}^{\top} \mathbbm{1}\{ \hat{\rho}_{ij}(k)\neq \hat{\rho}_{ij}(k-1),\ \hat{\rho}_{ij}(k)\neq0\} \right), \label{GFLdf}
\end{align}
where $\mathbbm{1}(\cdot)$ denotes a binary indicator function. That is, the degree of freedom here is simply the total number of nonzero fused groups over all $\hat{\rho}_{ij}$.

\paragraph{Remark} Strictly speaking, the degree-of-freedom formulae derived by  \cite{zou2007degrees} and \cite{tibshirani2012degrees} both require the regression of $\mathbfcal{Y}$ onto $\mathbfcal{X}$ to be homoscedastic, which is not the case for us, but we apply their results nonetheless because deriving similar results without the homescedastic assumption is currently an unsolved problem on its own. Hence, our degree-of-freedom formulae \eqref{GENdfSimplify} and \eqref{GFLdf} are necessarily ad-hoc approximations, but they are still useful in facilitating the choice of tuning parameters through the BIC, as our empirical results below will demonstrate.  

\section{Simulation}
\label{sec:simulation}

In this section, we perform simulation studies to assess the performance of the two proposed approaches, GEN and GFL, and compare them with the results from na\"{i}ve LASSO and sample estimates. The first objective is to evaluate the ability of our methods in uncovering the underlying networks of the synthetic data, and the second objective is to investigate the accuracy of the estimated partial correlations. 

We generate $p$ random functions $x_{i} (t)$, $i=1,\cdots,p$ on $[0,1]$ identically and independently for $n$ subjects. Let $x_{i} (t)=\mu(t)+e_{i} (t)$, where $\mu(t) = t + \sin(t)$. The random vector $\left( e_{1} (t),\cdots,e_p (t) \right)^\top$ is drawn from a centered multivariate Gaussian distribution with covariance $\Sigma (t)$ that gives rise to true partial correlation $\rho_{ij} (t)$. Zero partial correlations indicate absence of the connection, while nonzero $\rho_{ij} (t)$ indicates an edge, between nodes $x_i$ and $x_j$ at the time $t$. The magnitude of $\rho_{ij} (t)$ represents the strength of connectivity. To generate a sparse network, a sparse precision matrix is required at each time point. As we also would like it to change smoothly over time, extra care must be taken to ensure it is positive-definite at all time points as well. We generate two different scenarios (see Sections~\ref{subsec:s1} and \ref{subsec:s2} below) that satisfy these requirements. In both scenarios, $p=10$-dimensional normal random variables are simulated at each of 30 equally spaced time points on $[0,1]$. To investigate the effects of sample size on the performance, we consider $n=50$ and $n=200$ in both scenarios. Each simulation is repeated 100 times.

\subsection{Scenario 1}
\label{subsec:s1}

In Scenario 1, we characterize the random vector $\left( e_{1} (t),\cdots,e_p (t) \right)^\top$ as a linear combination from a set of $S$ uncorrelated $p$-dimensional Gaussian random vectors whose coefficients are smooth functions. More specifically, $\left( e_{1} (t),\cdots,e_p (t) \right)^\top=\sum_{s=1}^{S} B_s (t) \left( \xi_{1,s},\cdots,\xi_{p,s} \right)^\top$, where $\{ B_s (t), 1\leq s \leq S \}$ denote $S$ cubic B-spline basis functions defined on $[0,1]$. For each $s$, $(\xi_{1,s},\cdots,\xi_{p,s})^\top$ follows a centered multivariate Gaussian distribution with covariance matrix $\Sigma_s$. Employing the B-spline method to generate random components makes $e_i (t) =\sum_{s=1}^{S} B_s (t) \xi_{i,s} $ smooth over time. Furthermore, the linear combination of uncorrelated multivariate Gaussian random vectors allows us to construct time-varying graphical structures, as the true precision matrix of $(x_{1} (t),\cdots,x_{p} (t))^\top$ is given by $[\sum_{s=1}^{S} B_s^2(t) \Sigma_s]^{-1}$. Another advantage is the sparse graphic structure brought by the locally compact support property of the B-spline basis functions. Each basis function is non-zero over a small subinterval (see Figure \ref{fig:bspline} in the supplementary materials), so the graphic structure at one time point only involves a small number of B-spline basis functions and the corresponding Gaussian random vectors. Thus we can easily achieve a sparse graphical structure at each time point by carefully choosing sparse $\Sigma_s$'s. We take $S=13$, and the details of $\Sigma_s$'s are provided in \ref{subsec:supplement0}. Under this setting, there are only five true connections: $1-6$, $2-7$, $3-8$, $4-9$, and $5-10$, and the profiles of the corresponding partial correlations are depicted as the red lines in Figure \ref{fig:s1lines}.

\subsection{Scenario 2}
\label{subsec:s2}

In the second scenario, we follow \cite{peng2009partial} to generate temporal precision matrices directly, whose non-vanishing entries are functions of time. We first define the initial precision matrix $\bm{\Omega}(t)$ as a symmetric matrix whose diagonal entries are one. If there ever exists an edge between node $i$ and $j$, then $\bm{\Omega}[i,j] (t)$ takes values of $\pm f(t,T_s,T_e)$ or $\pm g(t,T_s,T_e)$ with probability 1/4, otherwise 0, where $[T_s,T_e]$ is a pre-defined active interval for the connection $i-j$. Let $f(t,T_s,T_e)$ and $g(t,T_s,T_e)$ be zero everywhere except on their active intervals $[T_s,T_e]$ where
$$f(t,T_s,T_e) =\frac{1}{2}\left[0.1+0.8\cdot \sin \left( \frac{t-T_s}{T_e-T_s}\pi \right)\right] \text{ and } g(t,T_s,T_e) =  \frac{1}{2}\left[0.1+0.8\cdot \left( \frac{t-T_s}{T_e-T_s} \right)\right] .$$
We take the average of $\bm{\Omega}$ and its transpose $\bm{\Omega}^{\top}$ as the true precision matrix to assure the symmetry, and then the covariance matrix $\Sigma(t)$ of the random components $\left( e_{1} (t),\cdots,e_p (t) \right)^\top$ is given by
$\bm{\Sigma}(t) [i,j]=\bm{\Omega}(t)^{-1}[i,j]/\sqrt{\bm{\Omega}(t)^{-1}[i,i]\cdot \bm{\Omega}(t)^{-1}[j,j]}.$

In this scenario, there are six true connections, $1-5$, $1-8$, $2-4$, $2-6$, $3-9$ and $7-10$, and the corresponding partial correlations (as functions of $t$) are displayed as the red lines in Figure \ref{fig:s2lines}.

\subsection{Results}
\label{subsec:simulationResult}



We evaluate different methods with two metrics: 
    (i)
    the estimation error $ \sum_{t=1}^{T} [\sum_{1\leq i , j \leq p} ( \hat{\rho}_{ij}(t) -  \rho_{ij}(t) )^2 ]^{1/2} $; and 
    (ii)
    the area under the ROC curve ({AUC}) 
    which, in our context, is equal to the  frequency that $|\hat{\rho}_{ij}(t)| > |\hat{\rho}_{i'j'}(t')|$ over all $(i,j,t)$-$(i',j',t')$ pairings such that 
    $\rho_{ij}(t)\neq 0$ and $\rho_{i'j'}(t')=0$.
While the first metric measures estimation quality, the second is simply an empirical estimate of the conditional probability that, given a truly-existing edge and a non-existing one, the estimated parameters would rank the true edge ahead of the non-existing one; thus, it measures the ability of different methods to detect the underlying network structure. 

Tables \ref{tabSim:difference} and \ref{tabSim:AUC} summarize the estimation errors and the estimated AUCs of the four methods over 100 simulation replicates, respectively. 
We can see that the performances of all methods improve when the sample size is increased from 50 to 200, as expected, and that our proposed methods offer substantial improvements over the sample estimate and na\"{i}ve LASSO. To gain more insights, we also selectively showcase some specific results below, all of which are based on one simulation rather than over 100 repetitions.


\begin{table}[ht] 
\centering
\begin{tabular}{|c cc cc|}
\hline
\multirow{2}{*}{Method} & \multicolumn{2}{c}{Scenario 1} & \multicolumn{2}{c|}{Scenario 2} \\
 & \multicolumn{1}{c}{$n=50$} & \multicolumn{1}{c }{$n=200$} & \multicolumn{1}{c}{$n=50$} & \multicolumn{1}{c|}{$n=200$} \\
\hline 
Sample & 24.73 (0.13) & 11.37 (0.05) & 43.77 (0.11) & 20.36 (0.04) \\ 
LASSO & 19.62 (0.11) & 4.86 (0.04) & 17.02 (0.06) & 10.05 (0.04) \\ 
GEN & 15.04 (0.06) & 6.10 (0.03) & 11.60 (0.06) & 7.62 (0.03) \\  
GFL & 6.78 (0.08) & 4.44 (0.04) & 11.18 (0.08) & 5.35 (0.04) \\ \hline
\end{tabular}
\caption{Mean estimation error across 100 replicates with standard error in parentheses. }
\label{tabSim:difference}
\end{table}

\begin{table}[ht] 
\centering
\begin{tabular}{|c|cc|cc|}
\hline
\multirow{2}{*}{Method} & \multicolumn{2}{c|}{Scenario 1} & \multicolumn{2}{c|}{Scenario 2} \\
 & \multicolumn{1}{c}{$n=50$} & \multicolumn{1}{c|}{$n=200$} & \multicolumn{1}{c}{$n=50$} & \multicolumn{1}{c|}{$n=200$} \\
\hline 
Sample & 0.896 (0.0016) & 0.915 (0.0014) & 0.776 (0.0024) & 0.910 (0.0016) \\ 
LASSO &0.881 (0.0005)  & 0.917 (0.0012) & 0.734 (0.0021) & 0.887 (0.0013) \\ 
GEN & 0.991 (0.0009) & 0.950 (0.0011) & 0.991 (0.0010) & 0.998 (0.0001) \\  
GFL & 0.938 (0.0019) & 0.966 (0.0016) & 0.903 (0.0028) & 0.993 (0.0006) \\ \hline
\end{tabular}
\caption{Mean and standard error of the estimated AUCs over 100 repeated simulations. }
\label{tabSim:AUC}
\end{table}

First, for the set $\{(i,j): \exists \ t \text{ s.t. } \rho_{ij}(t) \neq 0 \}$, Figures \ref{fig:s1lines} and \ref{fig:s2lines} show the estimated profiles $\hat{\rho}_{ij}(t)$ as a function of time from one simulation instance, respectively for Scenarios 1 and 2. Here, we see more clearly the improvement from $n=50$ to $n=200$. Not surprisingly, we also see that GEN produces smooth functions while GFL produces staircase-shaped functions. 


Next, Figure~\ref{fig:s1_gfl_network} displays the GFL-estimated network (with $n=200$) over every other time point in Scenario 1, with false-positive and false-negative edges clearly indicated at these time points as well. We can see that the estimated network structure does not show rapid bursts of change from time to time. The GEN-estimated network for this scenario and all estimated networks in Scenario 2 are displayed in Figures \ref{fig:s1_gen_network}, \ref{fig:s2_gen_network}, and \ref{fig:s2_gfl_network} in the supplementary materials. 

Finally, Figure~\ref{fig:s2_gen_contour} depicts the contours of the BIC, of the estimation error, and of the AUC in Scenario 2, over a grid of $(\lambda_1,\lambda_2)$ with the GEN penalty. It demonstrates that the tuning parameters selected by the BIC indeed lead to good solutions in terms of both metrics. The na\"{i}ve LASSO solution, with its tuning parameter also selected by the BIC, is indicated as well, whereas the sample estimate is, of course, at the origin $(0,0)$. We can see that, while the LASSO solution is clearly better than the sample solution, incorporating the additional GEN penalty provides substantial further improvements.  
Similar contour plots for the GFL penalty in this scenario and those in Scenario 1 are given in Figures~\ref{fig:s2_gfl_contour}, \ref{fig:s1_gen_contour}, and \ref{fig:s1_gfl_contour} in the supplementary materials.


\begin{figure}[htp]
\centering
\begin{subfigure}[b]{\textwidth}
\centering
\includegraphics[angle=270,width=0.3\linewidth]{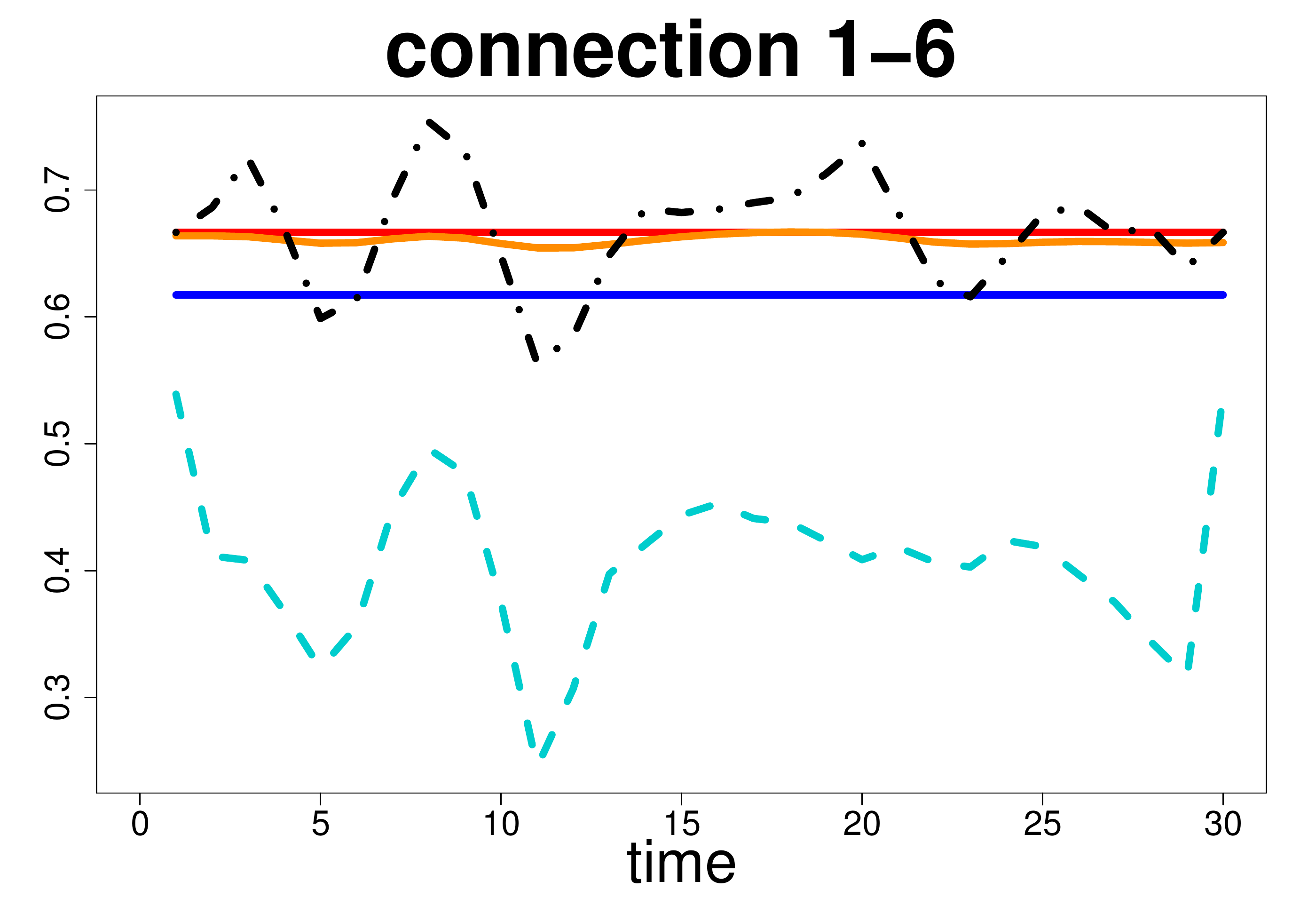}
\includegraphics[angle=270,width=0.3\linewidth]{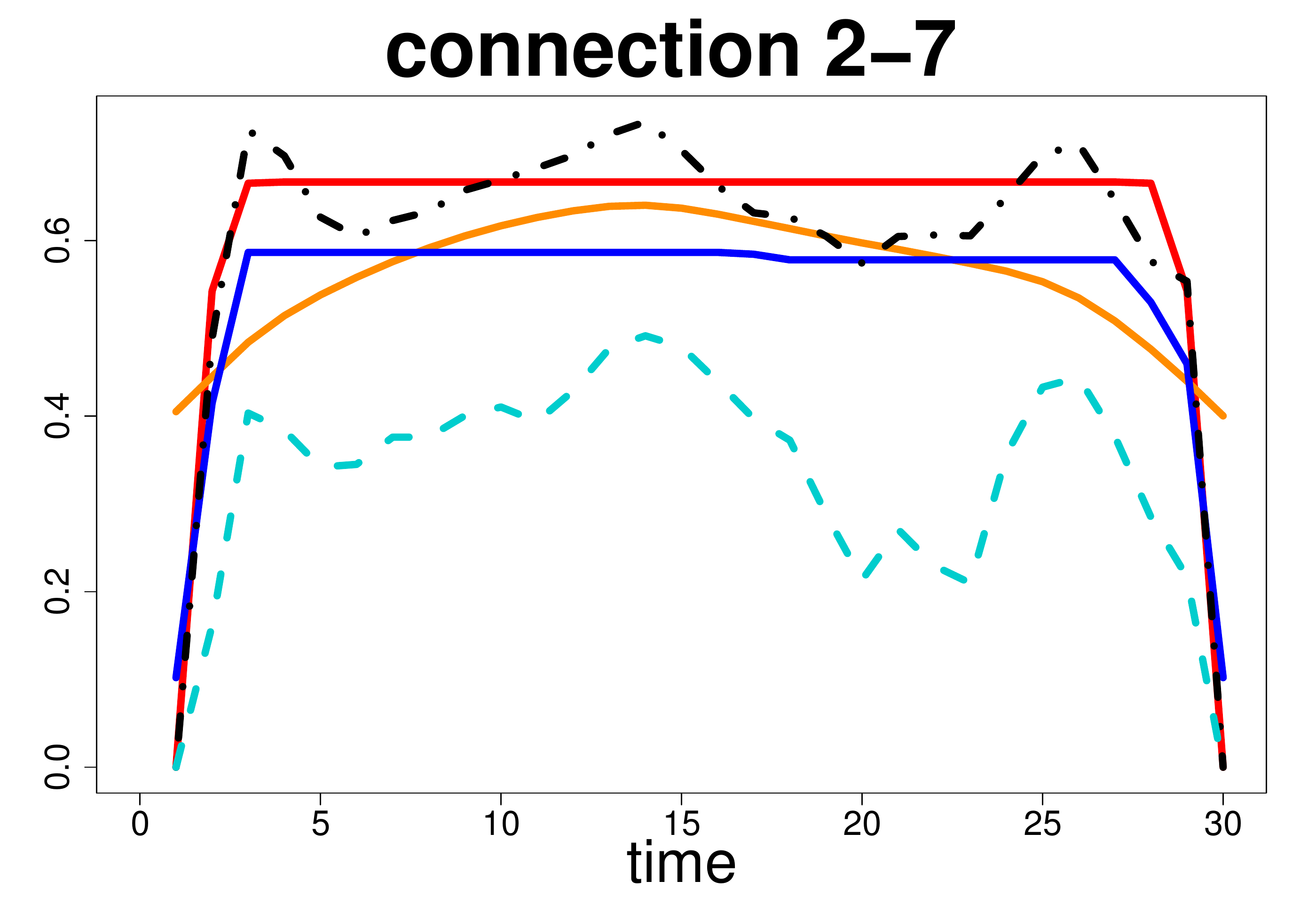}
\includegraphics[angle=270,width=0.3\linewidth]{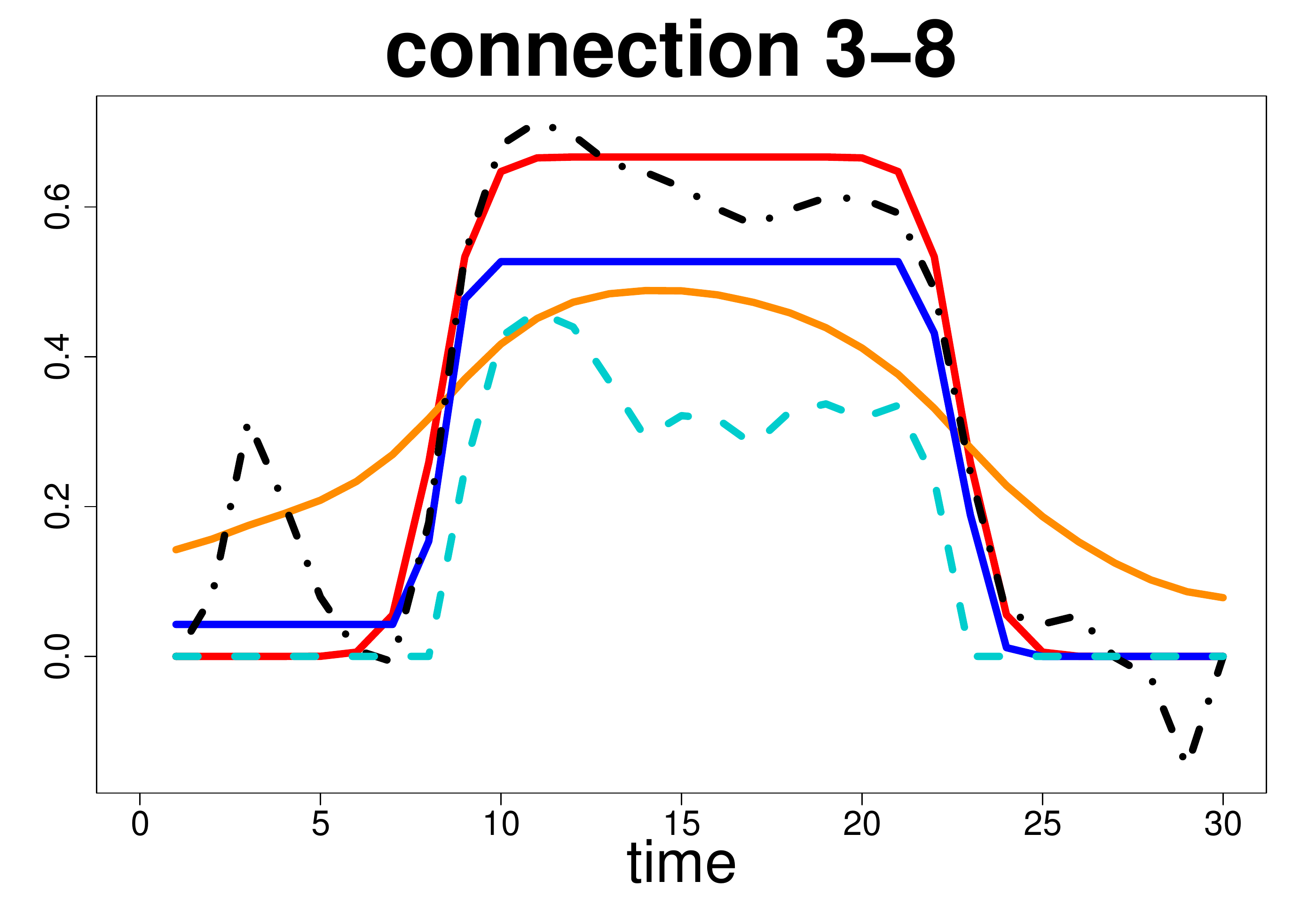}
\includegraphics[angle=270,width=0.3\linewidth]{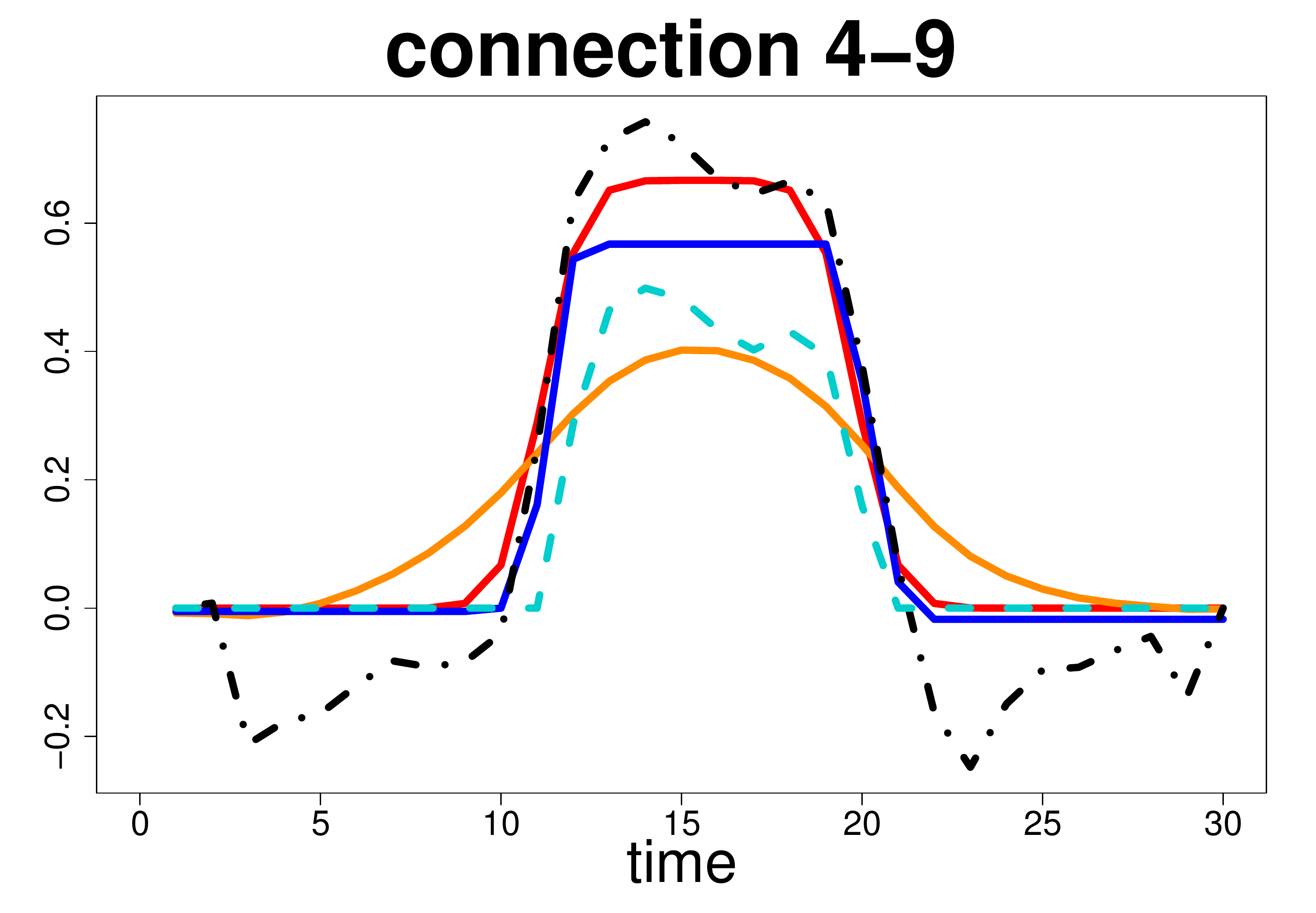}
\includegraphics[angle=270,width=0.3\linewidth]{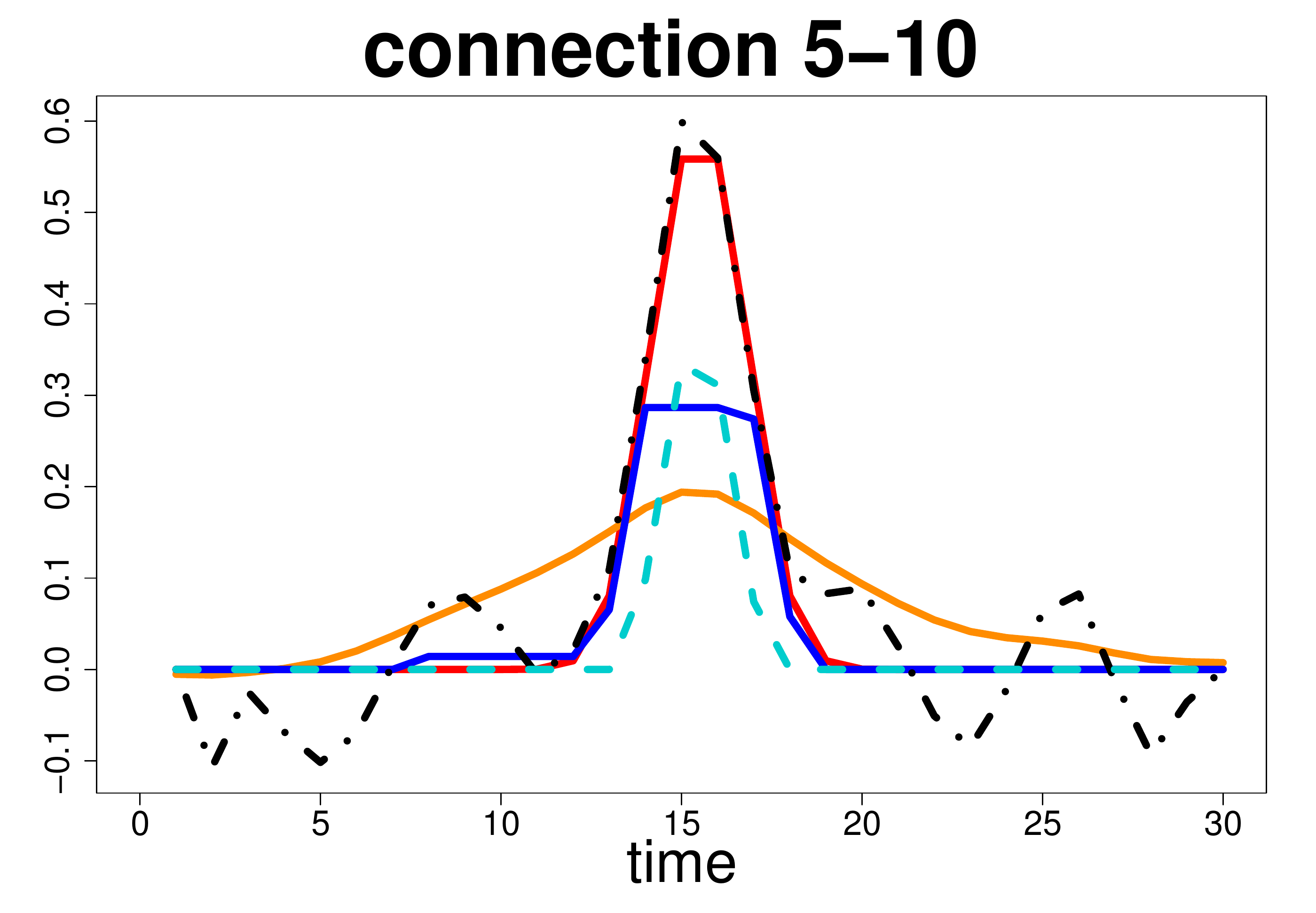}
\includegraphics[angle=270,width=0.3\linewidth]{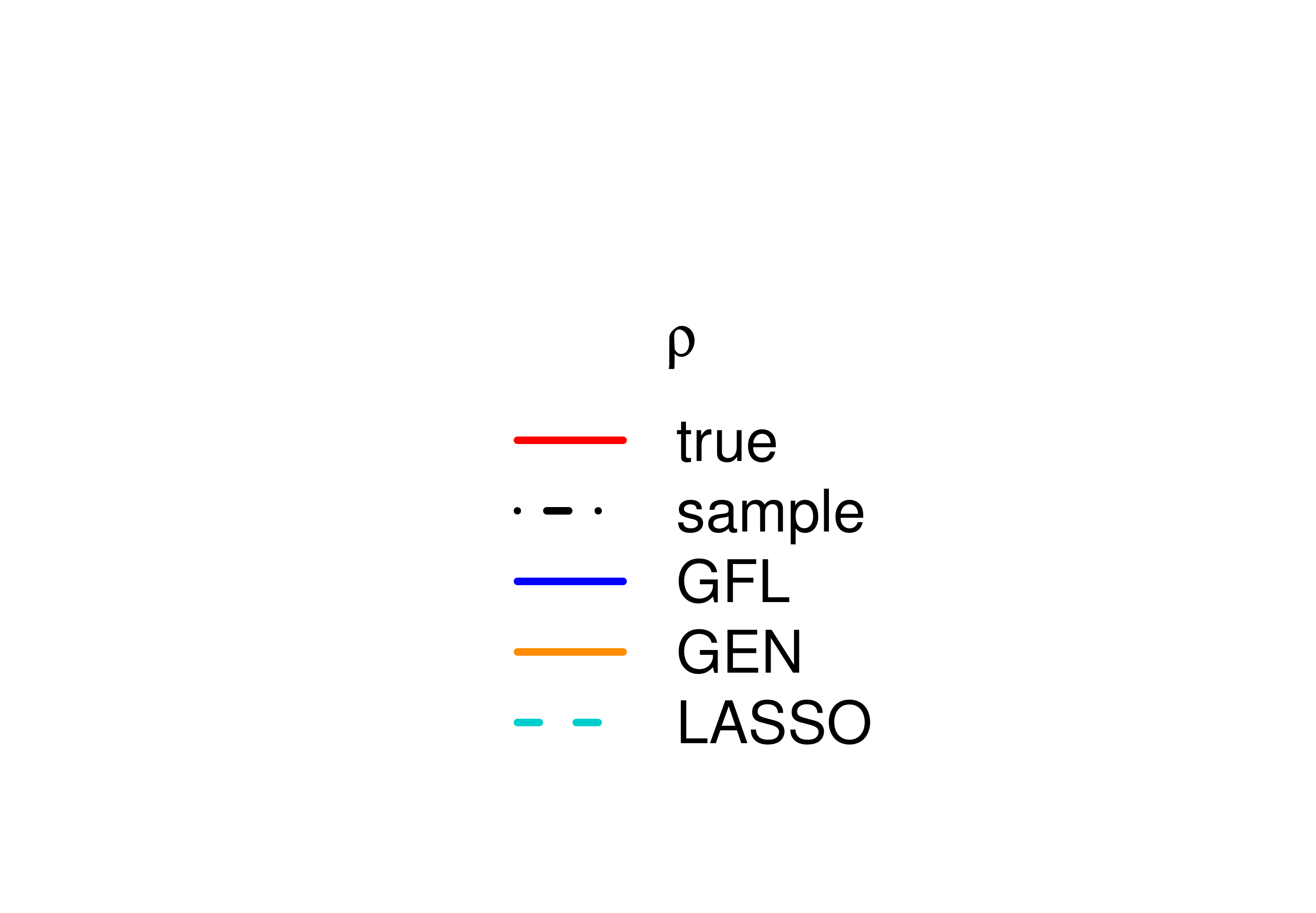}
\caption{Sample size 50}
\end{subfigure}
\begin{subfigure}[b]{\textwidth}
\centering
\includegraphics[angle=270,width=0.3\linewidth]{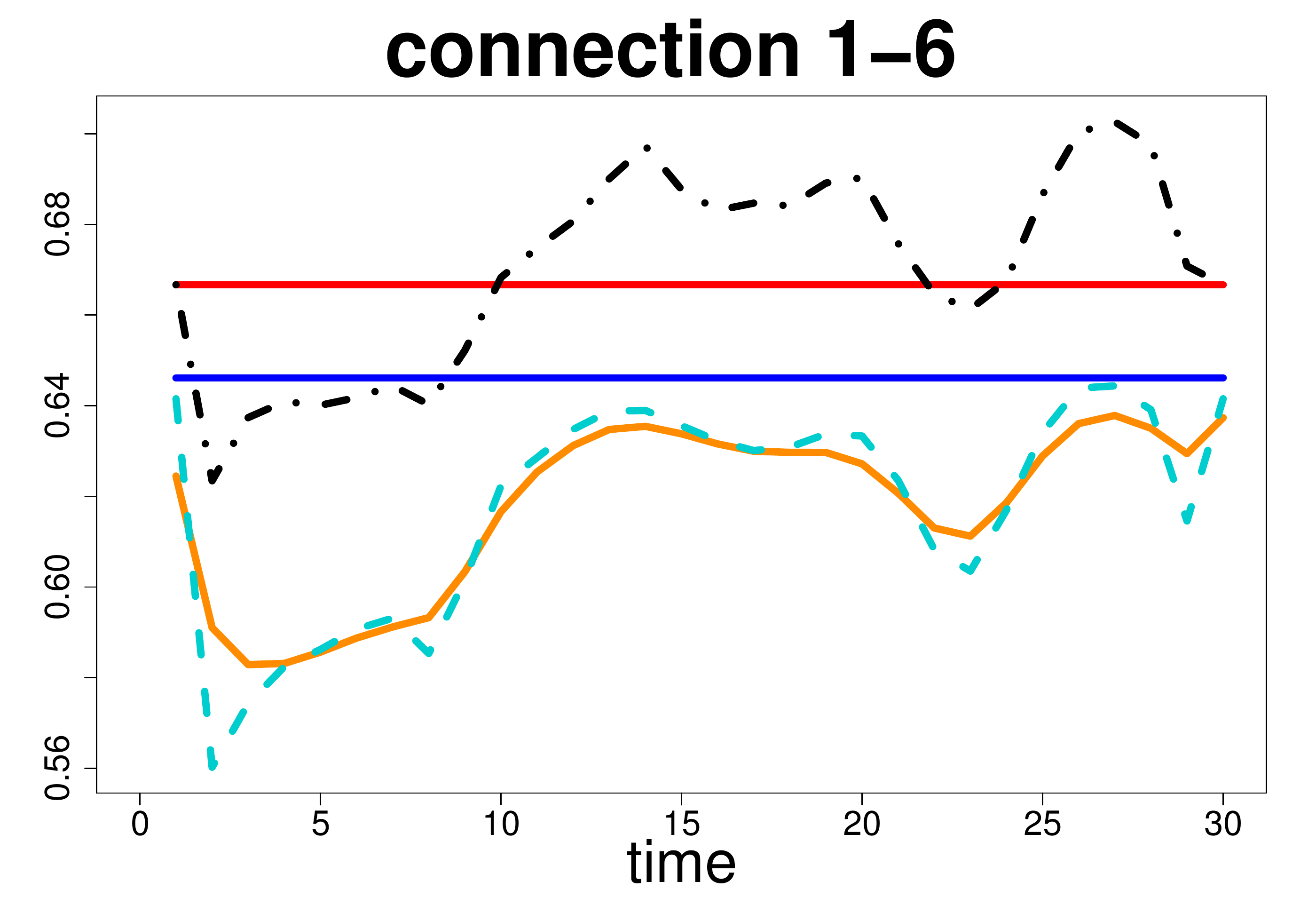}
\includegraphics[angle=270,width=0.3\linewidth]{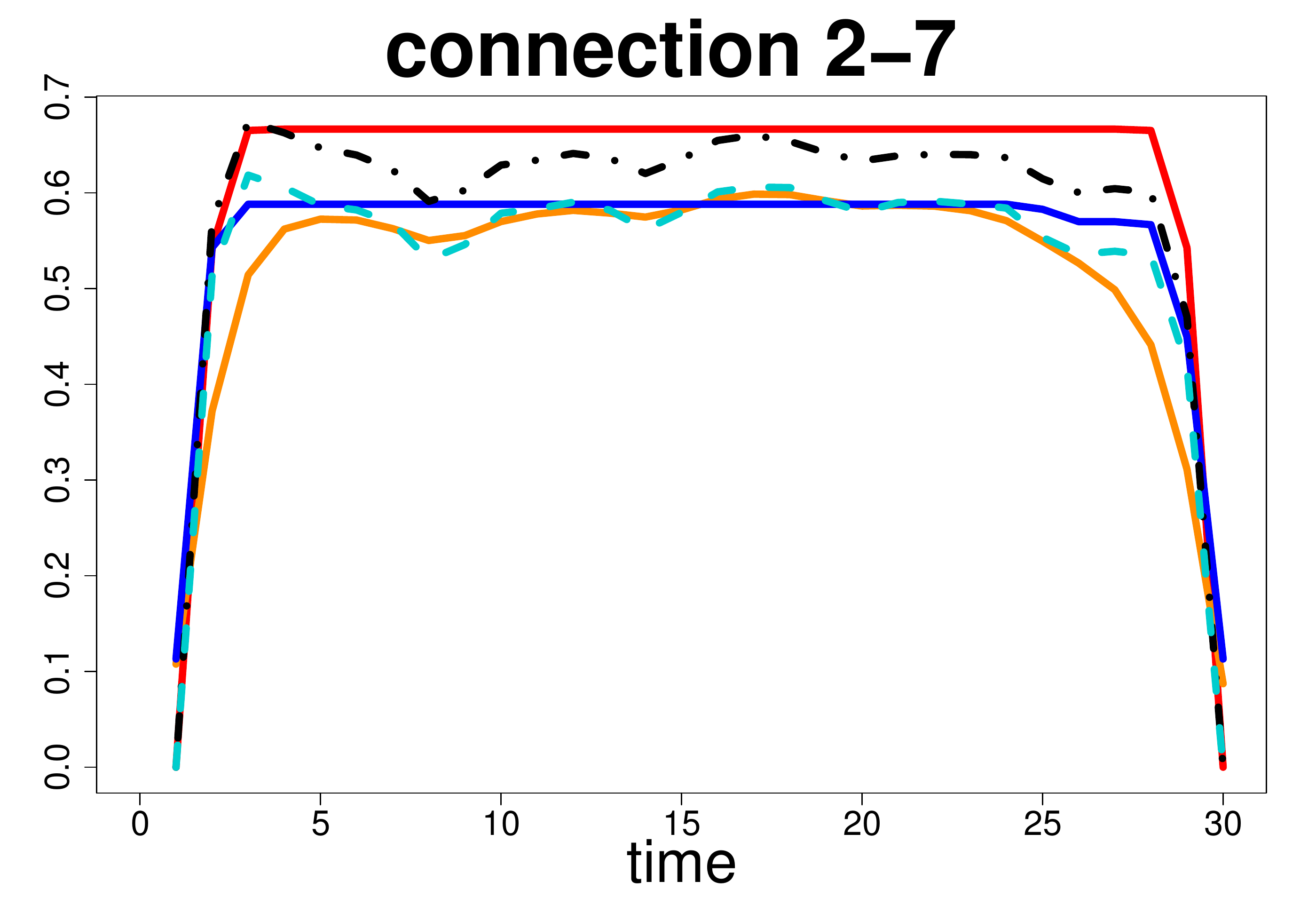}
\includegraphics[angle=270,width=0.3\linewidth]{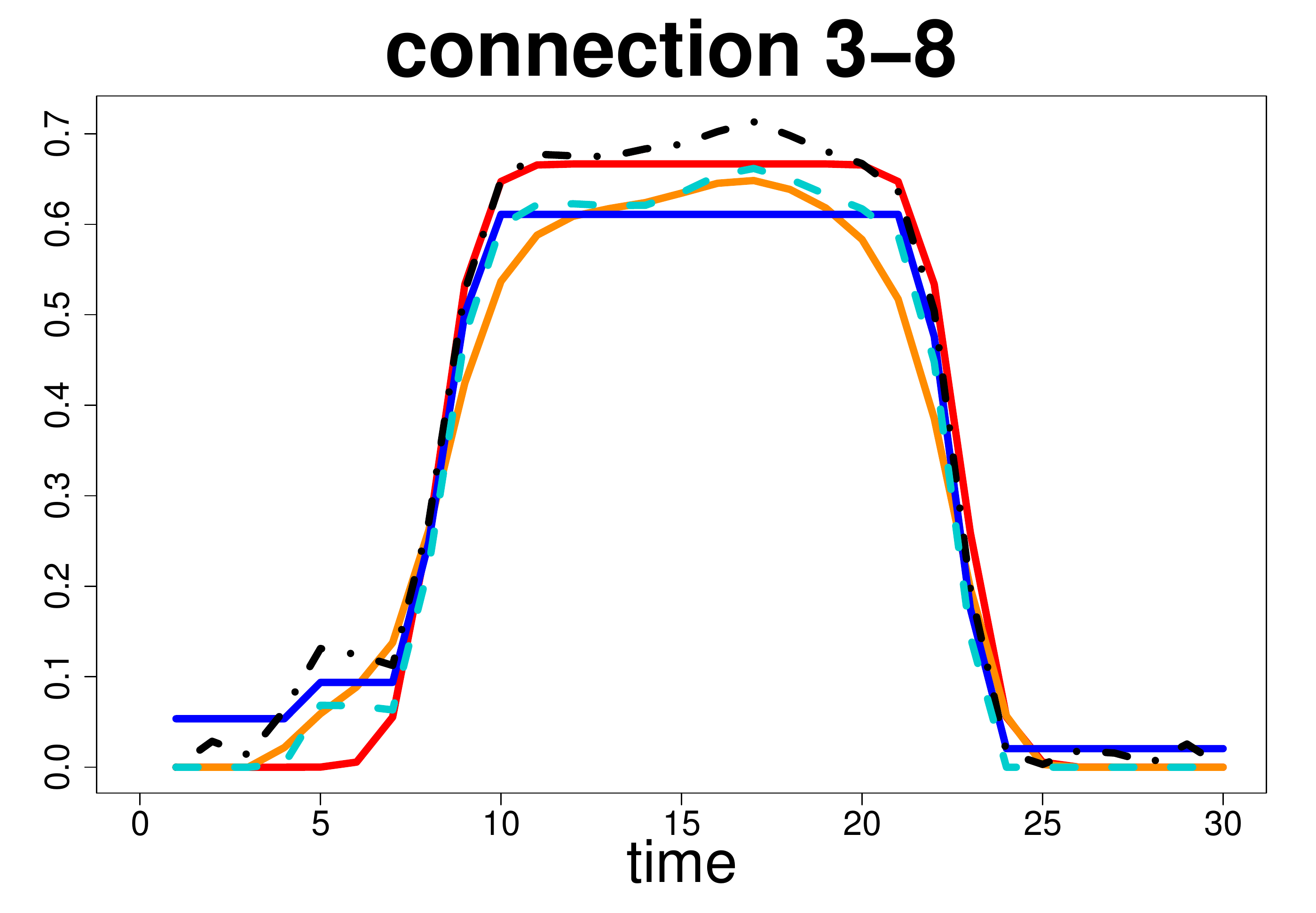}
\includegraphics[angle=270,width=0.3\linewidth]{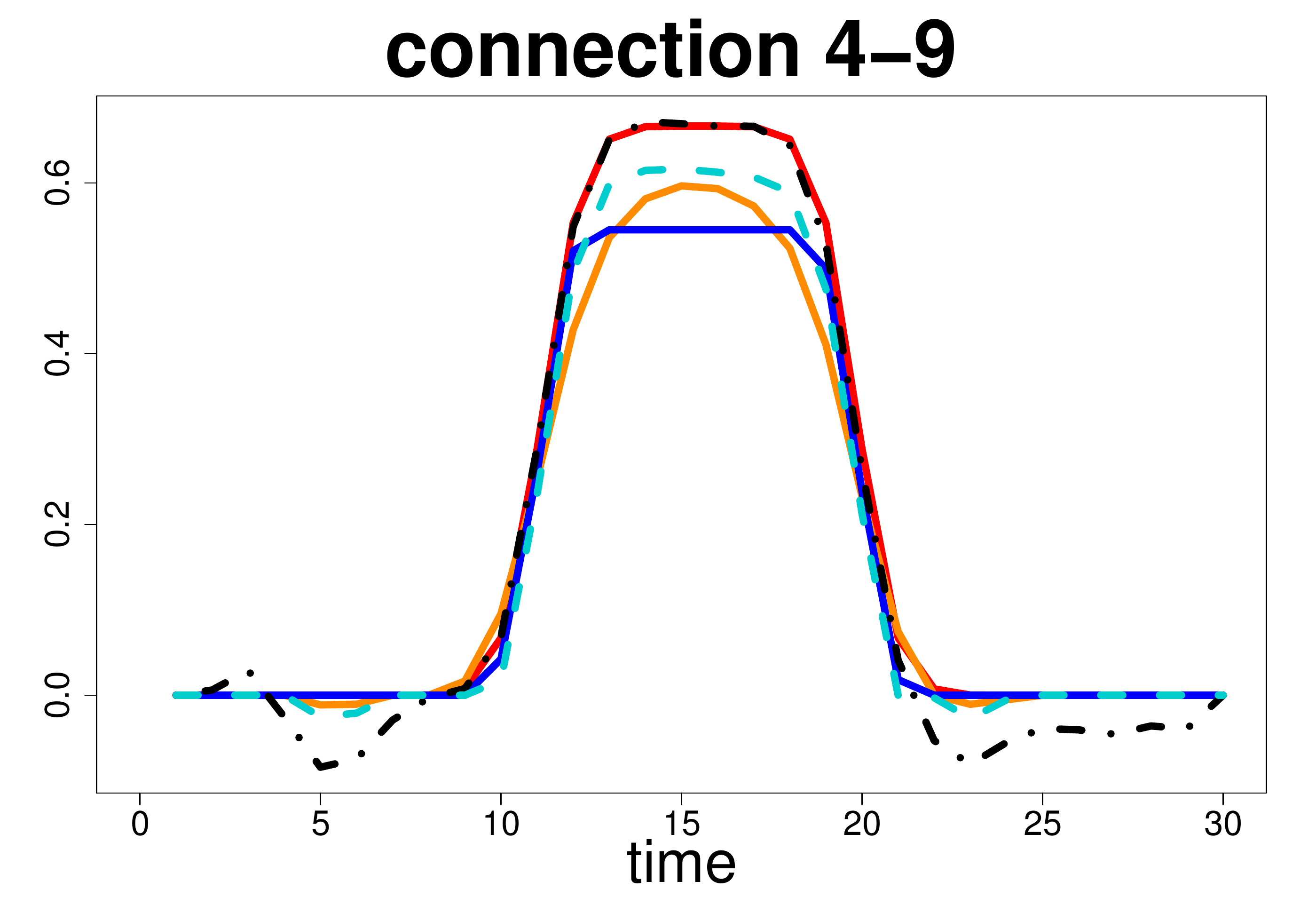}
\includegraphics[angle=270,width=0.3\linewidth]{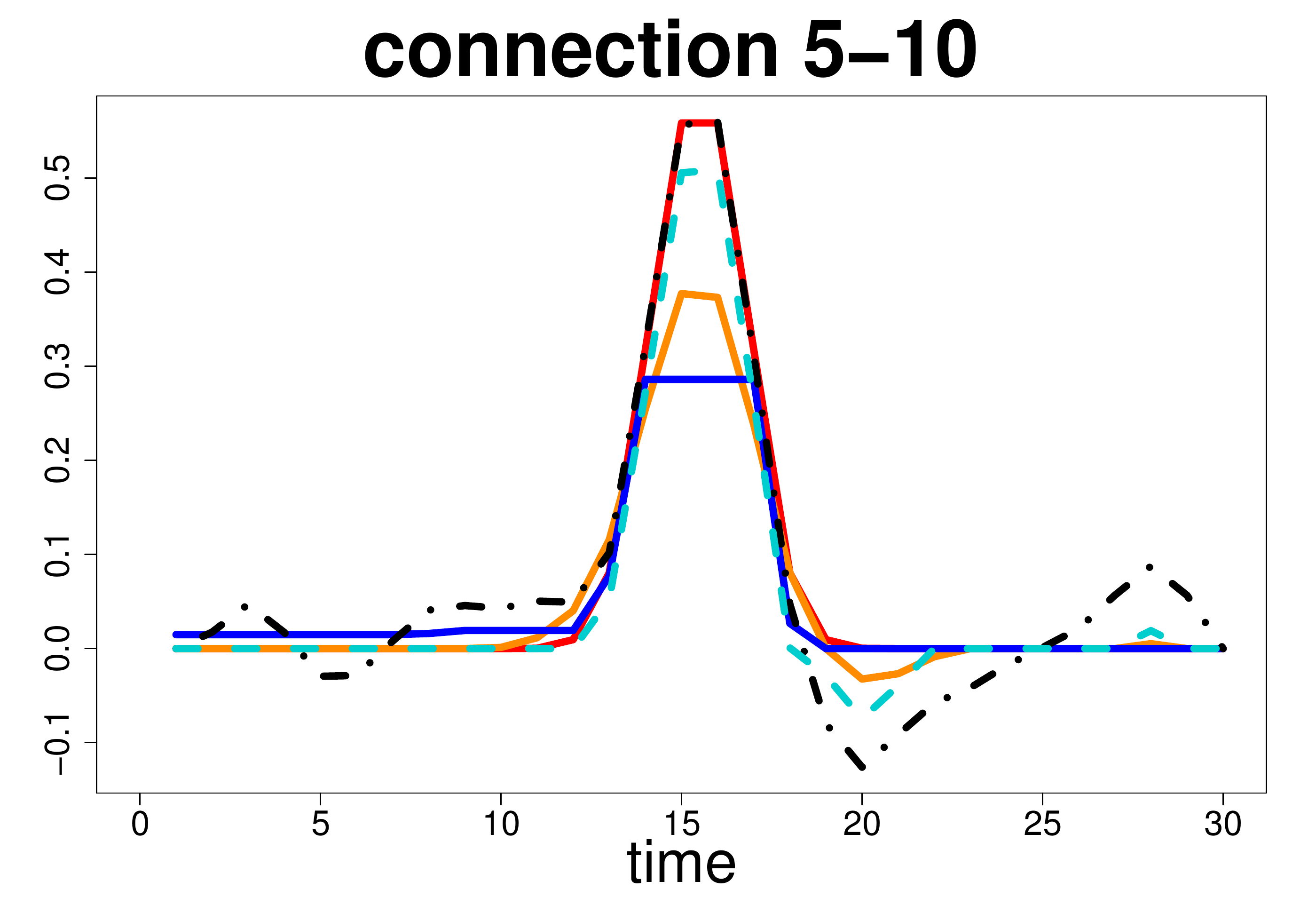}
\includegraphics[angle=270,width=0.3\linewidth]{profilelegend-eps-converted-to.pdf}
\caption{Sample size 200}
\end{subfigure}
\caption{Estimated partial correlations for true non-vanishing edges in Scenario 1.}
\label{fig:s1lines}
\end{figure}

\begin{figure}[htp]
\centering
\begin{subfigure}[b]{\textwidth}
\centering
\includegraphics[angle=270,width=0.3\linewidth]{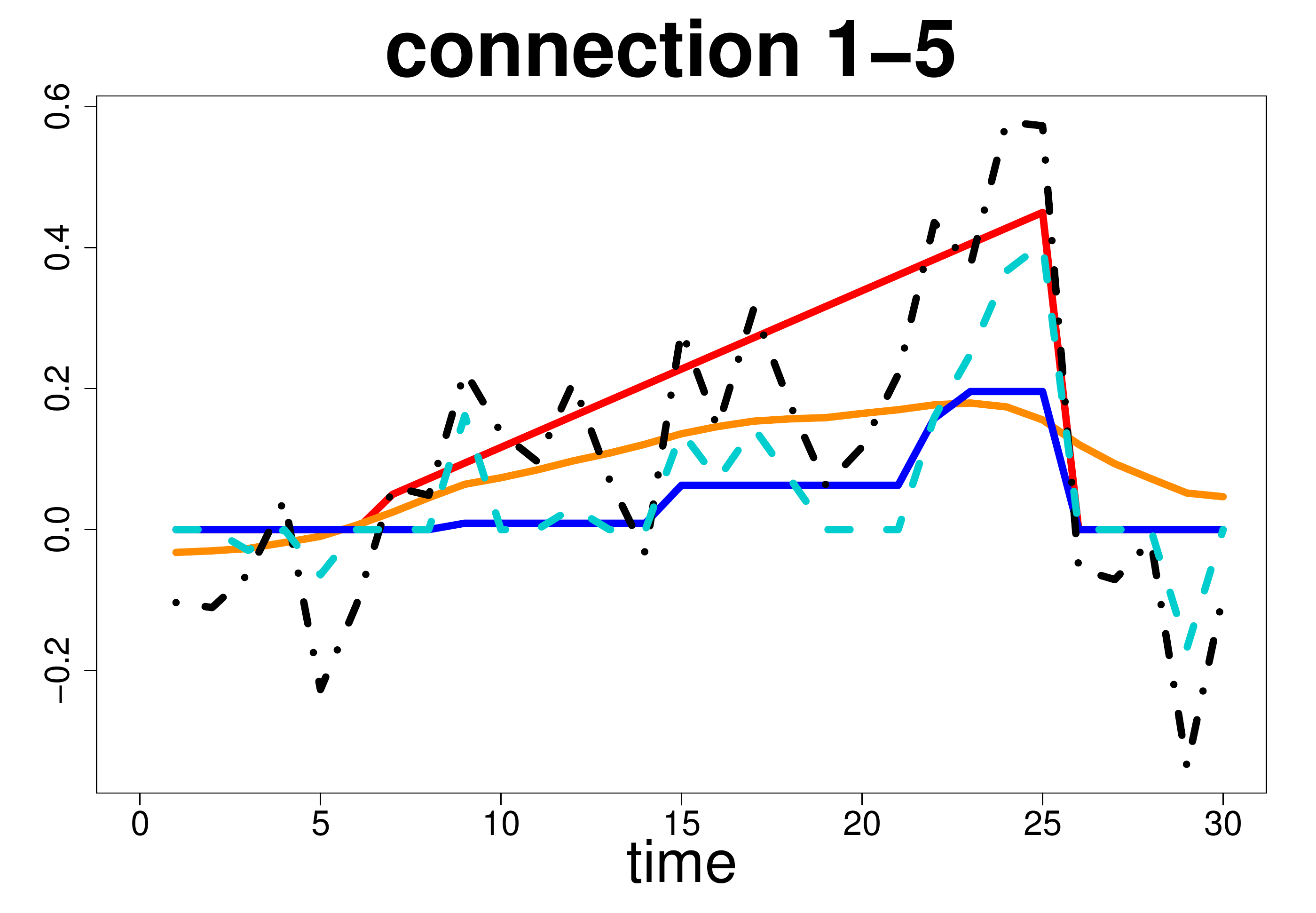}
\includegraphics[angle=270,width=0.3\linewidth]{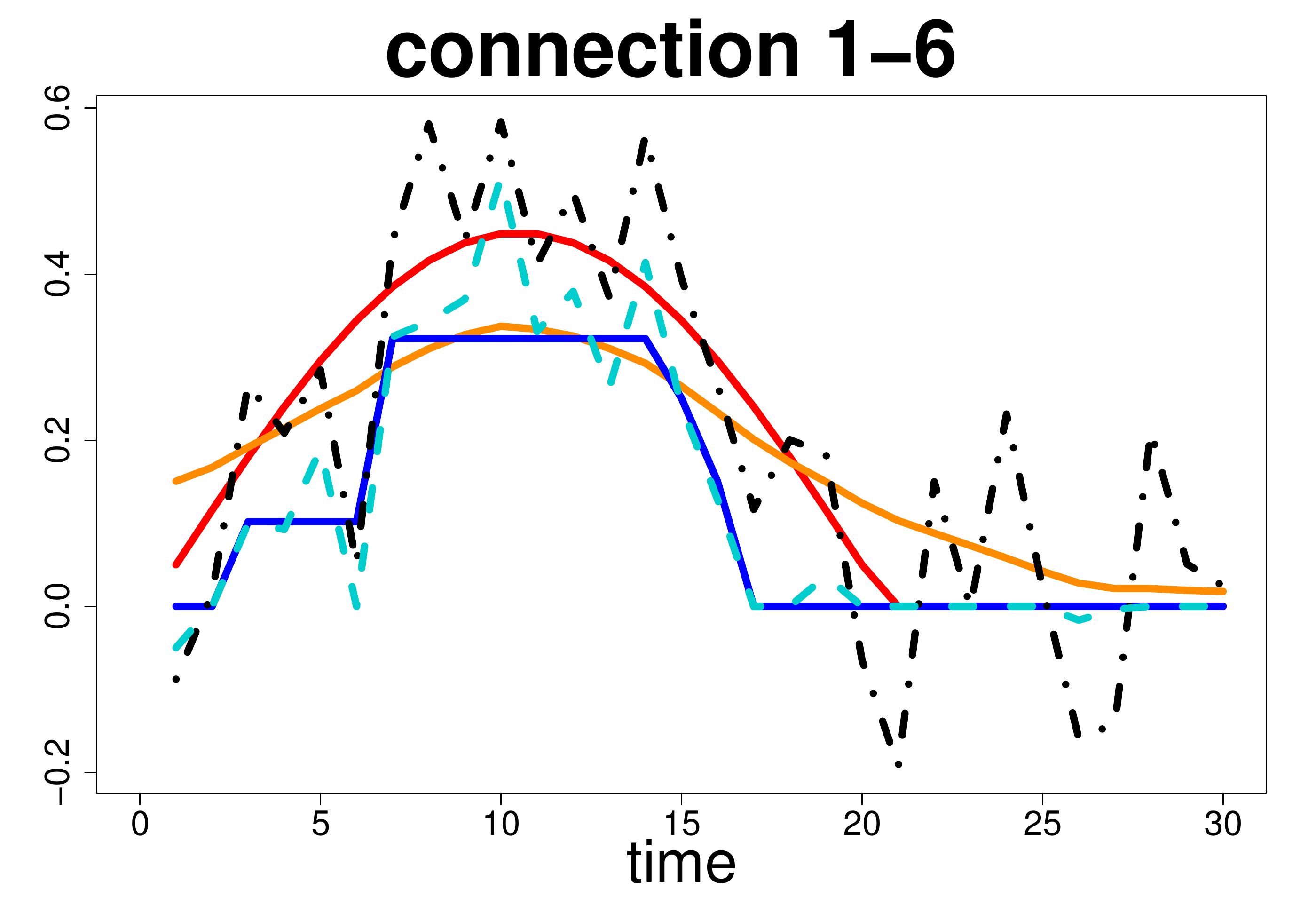}
\includegraphics[angle=270,width=0.3\linewidth]{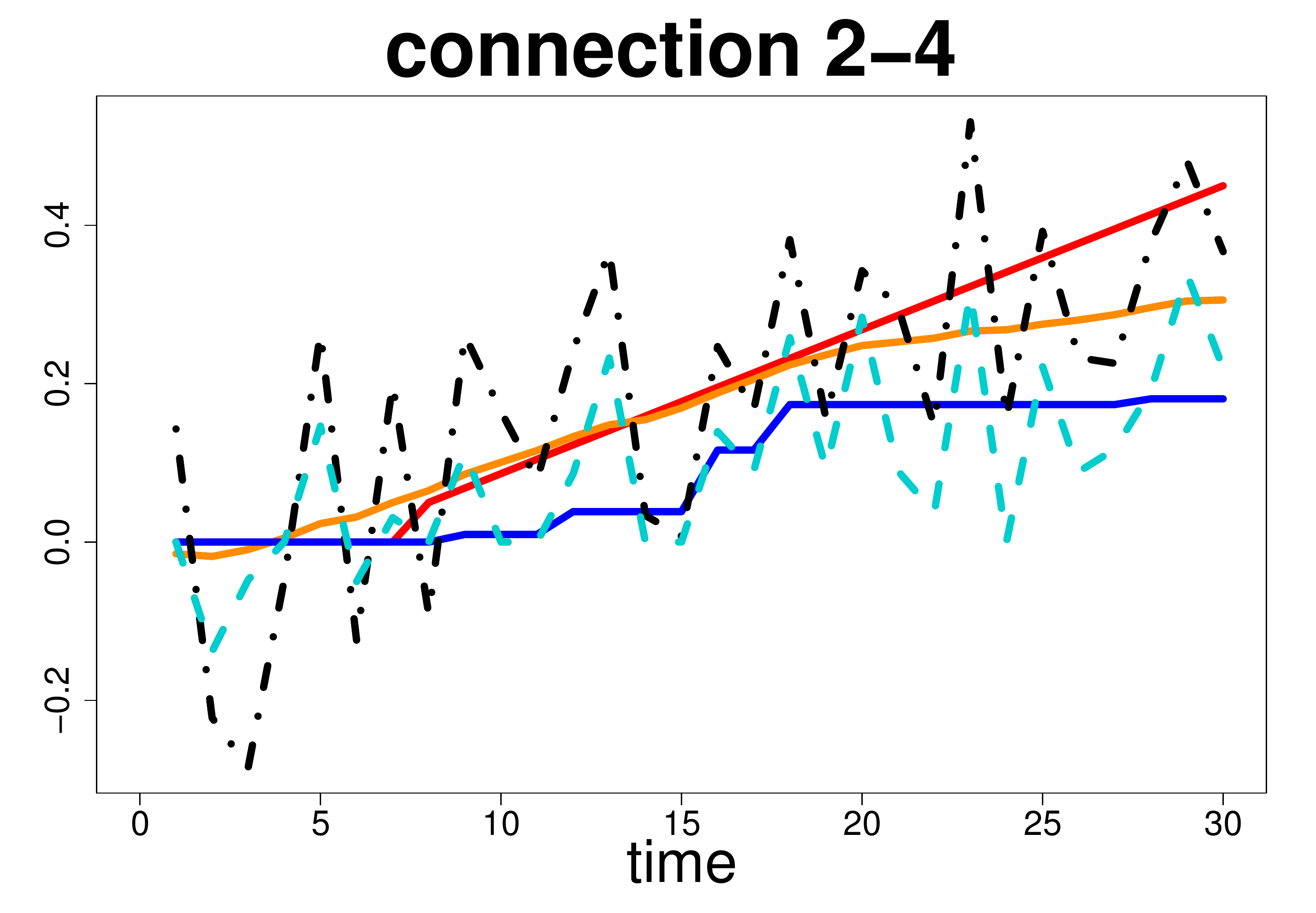}
\includegraphics[angle=270,width=0.3\linewidth]{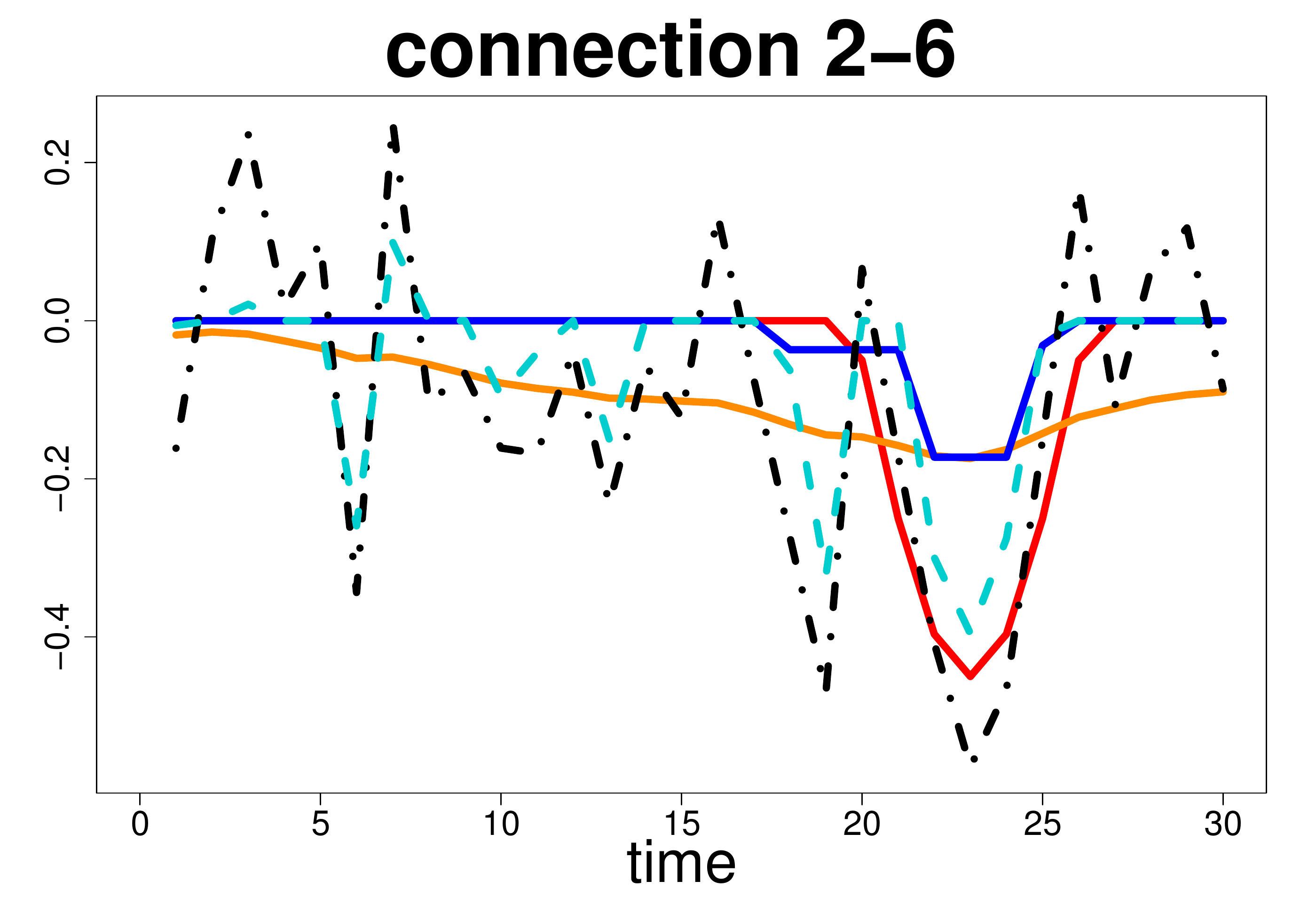}
\includegraphics[angle=270,width=0.3\linewidth]{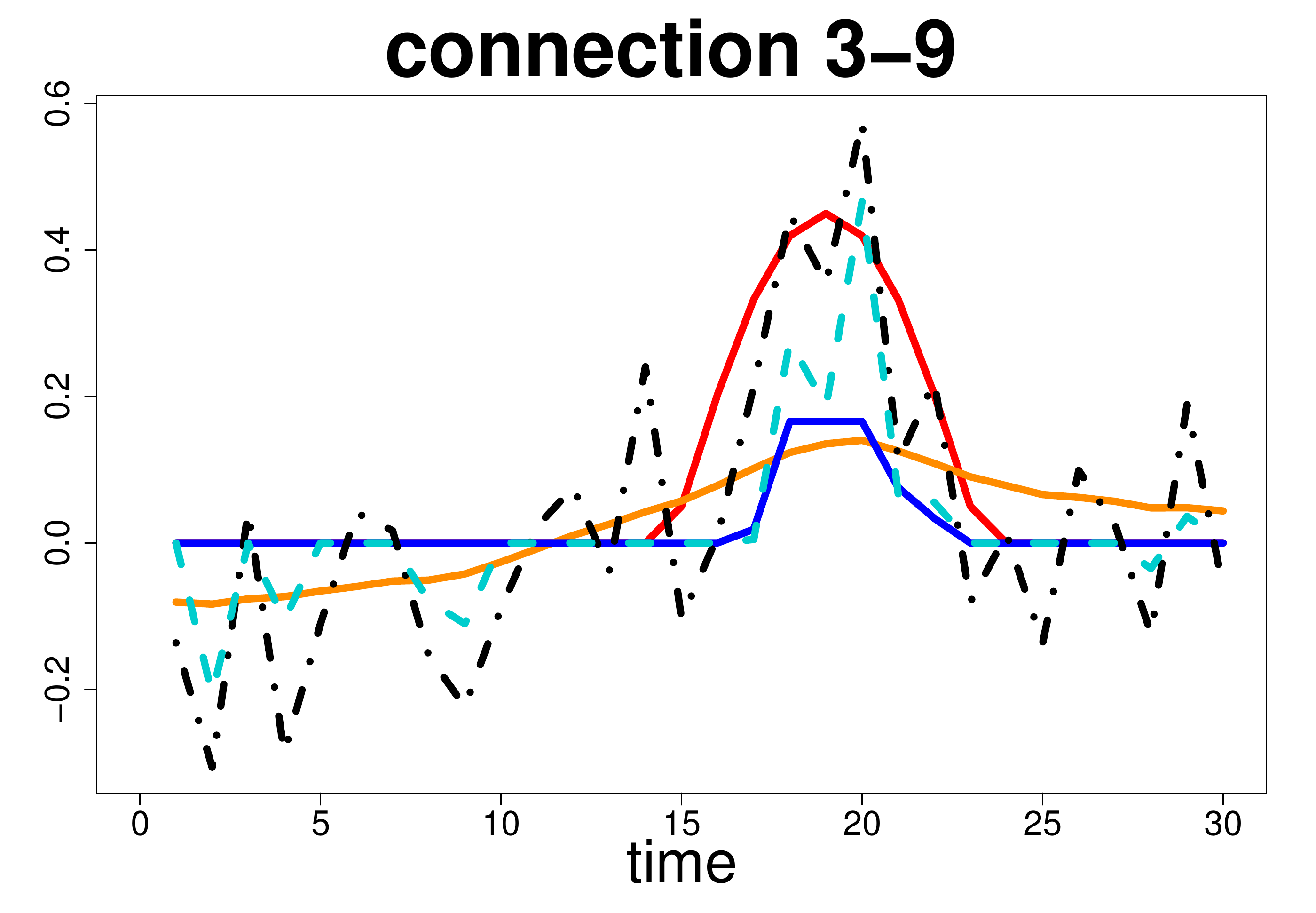}
\includegraphics[angle=270,width=0.3\linewidth]{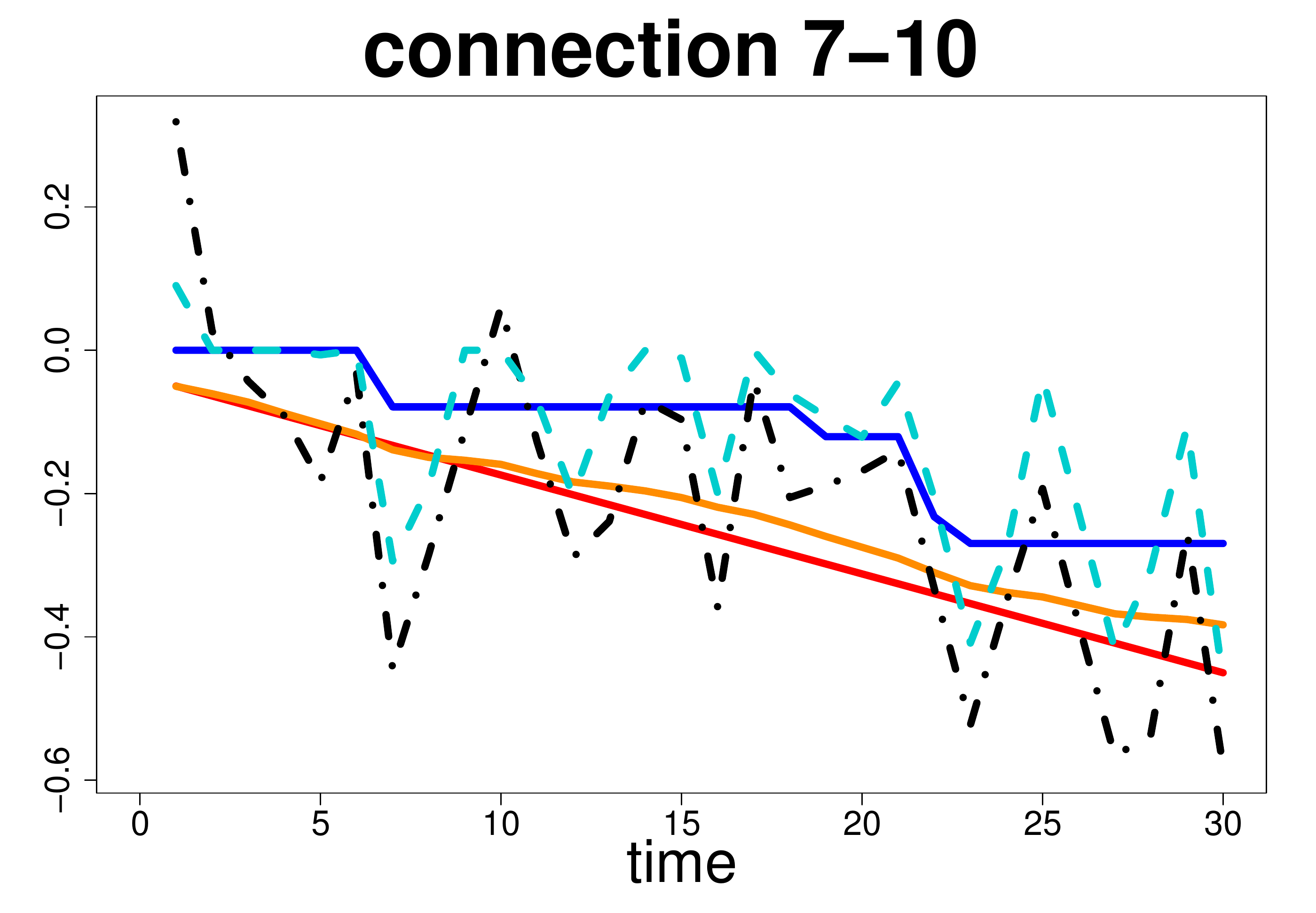}
\includegraphics[width=0.7\linewidth]{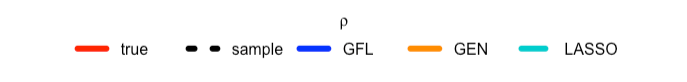}
\caption{Sample size 50}
\end{subfigure}
\begin{subfigure}[b]{\textwidth}
\centering
\includegraphics[angle=270,width=0.3\linewidth]{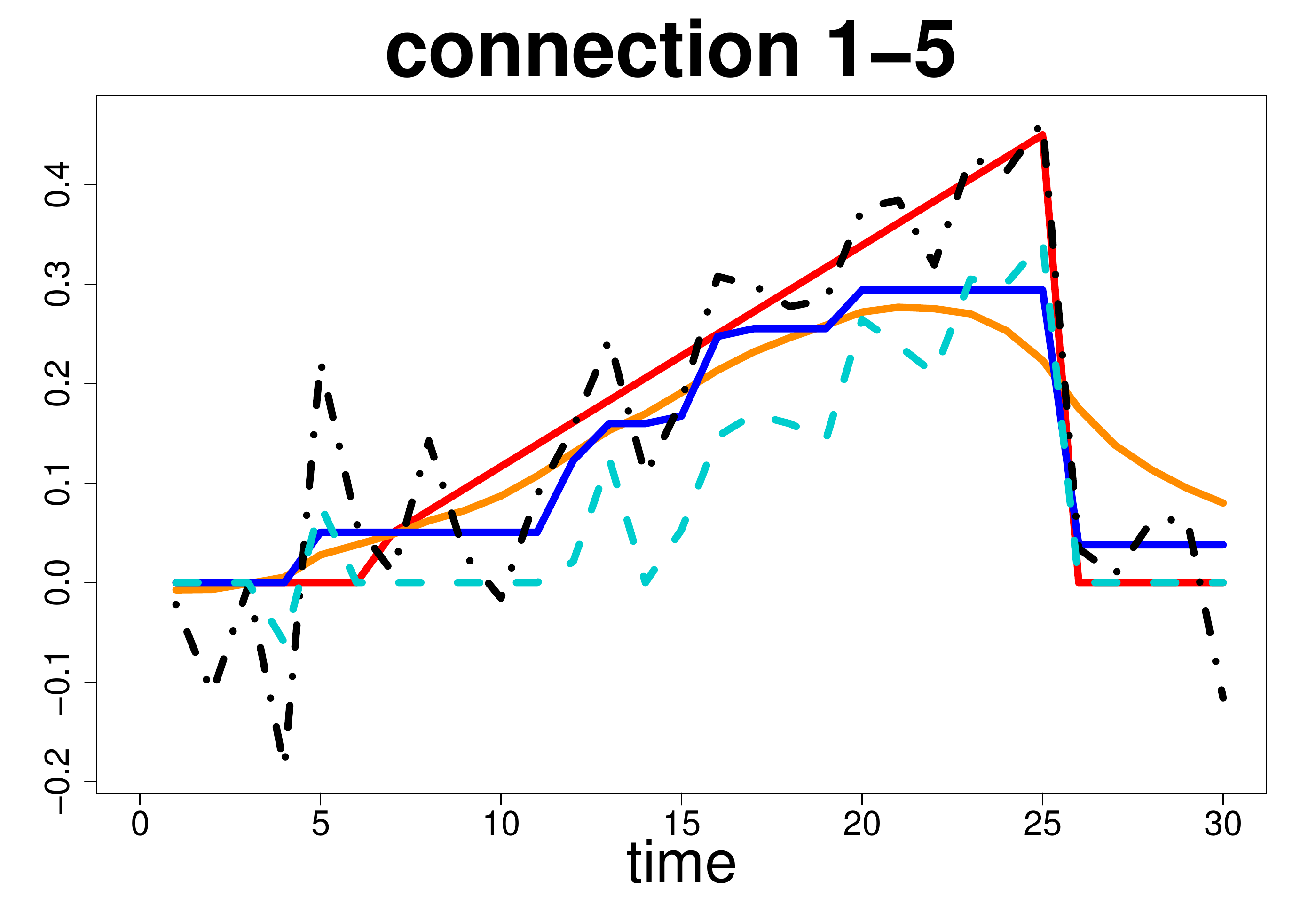}
\includegraphics[angle=270,width=0.3\linewidth]{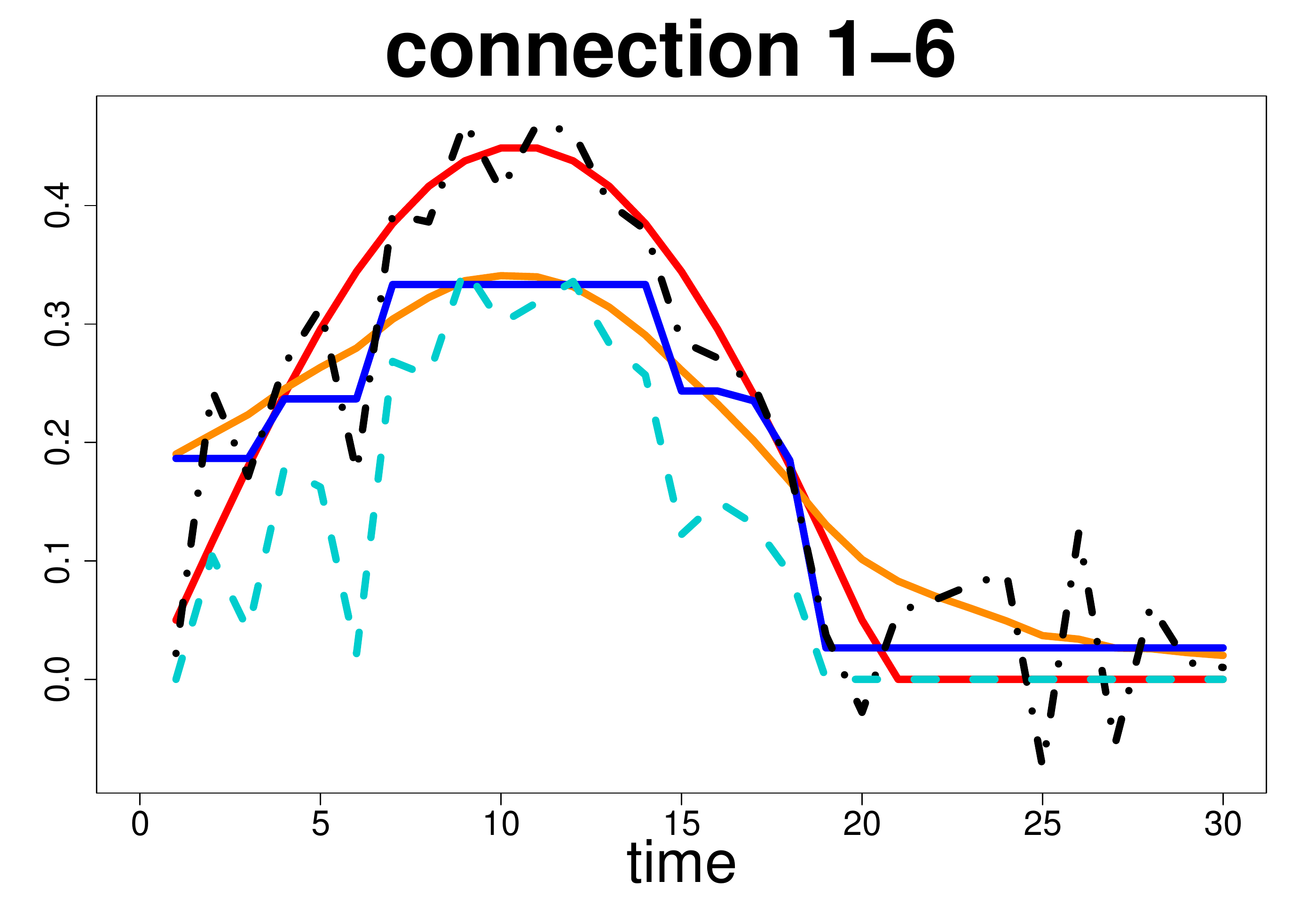}
\includegraphics[angle=270,width=0.3\linewidth]{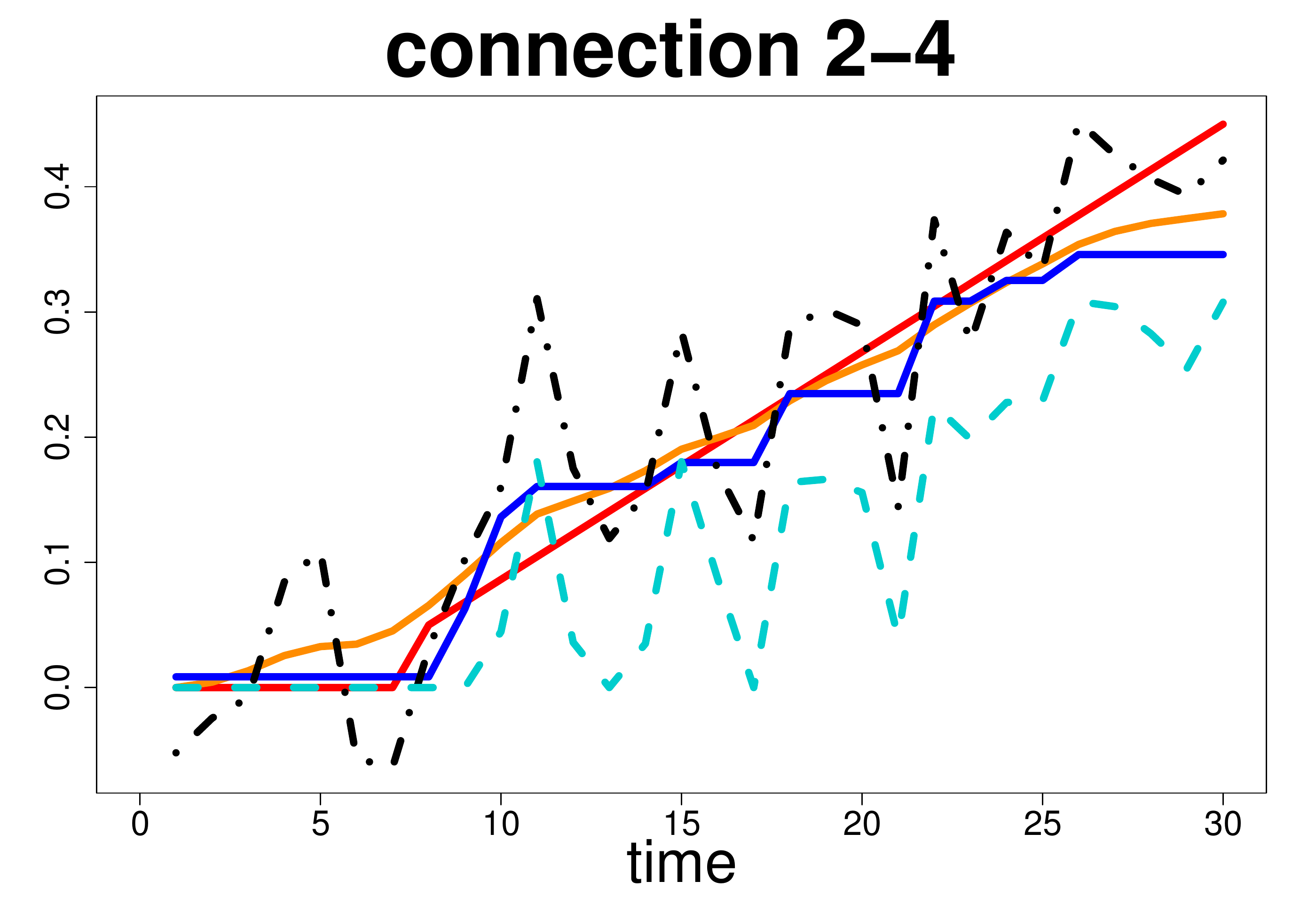}
\includegraphics[angle=270,width=0.3\linewidth]{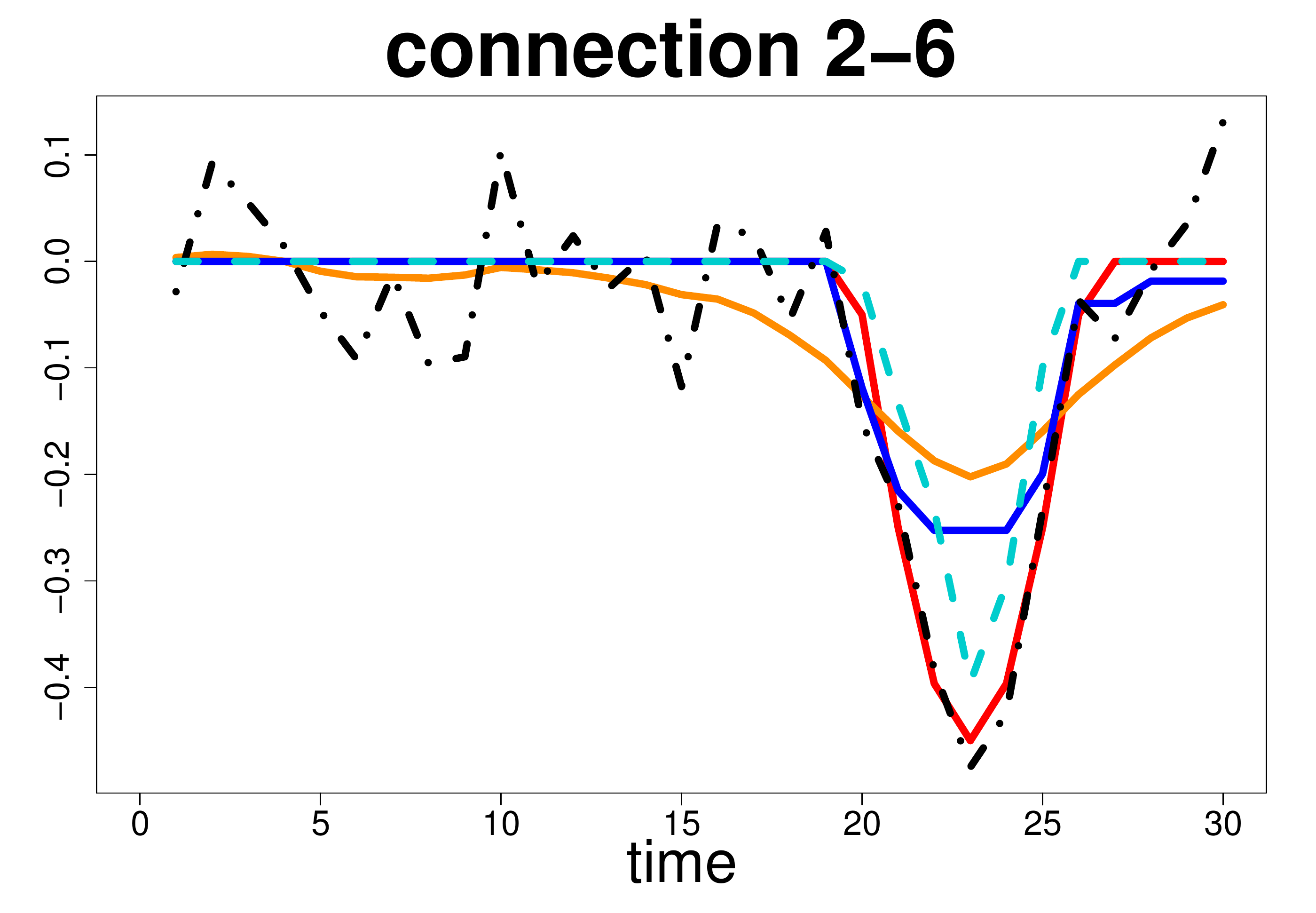}
\includegraphics[angle=270,width=0.3\linewidth]{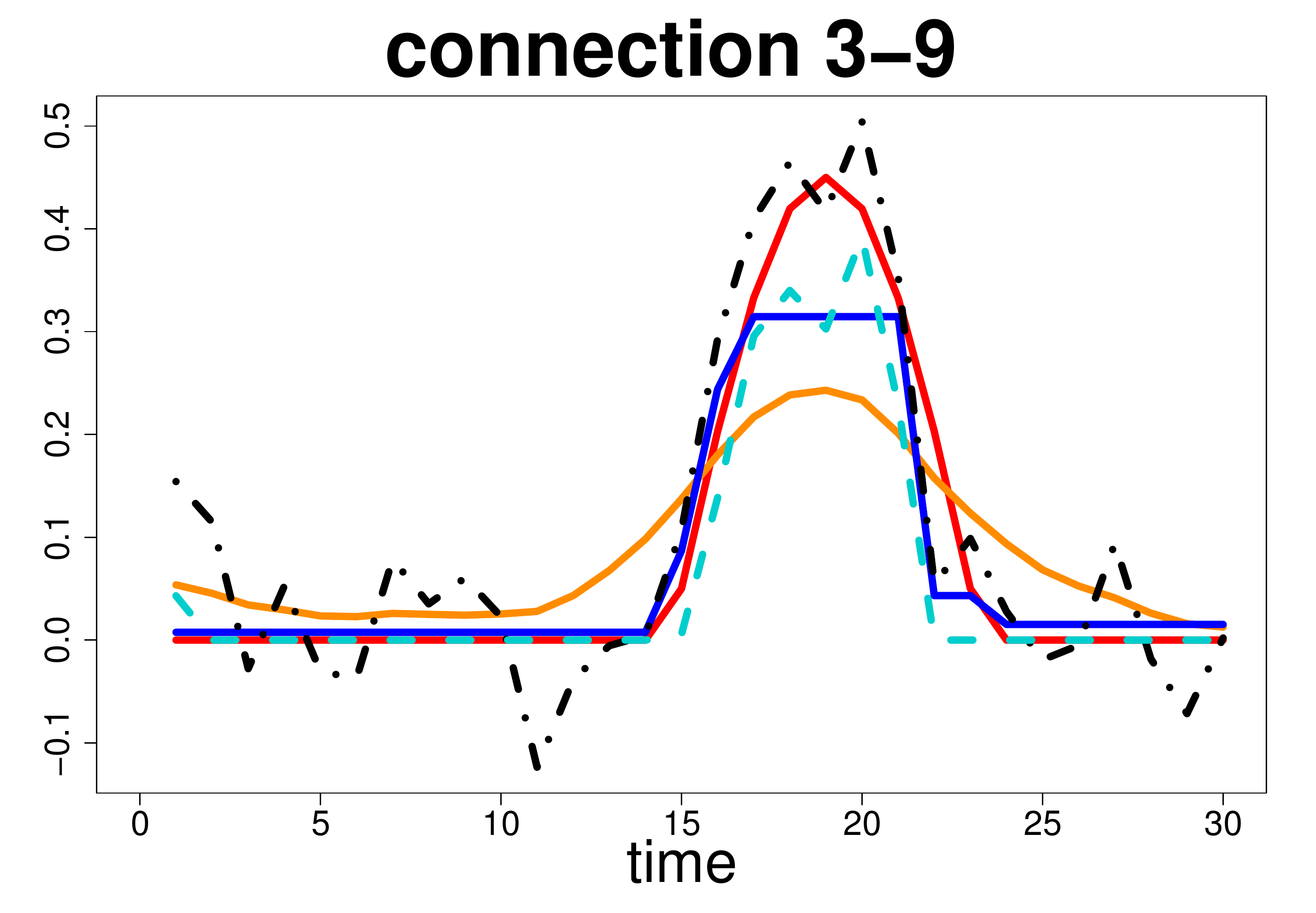}
\includegraphics[angle=270,width=0.3\linewidth]{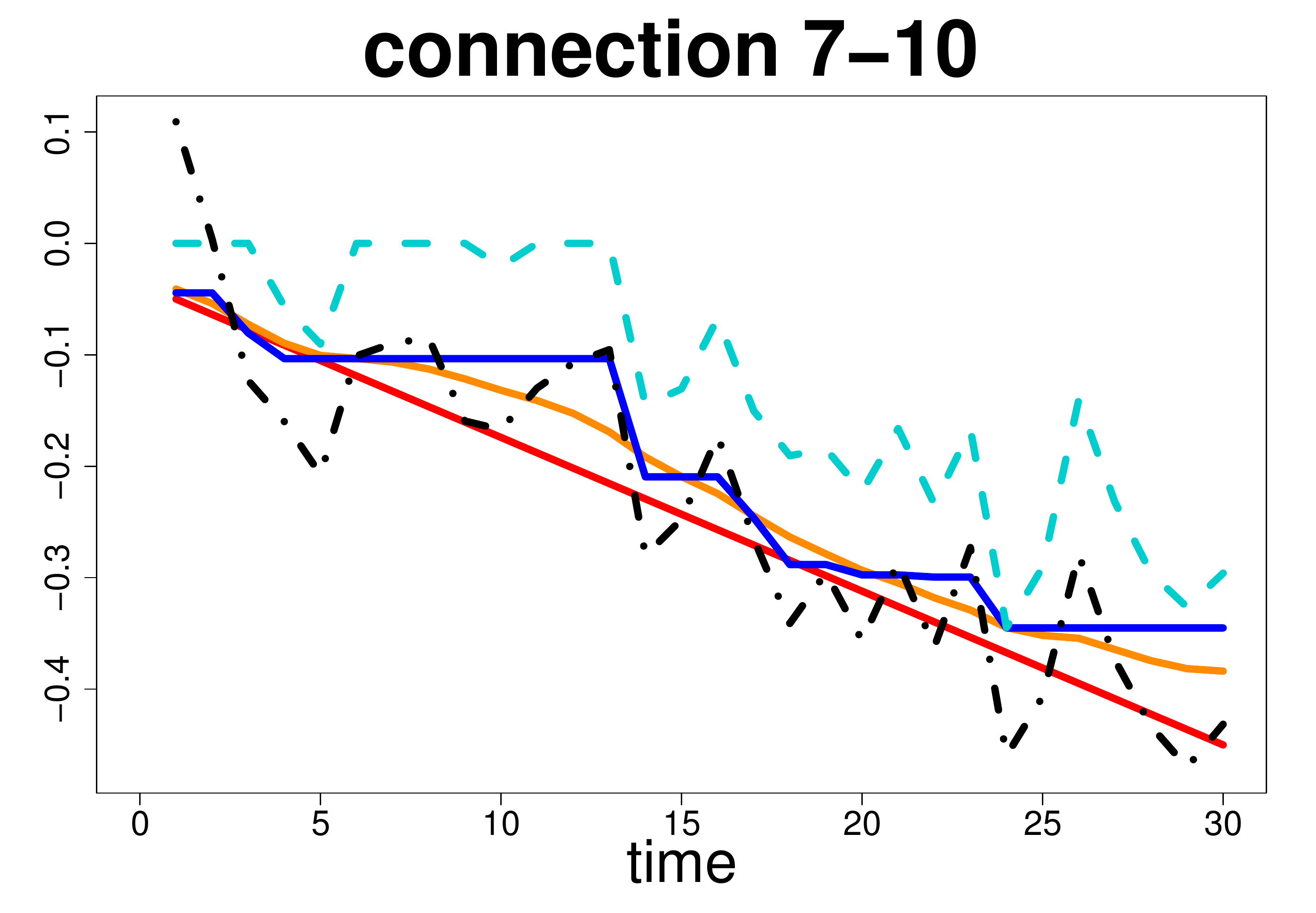}
\includegraphics[width=0.7\linewidth]{s2profilelegend}
\caption{Sample size 200}
\end{subfigure}
\caption{Estimated partial correlations for true non-vanishing edges in Scenario 2.}
\label{fig:s2lines}
\end{figure}

\begin{figure}[htp]
\centering
\begin{subfigure}[b]{\textwidth}
\centering
\includegraphics[width=0.18\linewidth,valign=t]{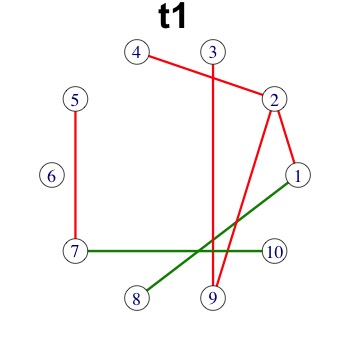}
\includegraphics[width=0.18\linewidth,valign=t]{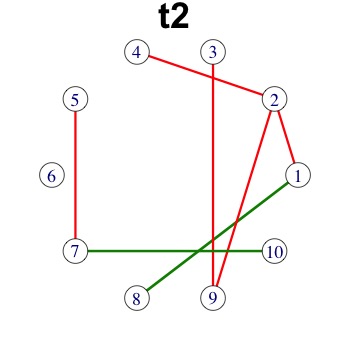}
\includegraphics[width=0.18\linewidth,valign=t]{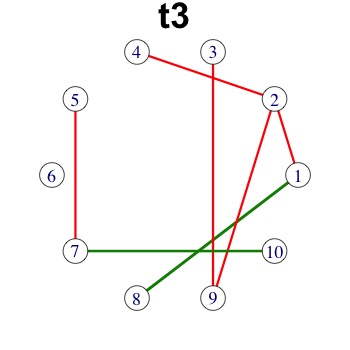}
\includegraphics[width=0.18\linewidth,valign=t]{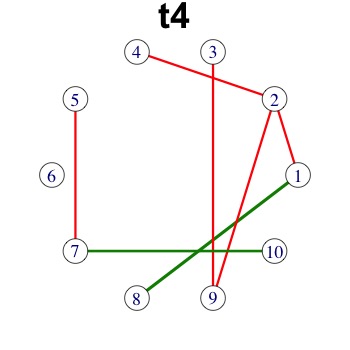}
\includegraphics[width=0.18\linewidth,valign=t]{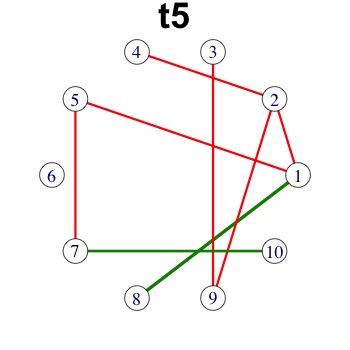}
\includegraphics[width=0.18\linewidth,valign=t]{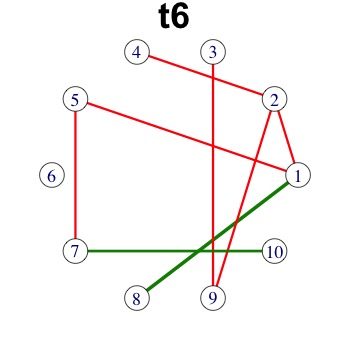}
\includegraphics[width=0.18\linewidth,valign=t]{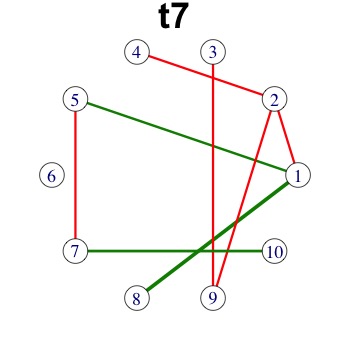}
\includegraphics[width=0.18\linewidth,valign=t]{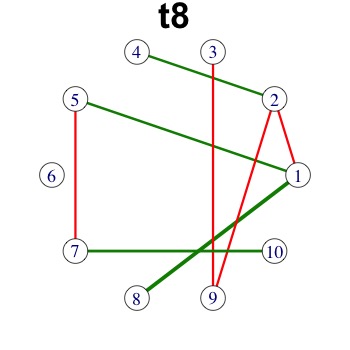}
\includegraphics[width=0.18\linewidth,valign=t]{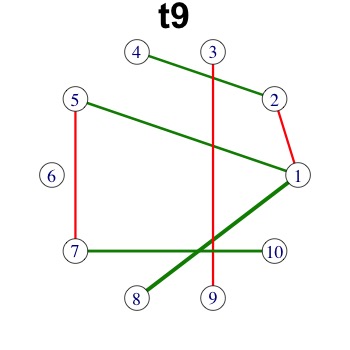}
\includegraphics[width=0.18\linewidth,valign=t]{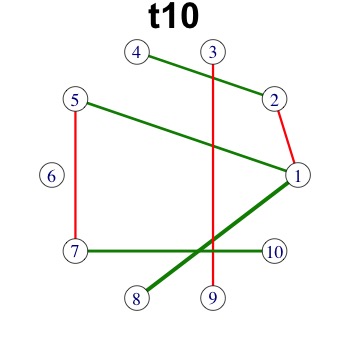}
\includegraphics[width=0.18\linewidth,valign=t]{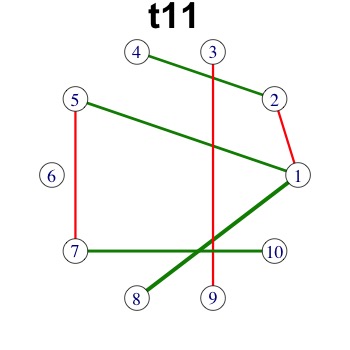}
\includegraphics[width=0.18\linewidth,valign=t]{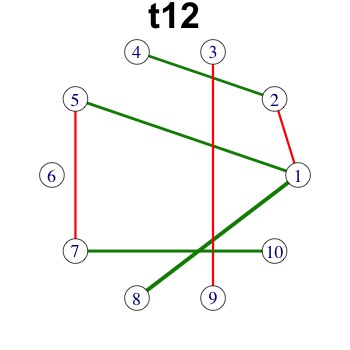}
\includegraphics[width=0.18\linewidth,valign=t]{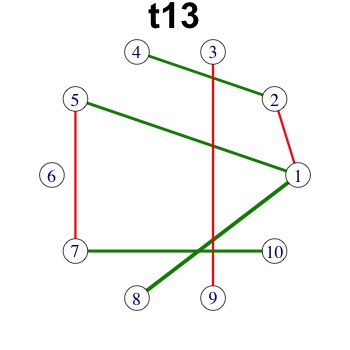}
\includegraphics[width=0.18\linewidth,valign=t]{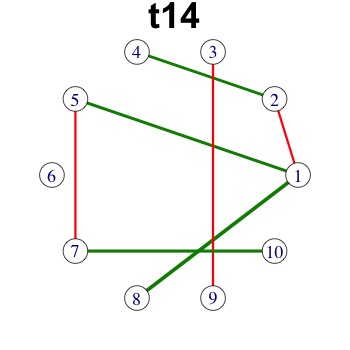}
\includegraphics[width=0.18\linewidth,valign=t]{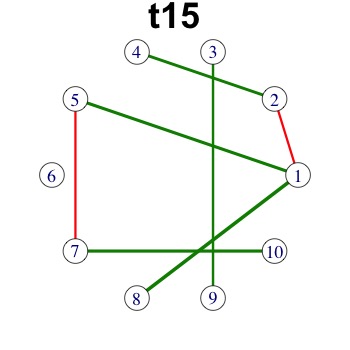}
\includegraphics[width=0.18\linewidth,valign=t]{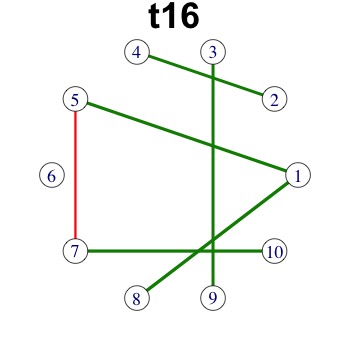}
\includegraphics[width=0.18\linewidth,valign=t]{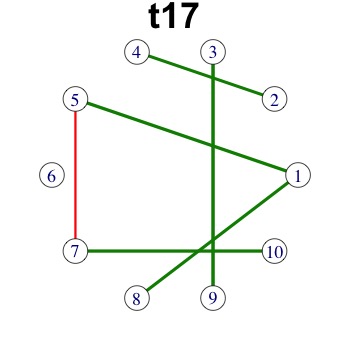}
\includegraphics[width=0.18\linewidth,valign=t]{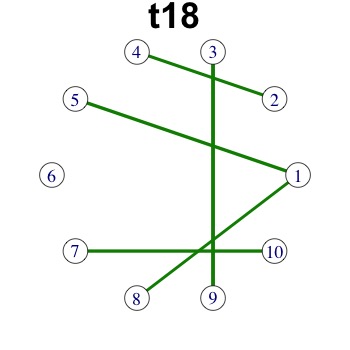}
\includegraphics[width=0.18\linewidth,valign=t]{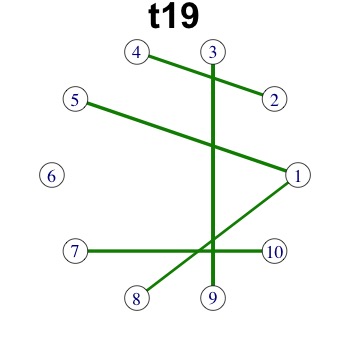}
\includegraphics[width=0.18\linewidth,valign=t]{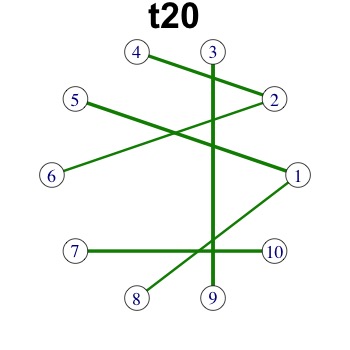}
\includegraphics[width=0.18\linewidth,valign=t]{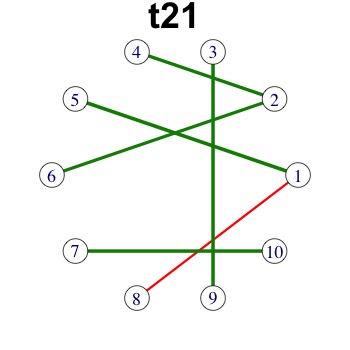}
\includegraphics[width=0.18\linewidth,valign=t]{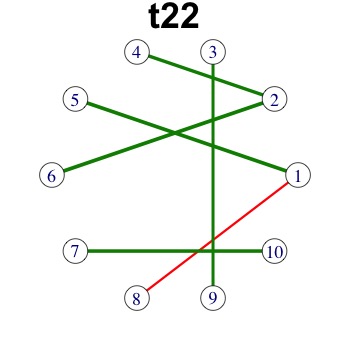}
\includegraphics[width=0.18\linewidth,valign=t]{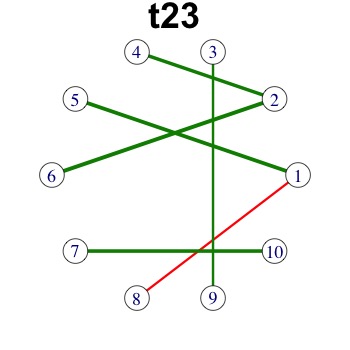}
\includegraphics[width=0.18\linewidth,valign=t]{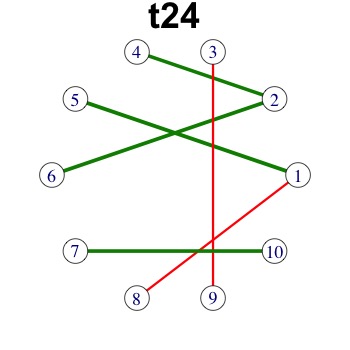}
\includegraphics[width=0.18\linewidth,valign=t]{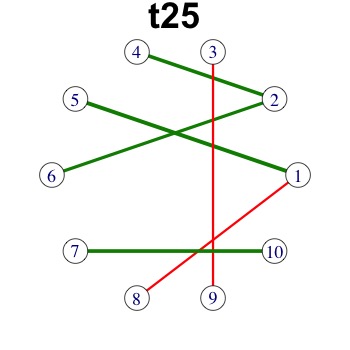}
\includegraphics[width=0.18\linewidth,valign=t]{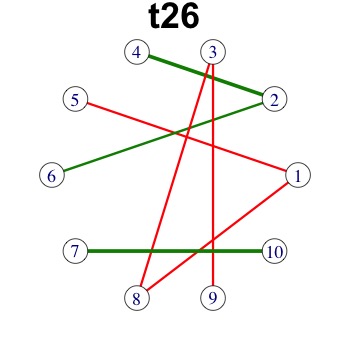}
\includegraphics[width=0.18\linewidth,valign=t]{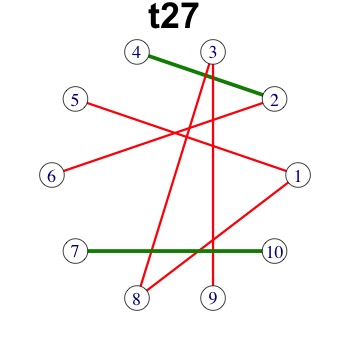}
\includegraphics[width=0.18\linewidth,valign=t]{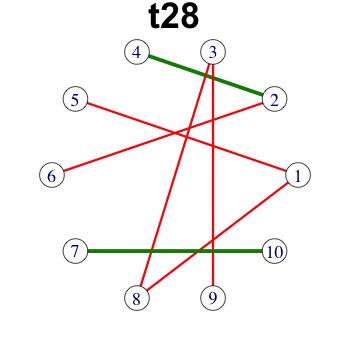}
\includegraphics[width=0.18\linewidth,valign=t]{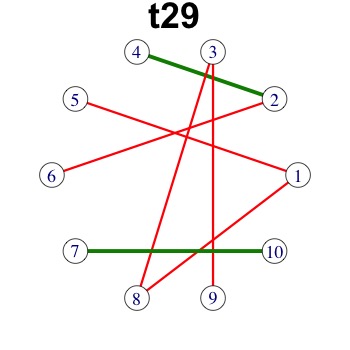}
\includegraphics[width=0.18\linewidth,valign=t]{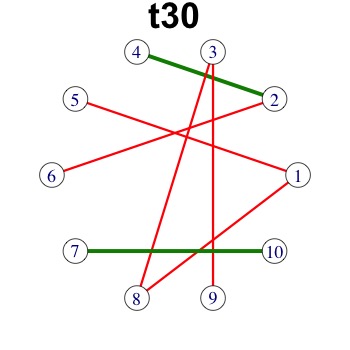}
\end{subfigure}
\caption{The GFL-based partial correlation networks at different time points in Scenario 1 with sample size 200. The green solid lines, the green dashed lines and the red solid lines represent the true positive, false negative and false positive connections, respectively. Thickness of each green solid line represents the magnitude of its underlying true partial correlation.}
\label{fig:s1_gfl_network}
\end{figure}

\begin{figure}[htp]
\centering
\begin{subfigure}[b]{\textwidth}
\centering
\includegraphics[angle=270,width=0.3\linewidth,valign=t]{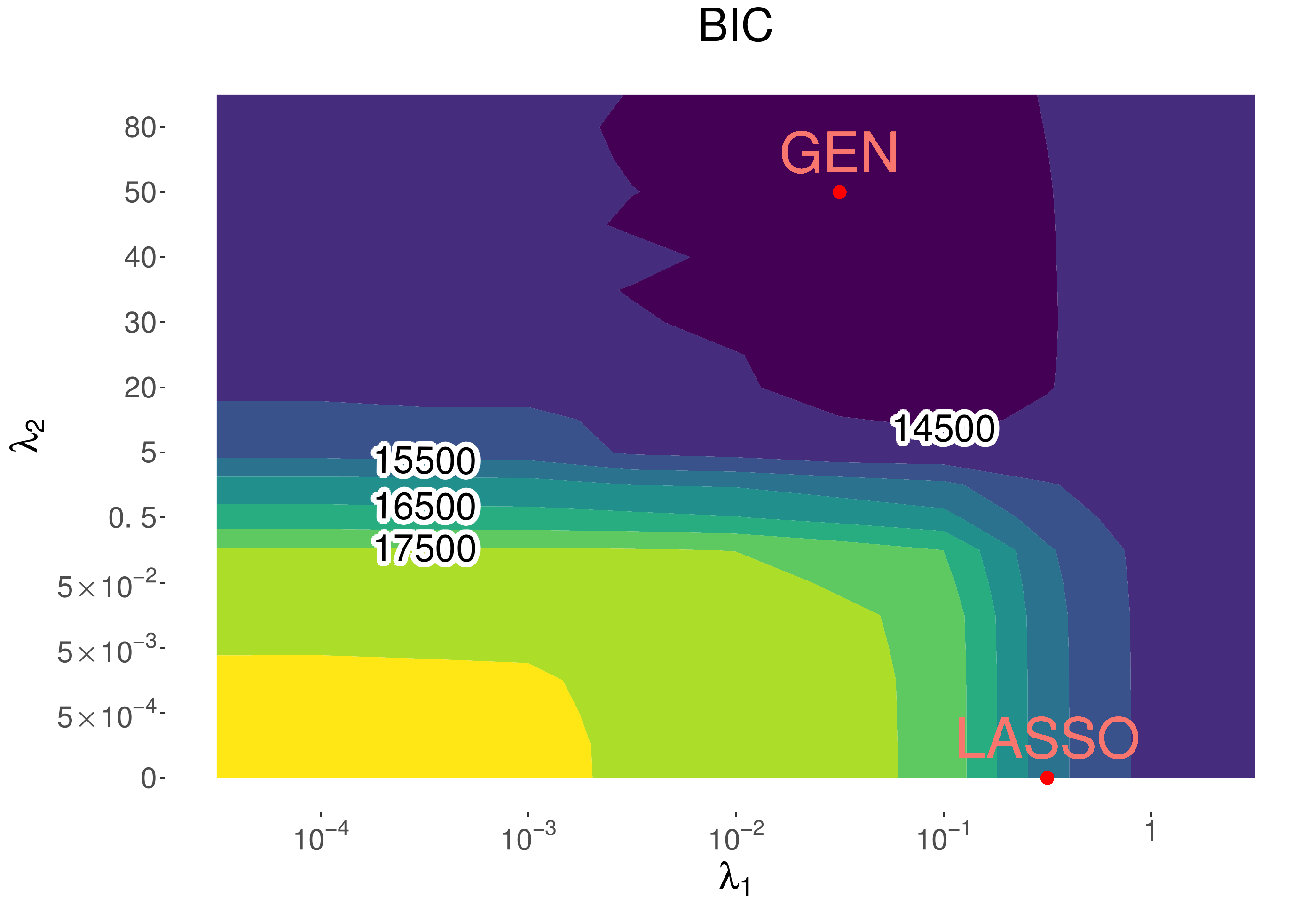}
\includegraphics[angle=270,width=0.3\linewidth,valign=t]{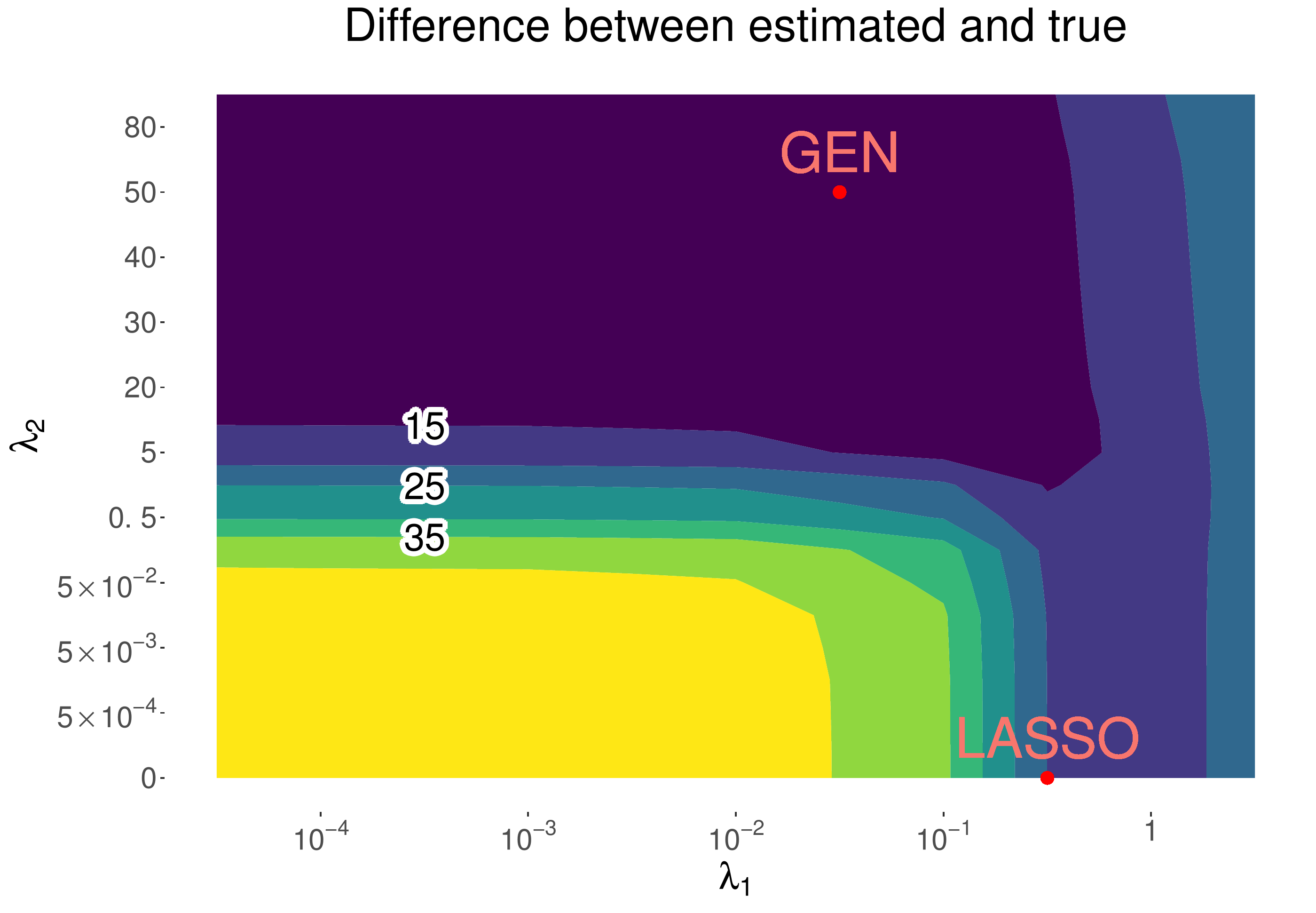}
\includegraphics[angle=270,width=0.3\linewidth,valign=t]{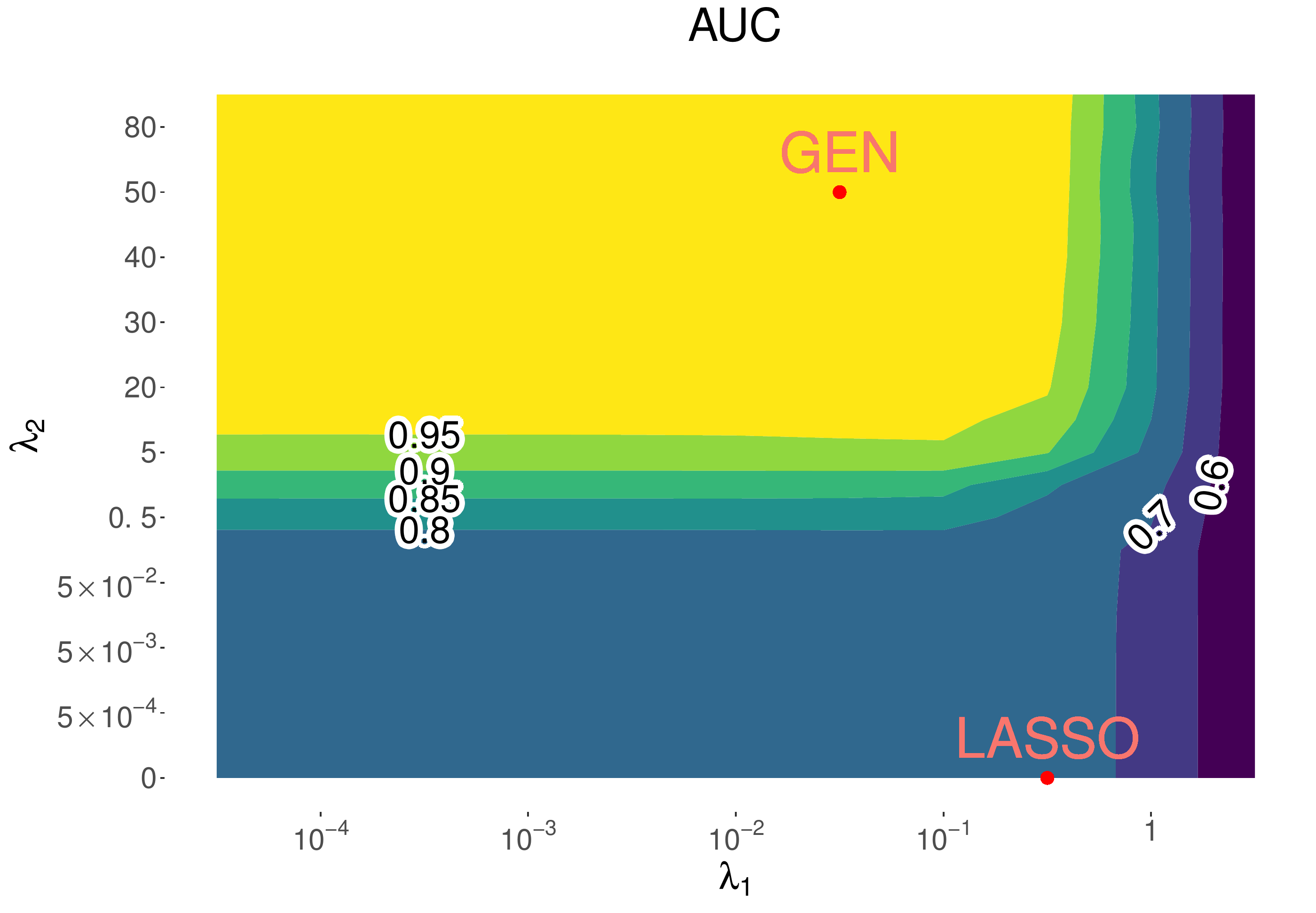}
\caption{sample size 50}
\end{subfigure}
\begin{subfigure}[b]{\textwidth}
\centering
\includegraphics[angle=270,width=0.3\linewidth,valign=t]{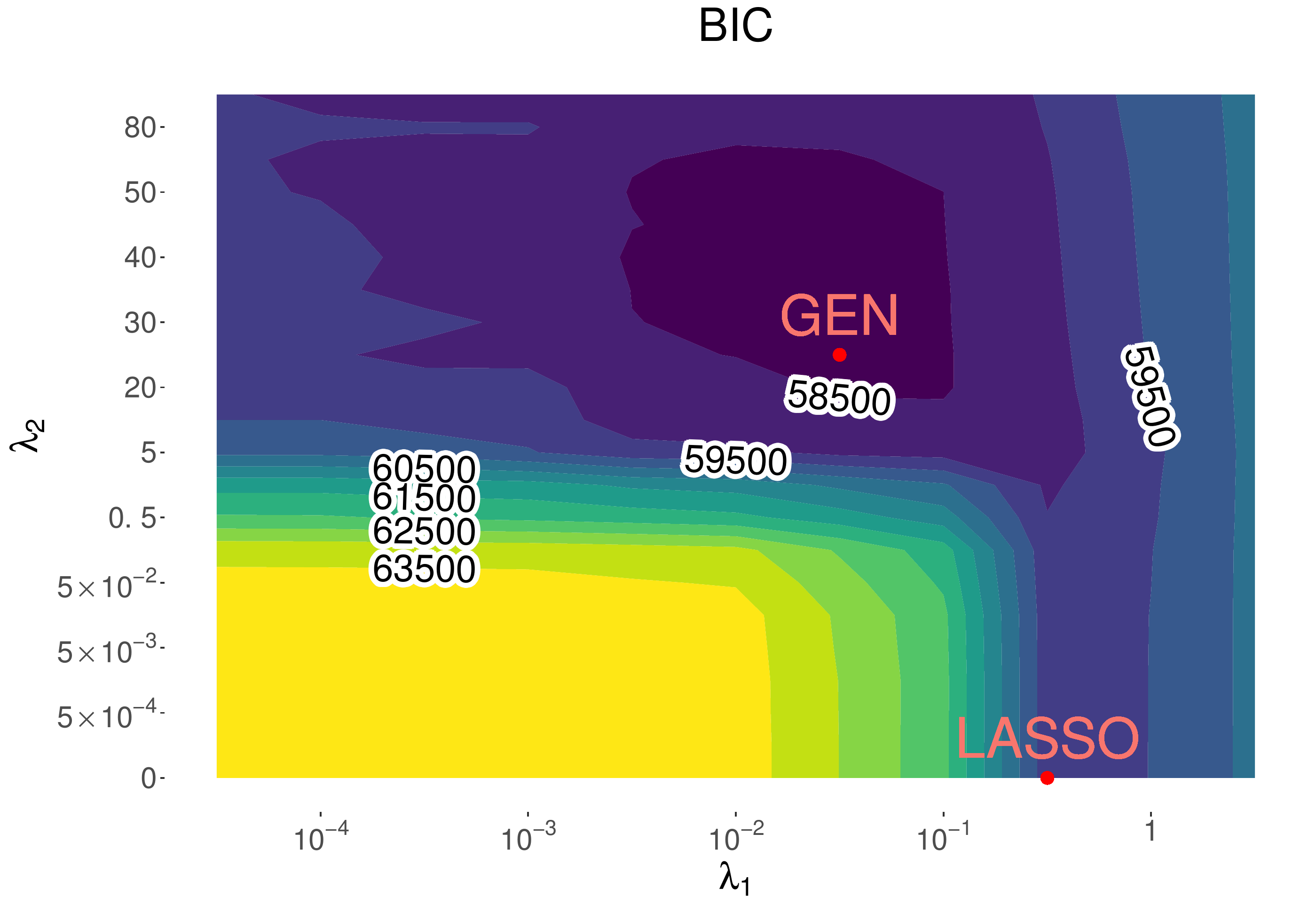}
\includegraphics[angle=270,width=0.3\linewidth,valign=t]{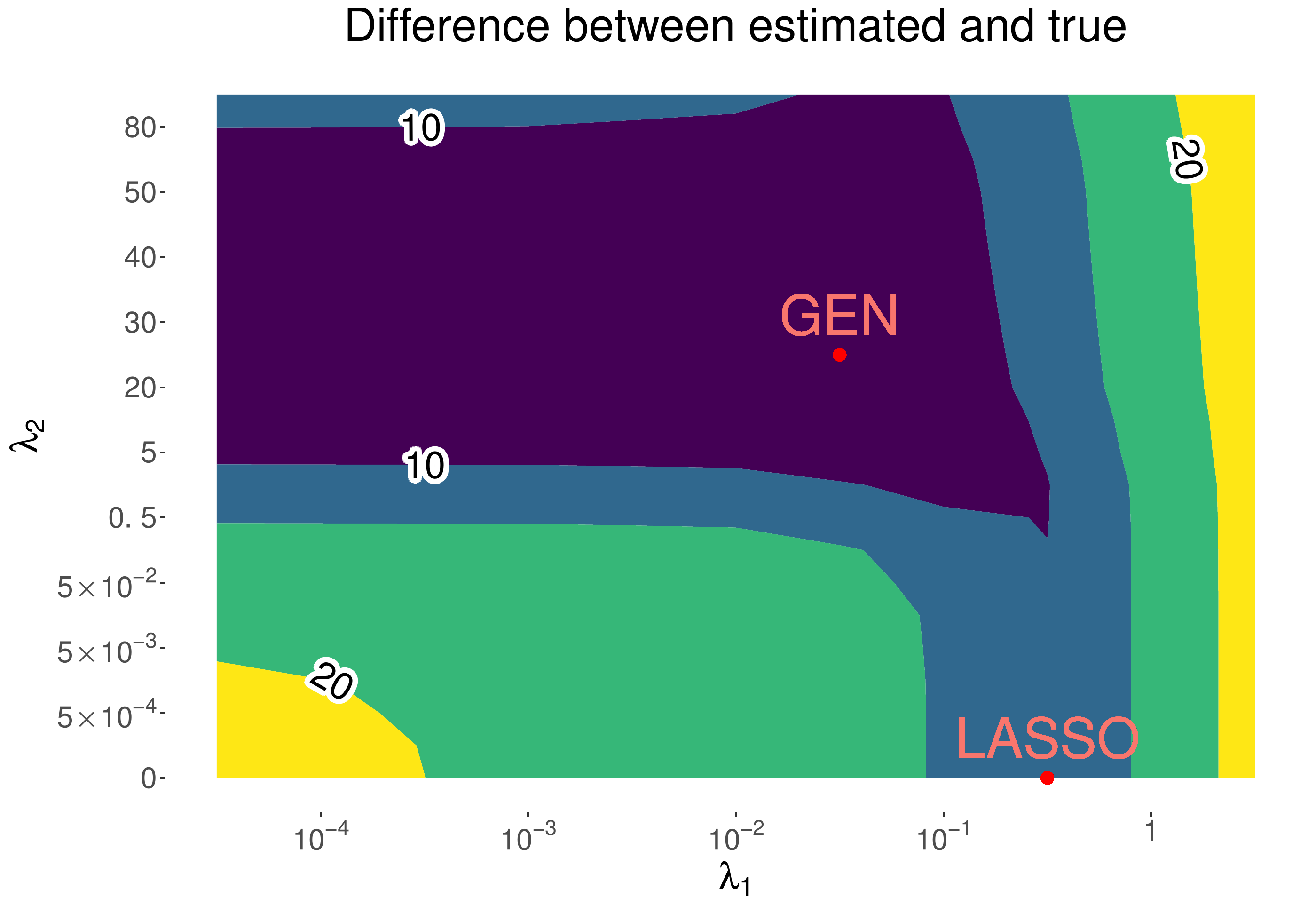}
\includegraphics[angle=270,width=0.3\linewidth,valign=t]{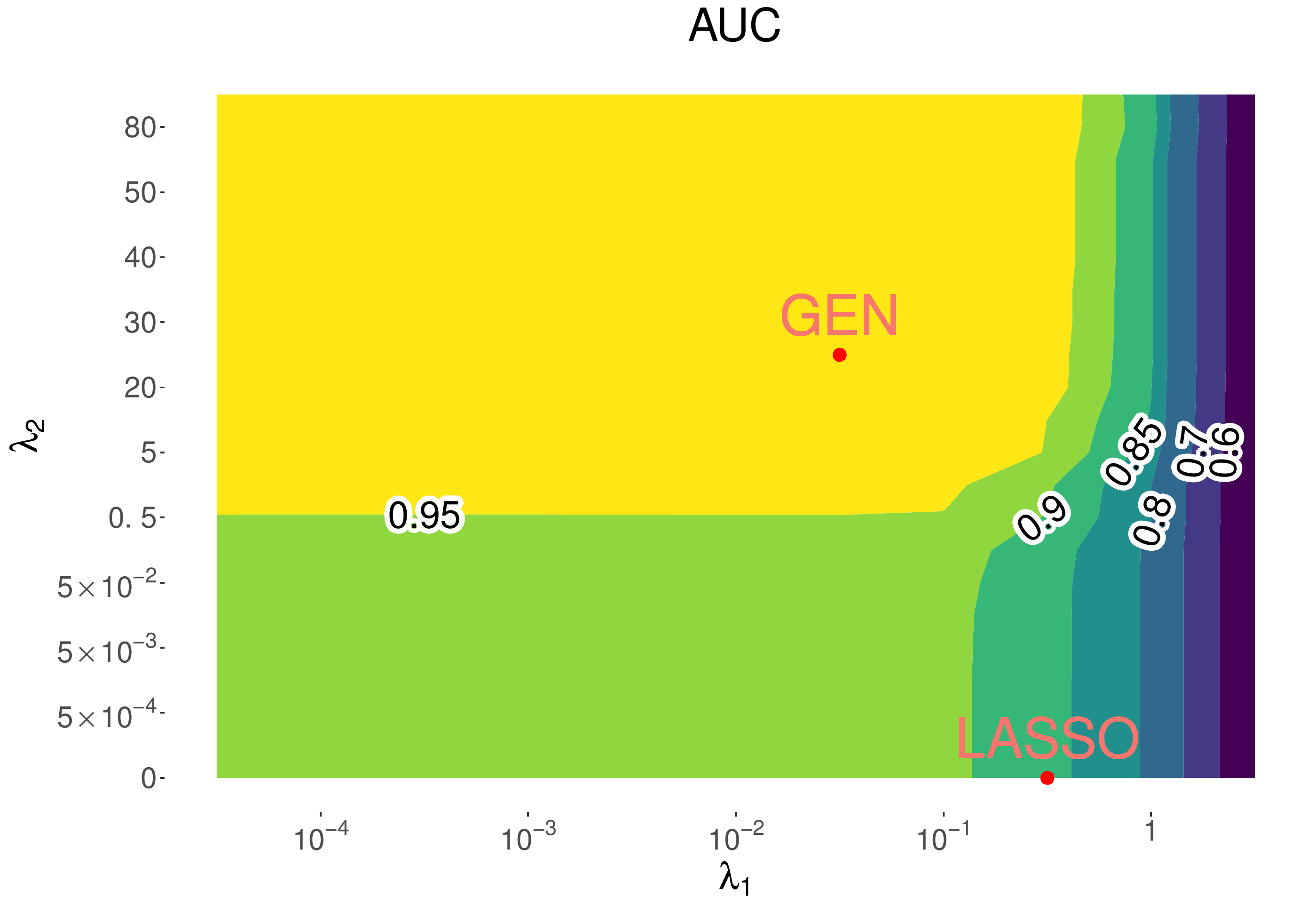}
\caption{sample size 200}
\end{subfigure}
\caption{Three performance metrics based on the GEN in Scenario 2. Darker areas represent lower values. 
}
\label{fig:s2_gen_contour}
\end{figure}

\section{Application: Analysis of ADHD data}
\label{sec:application}

In this section, we apply the proposed methods (GEN and GFL) to estimate time-varying brain connectivity with a real data set about attention deficit hyperactive disorder (ADHD).

ADHD is a mental health disorder characterized by impulsivity, motoric hyperactivity, and especially, attention deficits. With a global community prevalence of 5\% - 10\%, however, its causes are unknown, although gene damage may be a contributing factor. Diagnosis is mainly supported by clinical assessment based on long-term observations and on identifying a range of symptoms. 
Recently, the cerebellum that contains more than 50\% of the neurons in the brain has been thought of as having vital involvement in ADHD. MRI studies \citep{berquin1998cerebellum} suggest that the cerebellar hemispheric volumes in ADHD sufferers are up to 6\% smaller than healthy subjects, and ADHD children have less vermal volume than healthy ones. Subsequent researches consistently find significant differences between ADHD brains and healthy ones in posterior inferior lobe of the cerebellum (lobules VIII–X) and the posterior–inferior cerebellar vermis \citep{mostofsky1998evaluation}. Therefore, exploring the topological structures in the cerebellum region and their dynamic changes under ADHD is critical to a better understanding of the mechanisms underlying the disorder. Moreover, knowing the differences in brain connectivity between healthy and ADHD groups can contribute to the development of diagnosis methods. 

We use the resting-state fMRI data set collected at New York University Medical Center (NYU), one of the eight imaging sites contributing to the ADHD-200 Global Competition that was held to gather neuroimaging data for the classification of ADHD subjects. The data set consists of filtered and preprocessed resting-state data for 116 brain regions of interest (ROIs) segmented by the Automated Anatomical Labeling (AAL) atlas. We extracted data only from the $91$st to the $108$th ROIs, corresponding to the cerebellum region. For each particular region, the mean blood-oxygen-level dependent (BOLD) signal was recorded at 172 equally spaced time points. There are 98 healthy subjects and 118 ADHD patients. 

To characterize how time-varying associations in cerebellum regions differ for the ADHD and healthy groups, we apply both GEN and GFL to fit time-varying networks for the two groups separately. First, we center the mean BOLD signal at each ROI to zero at each time point. Then we estimate partial correlation networks using a series of $\lambda_1$ and $\lambda_2$ values. The BIC surfaces---i.e., \eqref{bic} for both methods---turn out to be quite flat, suggesting that there are no substantial differences among different solutions. For each method, we therefore examine two specific solutions: one, which we refer to as ``Result 1'', is given by the ``first'' strictly positive $(\lambda_1,\lambda_2)>(0,0)$ in our grid; and another, which we refer to as ``Result 2'', is given by a pair of $(\lambda_1,\lambda_2)$ in our grid such that the corresponding degree of freedom $\hat{df}(\lambda_1,\lambda_2)$ is closest to half of that from ``Result 1''. 

Table~\ref{tab1:application} summarizes some key features of the four solutions, including 
the degree of freedom, the total number of edges/connections over all time points, and the respective number of edges/connections during the first ($1 \leq t \leq 86$) and second ($87 \leq t \leq 172$) halves of the scanning period. Note that, for each set of solutions, the degrees of freedom and total number of edges/connections are similar between the healthy and ADHD groups, so it is reasonable and meaningful to compare them.

\begin{table}[htp]
\centering
\begin{tabular}{|c|cc|cc|}
\hline
\multirow{2}{*}{Generalized elastic net} & \multicolumn{2}{c|}{Result 1} & \multicolumn{2}{c|}{Result 2} \\ \cline{2-5} 
                  &    Healthy       &    ADHD       &    Healthy   &   ADHD    \\ \hline
degrees of freedom     &  685.59 & 748.21  &  292.53  &   312.04  \\ \hline
\# of connections   &  990  &   1073  &   1098   &  1171  \\ \hline
\shortstack{\# of connections during\\$(1^{st}, 2^{nd})$ half period}  &  (490,500)   & (269,804)    &  (541,557) &  (292,879) \\ \hline \hline
\multirow{2}{*}{Generalized fused LASSO} &   \multicolumn{2}{c|}{Result 1}          &    \multicolumn{2}{c|}{Result 2}          \\ \cline{2-5} 
 &  Healthy   &  ADHD  &  Healthy   & ADHD \\ \hline
degrees of freedom  &  937   & 853  &  563  &   504  \\ \hline
\# of connections   & 1606   & 1766 & 855  &  941 \\ \hline
\shortstack{\# of connections during\\$(1^{st}, 2^{nd})$ half period}   &  (780,826)   &  (509,1257)  &  (408,447) & (172,769)  \\ \hline
\end{tabular}
\caption{Summary of degrees of freedom and number of detected edges based on the GEN and the GFL with two pairs of selected tuning parameters. }
\label{tab1:application}
\end{table}

The main conclusions we can draw from all four sets of results---i.e., (Result 1, Result 2)$\times$(GEN, GFL)---turn out to be identical, which give us confidence in their scientific validity, to the extent justified by the quality of the data set itself. To reduce redundancy, therefore, we present only Result 1 from GFL in the main text; the other three sets of results are provided in \ref{subsec:supplement6} in the supplementary materials. 

First, Figure~\ref{fig:app_gfl9_connection} and, similarly, Figures~\ref{fig:app_gen11_connection}, \ref{fig:app_gen13_connection}, \ref{fig:app_gfl16_connection} in \ref{subsec:supplement6} show the frequency of connections between any two regions over all 172 time points of the entire scanning period. The left panel contains matrices where each entry represents the total number of occurrences for the corresponding connection, and these matrices are visualized as networks of eighteen cerebellar ROIs in the right panel, where the thickness of each edge is proportional to the number of occurrences for that connection. The most prominent observations here---from all four sets of results---are that (i) the connection between 7b\_L and 8\_L occurred only in the ADHD group but never appeared in the healthy group, and that (ii) the connection between 9\_L and 9\_R occurred a lot more often in the ADHD group than it did in the healthy group. 

Next, graphically displayed in Figure~\ref{fig:app_gfl9_table} here and, similarly, in Figures~\ref{fig:app_gen11_table}, \ref{fig:app_gen13_table}, \ref{fig:app_gfl16_table} in \ref{subsec:supplement6} are the estimated partial correlations at each time point. They provide further information as to when different interactions occur and how they change over time. 
A consistent observation here---again, from all four sets of results (also see Table~\ref{tab1:application})---is that, other than connections that persistently show up during the entire scanning period, the time points at which specific connections occur between ROIs are markedly different for the two groups. In particular, for the healthy group, there were more or less equal number of connections during the first and second halves of the scanning period; whereas, for the ADHD group, many more connections occurred during the second half than the first.




\begin{figure}[htp]
\centering
\begin{subfigure}[b]{\textwidth}
\centering
\includegraphics[angle=270,width=0.45\linewidth,valign=t]{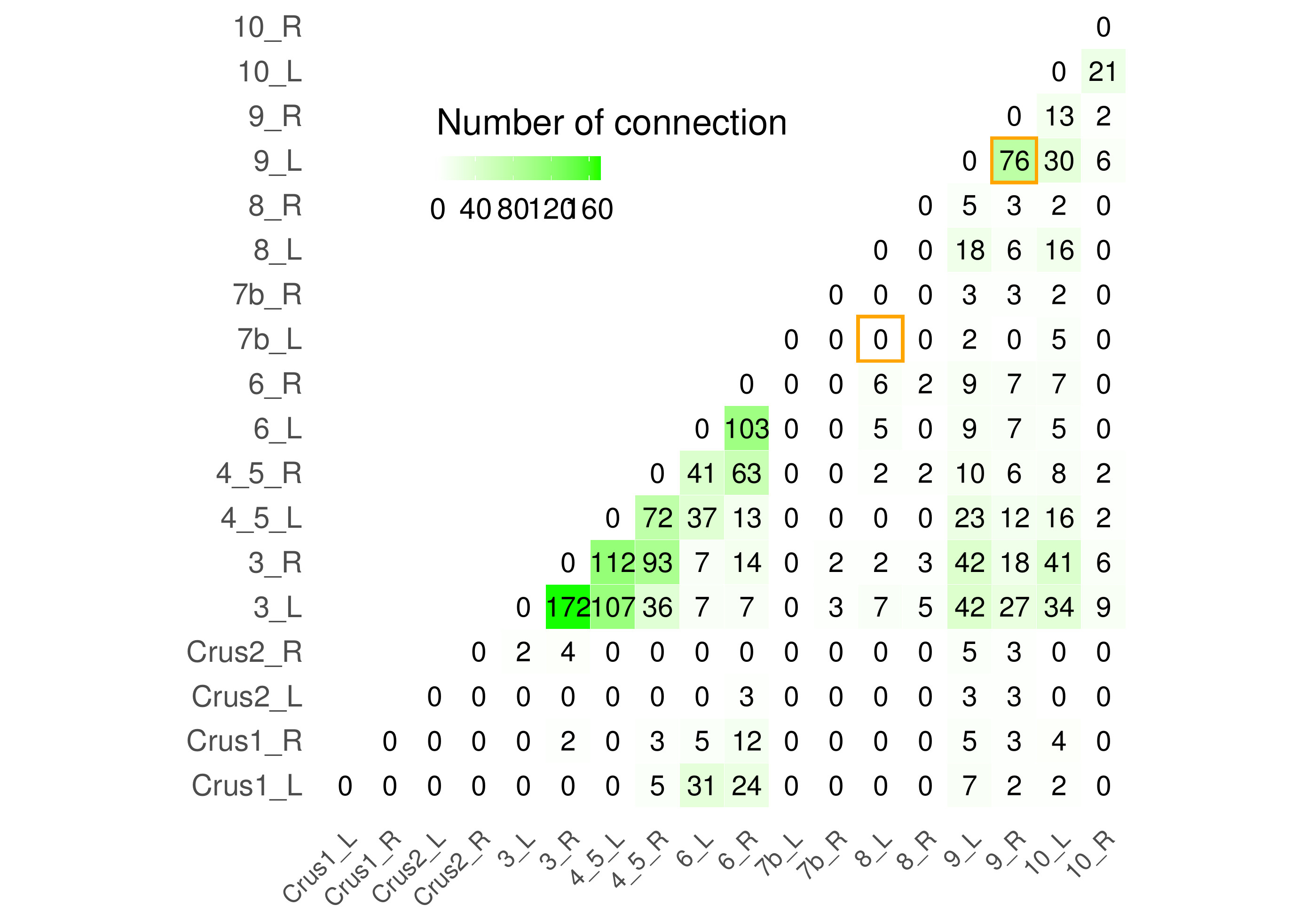}
\includegraphics[width=0.32\linewidth,valign=t]{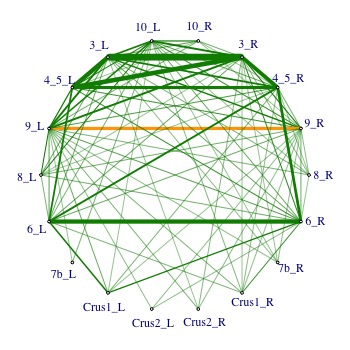}
\caption{Healthy group}
\end{subfigure}
\begin{subfigure}[b]{\textwidth}
\centering
\includegraphics[angle=270,width=0.45\linewidth,valign=t]{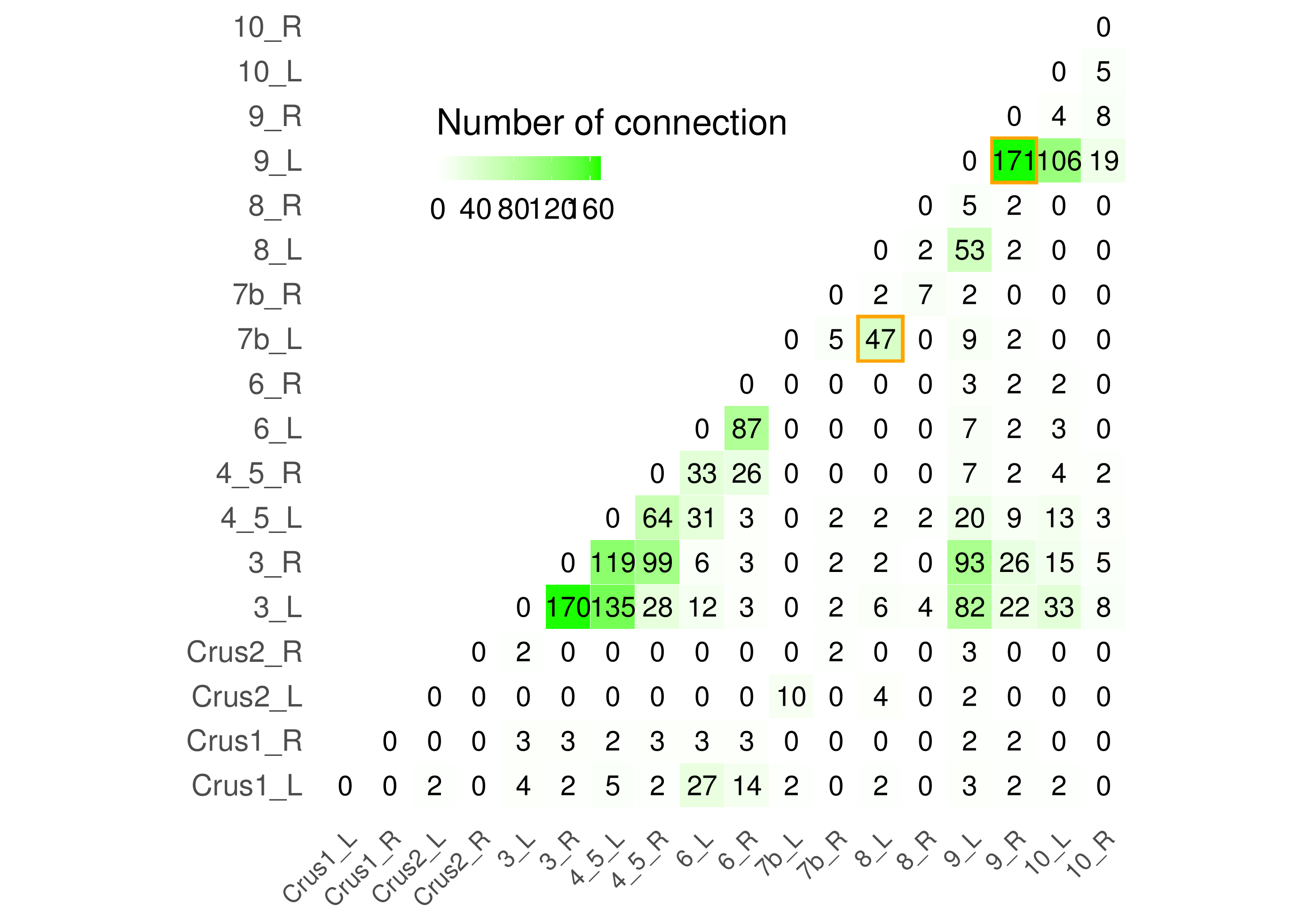}
\includegraphics[width=0.32\linewidth,valign=t]{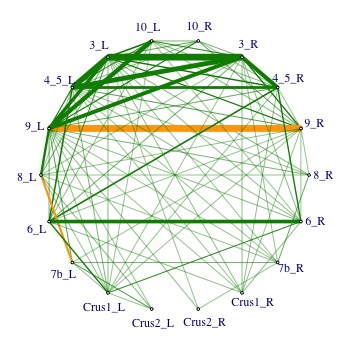}
\caption{ADHD group}
\end{subfigure}
\caption{[Result 1,GFL] Aggregated connections between different cerebellum regions over 172 time points based on the GFL. The yellow squares on the left highlight the number of the edge 7b\_L - 8\_L and the edge 9\_L - 9\_R.}
\label{fig:app_gfl9_connection}
\end{figure}

\begin{figure}[htp]%
\centering
\begin{subfigure}[b]{0.47\textwidth}
\centering
\includegraphics[width=\linewidth,valign=t]{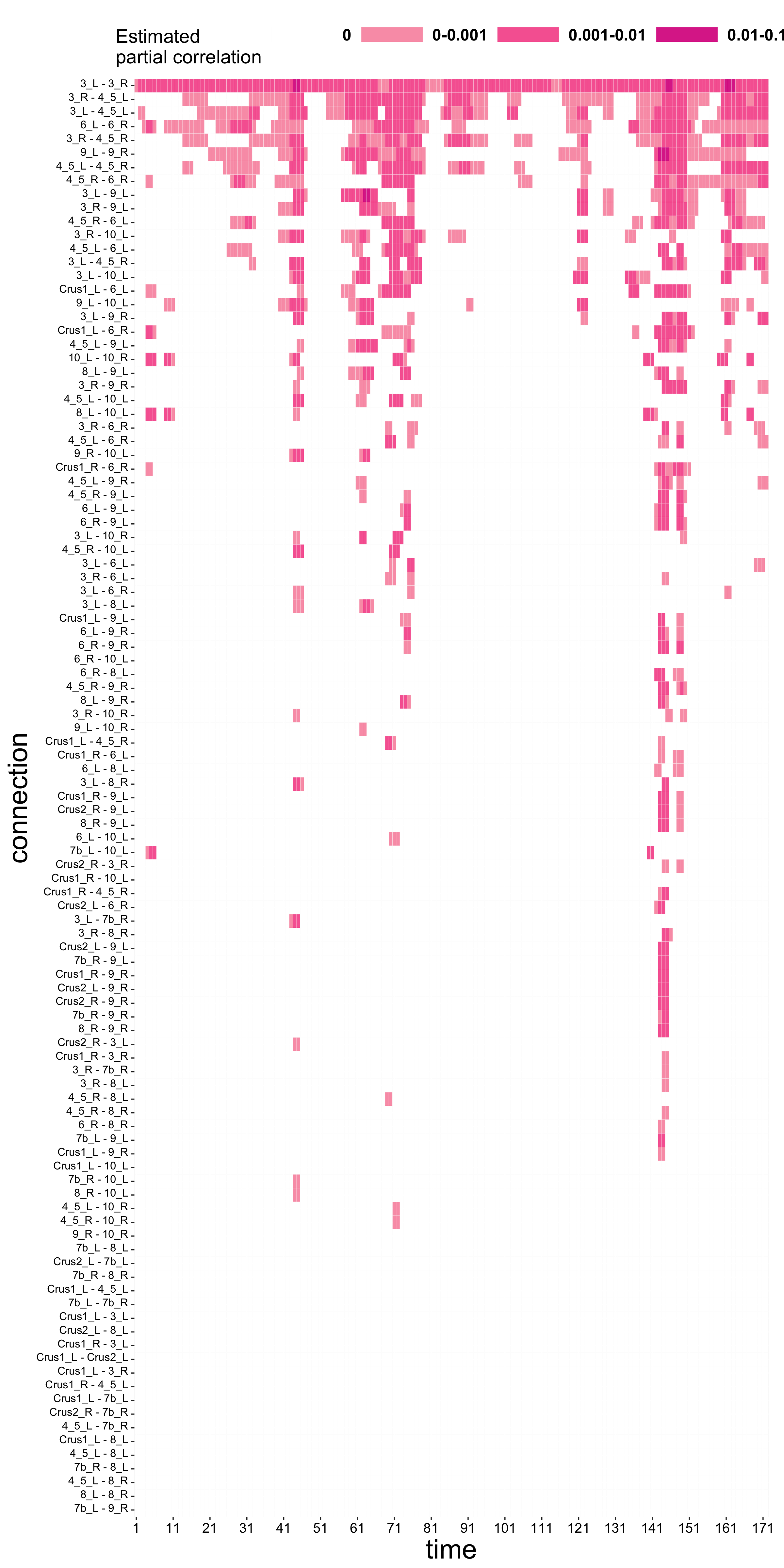}
\caption{Healthy group}
\end{subfigure}
~
\begin{subfigure}[b]{0.47\textwidth}
\centering
\includegraphics[width=\linewidth,valign=t]{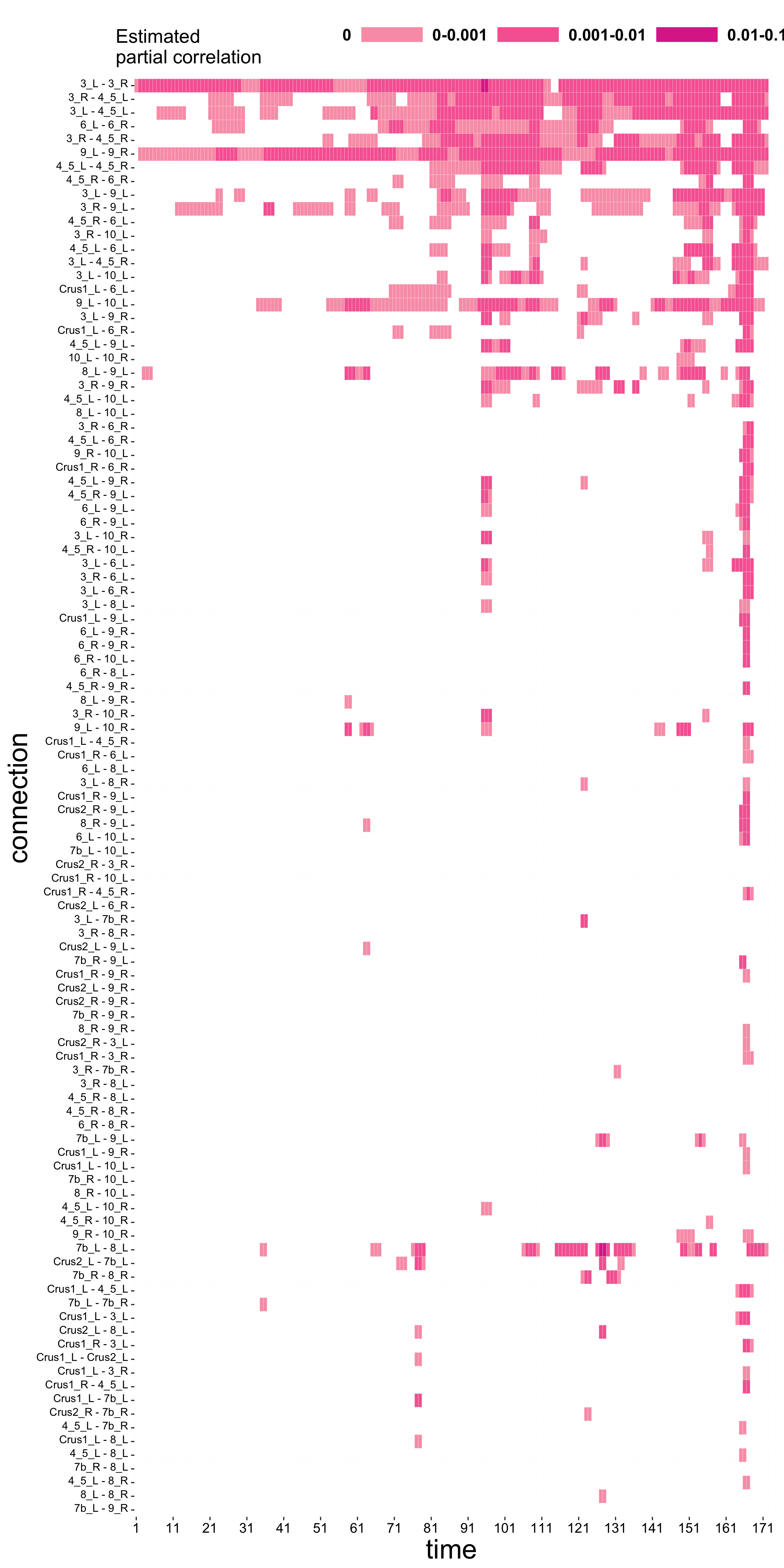}
\caption{ADHD group}
\end{subfigure}
\caption{[Result 1, GFL] Estimated partial correlations between different cerebellum regions over 172 time points based on the GFL. Each cell corresponds to a connection at a given time point, and the color represents the magnitude of the estimated partial correlation. }
\label{fig:app_gfl9_table}
\end{figure}

\section{Conclusion}
\label{sec:conc}

Although generalizing the elastic net and the fused LASSO as we have done in this paper are not the only ways to model time-varying network data, they are useful additions to the existing toolbox. While the idea of imposing $l_1$ and $l_2$ penalties on $\bm{D}\bm{\theta}$ may be quite straight-forward, the resulting optimization problems are not exactly trivial to solve. Some valuable lessons from our work are: first, the ADMM provides a unifying framework for solving both the GEN and the GLF problems, and one ``only'' has to parameterize the penalty functions differently for the two approaches; second, ``tricks'' that exploit special structures in otherwise large matrices are always useful.

The approximate degree-of-freedom formulae we derived are useful---and perhaps even sufficient as we have demonstrated---for practical purposes, but the correct degrees of freedom remain elusive and an open problem.

Finally, our analysis of the fMRI data set also leads to some interesting speculations. On the one hand, that the BIC-surfaces are quite flat over a wide range of $(\lambda_1,\lambda_2)$ values may be an indication that there is limited information in this data set. Indeed, scientists have questioned the usefulness of resting-state fMRI scans for studying ADHD \citep{lurie2020questions}. On the other hand, that the same conclusions can be drawn from four different solutions, with remarkably different degrees of freedom, is in itself a strong testament that these conclusions are probably not false discoveries. This is in line with the basic philosophy behind stability selection \citep{meinshausen2010stability}. If no scientific explanation is immediately available, we think they are at least genuine artefacts of this particular data set.

\bibliographystyle{chicago}
\bibliography{TVnet}

\newpage
\begin{center}	
\Large{\textbf{
Supplement to ``Two Gaussian regularization methods for time-varying networks"}} \\
Jie Jian, Peijun Sang, and Mu Zhu \\
Department of Statistics and Actuarial Science, University of Waterloo
\end{center}

	\setcounter{equation}{0}
	\setcounter{figure}{0}
	\setcounter{table}{0}
	\setcounter{page}{1}
	\makeatletter
	\renewcommand{\theequation}{S\arabic{equation}}
	\renewcommand{\thefigure}{S\arabic{figure}}
	\renewcommand{\thetable}{S\arabic{table}}

	\setcounter{section}{0}
	\renewcommand{\thesection}{S.\arabic{section}}

The supplementary material contains some technical details and additional results of numerical studies in the main manuscript.

\section{Fast computation of $\bm{\theta}-$update in ADMM }
\label{subsec:supplement1}
This section introduces the details of fast updating $\bm{\theta}$ in the ADMM algorithm discussed in Section 3.1. In the $\bm{\theta}-$update identified at Equation (15), $\mathbfcal{X}^{\top}  \mathbfcal{X}$ is a diagonal block matrix with the $k^{th}$ diagonal block $\mathbfcal{X}^{\top} (t_k)  \mathbfcal{X} (t_k)$,
and $\bm{D}^{\top} \bm{D} $ is a ${Tp(p-1)/2\times Tp(p-1)/2}$ block tridiagonal matrix given by
\begin{align*}
\bm{D}^{\top} \bm{D}=
\begin{bmatrix}
I & -I & 0 & 0 & 0 & 0 \\
-I & 2I & -I & 0 & 0  & 0 \\
0 & -I & 2I & -I & 0 & 0 \\
 &  &  & \ddots  &  &  \\
0 & 0 & 0 & -I & 2I & -I \\
0 & 0 & 0 & 0 & -I & I \\
\end{bmatrix},
\end{align*} 
where each block matrix has dimension $p(p-1)/2 \times p(p-1)/2$.

It follows that 
\begin{align}
 \dfrac{2}{n}  \mathbfcal{X}^{\top}  \mathbfcal{X} + 2\lambda_2 \bm{D}^{\top} \bm{D} + a I 
 &=  2\lambda_2
\begin{bmatrix}
A_1 & -I & 0 & 0 & 0 & 0 \\
-I & A_2 & -I & 0 & 0  & 0 \\
0 & -I & A_3 & -I & 0 & 0 \\
 &  &  & \ddots  &  &  \\
0 & 0 & 0 & -I & A_{T-1} & -I \\
0 & 0 & 0 & 0 & -I & A_T \\
\end{bmatrix}, \label{eq:app1}
\end{align} 
where $$
A_i=
\begin{cases}
\frac{\mathbfcal{X}(t_i)^{\top} \mathbfcal{X}(t_i)}{n\lambda_2}+(1+\frac{a}{2\lambda_2})I, \text{ if } i=1 \text{ and } T\\
\frac{\mathbfcal{X}(t_i)^{\top} \mathbfcal{X}(t_i)}{n\lambda_2}+(2+\frac{a}{2\lambda_2})I.~\text{otherwise}
\end{cases}.
$$
Denote by $H$ the matrix on the right-hand side of \eqref{eq:app1}. Updating $\bm{\theta}$ in (15) is to solve the linear system $H\bm{\theta} = \left( 2n^{-1}\mathbfcal{Y}^{\top} \mathbfcal{X} + a (\bm{z}- \bm{u} ) \right) $. We provide an efficient approach to find the inverse of $H$.

We first multiply $H$ by a sequence of lower triangular matrices, denoted by $L_i$'s, on its left to convert $H$ to an upper triangular matrix. In particular, we take $T = 4$, i.e., four time points, as an example to illustrate this procedure. Simple algebra yields the following three steps:
\begin{align*}
(1)\quad L_1 H &=
\begin{bmatrix}
I & 0 & 0 & 0 \\
A_1^{-1} & I & 0 & 0 \\
0 &0 & I & 0 \\
 0 & 0 & 0 & I\\
\end{bmatrix} 
\begin{bmatrix}
A_1 & -I & 0 & 0 \\
-I & A_2 & -I & 0 \\
0 & -I & A_3 & -I \\
 0 & 0 & -I & A_4 \\
\end{bmatrix} 
=
\begin{bmatrix}
A_1 & -I & 0 & 0 \\
0 & A_2-A_1^{-1} & -I & 0 \\
0 & -I & A_3 & -I \\
 0 & 0 & -I & A_4 \\
\end{bmatrix} .
\end{align*}

\begin{align*}
(2) \quad L_2 L_1 H & =
\begin{bmatrix}
I & 0 & 0 & 0 \\
0 & I & 0 & 0 \\
0 & (A_2-A_1^{-1})^{-1} & I & 0 \\
 0 & 0 & 0 & I\\
\end{bmatrix} 
\begin{bmatrix}
A_1 & -I & 0 & 0 \\
0 & A_2-A_1^{-1} & -I & 0 \\
0 & -I & A_3 & -I \\
 0 & 0 & -I & A_4 \\
\end{bmatrix} \\
& =
\begin{bmatrix}
A_1 & -I & 0 & 0 \\
0 & A_2-A_1^{-1} & -I & 0 \\
0 & 0 & A_3-(A_2-A_1^{-1})^{-1} & -I \\
 0 & 0 & -I & A_4 \\
\end{bmatrix} .
\end{align*}

\begin{align*}
(3) \quad L_3 L_2 L_1 H & =
\begin{bmatrix}
I & 0 & 0 & 0 \\
0 & I & 0 & 0 \\
0 & 0 & I & 0 \\
 0 & 0 & (A_3-(A_2-A_1^{-1})^{-1})^{-1} & I\\
\end{bmatrix} 
\begin{bmatrix}
A_1 & -I & 0 & 0 \\
0 & A_2-A_1^{-1} & -I & 0 \\
0 & 0 & A_3-(A_2-A_1^{-1})^{-1} & -I \\
 0 & 0 & -I & A_4 \\
\end{bmatrix}  \\
& =
\begin{bmatrix}
A_1 & -I & 0 & 0 \\
0 & A_2-A_1^{-1} & -I & 0 \\
0 & 0 & A_3-(A_2-A_1^{-1})^{-1} & -I \\
 0 & 0 & 0 & A_4-(A_3-(A_2-A_1^{-1})^{-1})^{-1} \\
\end{bmatrix} .
\end{align*}
Let $B_1=A_1^{-1}$ and $B_i=(A_i-B_{i-1})^{-1}$. The formula above can be rewritten as:
\begin{align*}
L_3 L_2 L_1 H & = 
\begin{bmatrix}
I & 0 & 0 & 0 \\
B_1 & I & 0 & 0 \\
0 &0 & I & 0 \\
 0 & 0 & 0 & I\\
\end{bmatrix} 
\begin{bmatrix}
I & 0 & 0 & 0 \\
0 & I & 0 & 0 \\
0 & B_2 & I & 0 \\
 0 & 0 & 0 & I\\
\end{bmatrix}
\begin{bmatrix}
I & 0 & 0 & 0 \\
0 & I & 0 & 0 \\
0 & 0 & I & 0 \\
 0 & 0 & B_3 & I\\
\end{bmatrix} H\\
&= 
\begin{bmatrix}
B_1^{-1} & -I & 0 & 0 \\
0 & B_2^{-1} & -I & 0 \\
0 & 0 & B_3^{-1} & -I \\
 0 & 0 & 0 & B_4^{-1}\\
\end{bmatrix}.
\end{align*}
To eliminate upper off-diagonal blocks, we define a sequence of upper triangular matrices: 
\begin{align*}
U_1  =
 \begin{bmatrix}
I & 0 & 0 & 0 \\
0 & I & 0 & 0 \\
0 & 0 & I &  B_4 \\
 0 & 0 & 0 & I \\
\end{bmatrix} ,
U_2  =
 \begin{bmatrix}
I & 0 & 0 & 0 \\
0 & I & B_3 & 0 \\
0 & 0 & I &  0 \\
 0 & 0 & 0 & I \\
\end{bmatrix} ,
U_3 =
 \begin{bmatrix}
I & B_2 & 0 & 0 \\
0 & I & 0 & 0 \\
0 & 0 & I &  0 \\
 0 & 0 & 0 & I \\
\end{bmatrix}.
\end{align*}
Sequentially multiplying by $U_i$'s on the left of $L_3 L_2 L_1 H $ yields
\begin{align*}
U_3 U_2 U_1 L_3 L_2 L_1 H 
 =
\begin{bmatrix}
B_1^{-1} & 0 & 0 & 0 \\
0 & B_2^{-1} & 0 & 0 \\
0 & 0 & B_3^{-1} & 0 \\
 0 & 0 & 0 & B_4^{-1}\\
\end{bmatrix}.
\end{align*}
Lastly, define $$\tilde{B}_i=\begin{bmatrix}
I & 0 & 0 & 0 & 0\\
0 & I & 0 & 0 & 0\\
0 & 0 & B_i & 0 & 0 \\
0 & 0 & 0 & I & 0 \\
 0 & 0 & 0 & 0 & I\\
\end{bmatrix}
$$ 
for $i = 1, 2, 3$. 
Then $\tilde{B}_4 \tilde{B}_3 \tilde{B}_2 \tilde{B}_1 U_3 U_2 U_1 L_3 L_2 L_1 H $ is an identity matrix. 

In summary, an updated $\boldsymbol{\theta}$ can be obtained through
\begin{align*}
\bm{\theta^{k+1}} &= \left(\frac{2}{n} \mathbfcal{X}^{\top} \mathbfcal{X} + \frac{2\lambda_2}{n} \bm{D}^{\top} \bm{D} + a I\right)^{-1} \left[\frac{2}{n}\mathbfcal{X}^{\top} \mathbfcal{Y}+\alpha (z^{k}-u^{k})\right]\\
& =\frac{n}{2\lambda_2} \cdot H^{-1} \cdot \left[\frac{2}{n} \mathbfcal{X}^{\top} \mathbfcal{Y} +a (z^{k}-u^{k})\right] \\
&=\tilde{B}_4 \tilde{B}_3 \tilde{B}_2 \tilde{B}_1 U_3 U_2 U_1 L_3 L_2 L_1\frac{n}{2\lambda_2} \cdot \left[\frac{2}{n}\mathbfcal{X}^{\top} \mathbfcal{Y}+a (z^{k}-u^{k})\right].
\end{align*}
This sequence of operations enables us to quickly find the inverse of $2n^{-1} \mathbfcal{X}^{\top} \mathbfcal{X} + 2\lambda_2n^{-1} \bm{D}^{\top} \bm{D} + a I$, thus the computational efficiency of the ADMM algorithm is greatly enhanced. 

\section{Degrees of freedom in GEN}
\label{subsec:supplement2}
In this section, we derive the degrees of freedom in GEN under the framework of Stein's unbiased risk estimation (SURE) \citep{stein1981estimation}. \cite{zou2007degrees} provides a theoretical justification of the degrees of freedom in the standard LASSO problem, and we will use it to derive an unbiased estimate of the degrees of freedom in the GEN problem with a homoscedastic assumption. We state the main theorem as the following.
\begin{theorem} \label{propGENdf}
Suppose $\bm{y}_{n \times 1} \sim N(\bm{\mu},\sigma^2 \bm{I})$, where $\bm{\mu} \in \real^{n}$ denotes the mean vector and $\sigma^2$ denotes the common variance of each component. Given a design matrix $\bm{X}\in \mathbb{R}^{n\times p}$ and two tuning parameters, $\lambda_1$ and $\lambda_2$, we consider the generalized elastic net problem
\begin{align}
\hat{\bm{\beta}} = \argmin\limits_{\bm{\beta}\in  \mathbb{R}^{ p}} \left\{\frac{1}{n}  \| \bm{y} -    \bm{X} \bm{\beta} \| ^2
 + \lambda_1 \| \bm{\beta}  \|_1 + \lambda_2 \| \bm{D} \bm{\beta } \|_2^2\right\},  \label{generalGEN}
\end{align}
where $\bm{D}_{m\times p}$ is defined as in Equation \eqref{Dmat}. If definition of degrees of freedom is given by Equation \eqref{df}, then
\begin{align*}
df( \bm{X}\hat{\bm{\beta}})=Tr\left( \bm{X}_{\mathcal{A}}  \left(  \bm{X}_{\mathcal{A}}^{\top}  \bm{X}_{\mathcal{A}} + n\lambda_2 \bm{D}_{\mathcal{A}}^{\top} \bm{D}_{\mathcal{A}}\right)^{-1} \bm{X}_{\mathcal{A}}^{\top}\right) . 
\end{align*}
Here $\mathcal{A}=\{j: \hat{\bm{\beta}}_j \neq 0 \} $ denotes the collection of column indices of $\bm{X}$ corresponding to the active features, and  $B_{\mathcal{A}}$ denotes a submatrix of $B$ that contains only the columns indexed by $\mathcal{A}$.
\end{theorem}
To prove Theorem \ref{propGENdf}, we first introduce  definition of effective degrees of freedom for a general fitted function as described in  Lemma~\ref{theoEfrom}. Furthermore, Lemma~\ref{steinLemma} indicates that if the fitted function is almost differentiable, the effective degrees of freedom can be simplified as the gradient of the fitted function. In Lemma~\ref{genSolution}, we give an explicit form of the fitted function in GEN. Lemma~\ref{ADiff} shows that the fitted function of GEN is uniformly Lipschitz, and thus by Lemma~\ref{LCandDiff} the fitted function in GEN is almost differentiable.      

One natural definition of effective degrees of freedom comes from the well-known identity of \textit{optimism} in \cite{efron1986biased}.

\begin{lemma}[Optimism theorem \citep{efron1986biased}] \label{theoEfrom}
Suppose $\bm{y}_{n\times 1} \sim (\bm{\mu},\sigma^2 \bm{I})$, where $\bm{\mu}$ is the true mean vector and $\sigma^2$ is the common variance of each component. Let $\hat{\bm{\mu}}=\delta(\bm{y})$ denote the fitted function of some fitting technique $\delta$, and $\bm{y}^{new}$ be the new response vector generated from the distribution $(\bm{\mu},\sigma^2 \bm{I})$. Then
\begin{align}
\E\left\{  \|  \bm{y}^{new} - \hat{\bm{\mu}} \|^2 \right\} - E\left\{  \| \bm{y}- \hat{\bm{\mu}}  \|^2 \right\}=2\sum\limits_{i=1}^n \cov\left( y_i,\hat{\mu}_i \right). \label{efron}
\end{align}
\end{lemma}

The right-hand side of \eqref{efron} is referred to as {\it the optimism} of the estimator $\hat{\bm{\mu}}$. Based on \eqref{efron}, the degrees of freedom can be defined as
\begin{align}
df(\hat{\bm{\mu}})=\dfrac{1}{\sigma^2} \sum\limits_{i=1}^n \cov\left( y_i,\hat{\mu}_i \right). \label{df}
\end{align}
Stein's Lemma \citep{stein1981estimation} can further simplify the right-hand side of \eqref{df}.

\begin{lemma}[Stein's Lemma] \label{steinLemma}
Suppose that $\hat{\bm{\mu}}: \mathbb{R}^n \rightarrow \mathbb{R}^n $ is almost differentiable and let $\nabla \cdot \hat{\bm{\mu}} = \sum_{i=1}^n \dfrac{\partial \hat{\mu}_i}{\partial y_i} $. If $\bm{y} \sim N(\bm{\mu},\sigma^2 I )$, then
\begin{align}
\dfrac{1}{\sigma^2} \sum\limits_{i=1}^n \cov\left( y_i,\hat{\mu}_i \right)= \E(\nabla \cdot \hat{\bm{\mu}}). \label{steinDf}
\end{align}
\end{lemma}

Therefore, \eqref{df} and \eqref{steinDf} imply that
\begin{align}
\hat{df}(\bm{\mu})=\nabla \cdot \hat{\bm{\mu}}
\label{steinDf2}
\end{align}
is an unbiased estimate of the degrees of freedom if $\hat{\bm{\mu}}$ is almost differentiable. Next, we will first find the fitted function $\hat{\bm{\mu}}$ in the GEN problem \eqref{generalGEN}, and then show it is almost differentiable. Lastly, we find an explicit form of $\nabla \cdot \hat{\bm{\mu}}$ as the unbiased estimator of the (effective) degrees of freedom.

\begin{lemma}[Solution to GEN] \label{genSolution}
The GEN problem \eqref{generalGEN} can be written as a lasso-type problem with the augmented dataset $
\tilde{\bm{y}}=\begin{bmatrix}
\bm{y}_{n\times1}\\
\bm{0}_{m\times 1}
\end{bmatrix}
$
and 
$
\tilde{\bm{X}}=\begin{bmatrix}
\bm{X}_{n\times p}\\
\sqrt{n\lambda_2}\cdot \bm{D}_{m\times p}
\end{bmatrix}
$:
\begin{align*}
\hat{\bm{\beta}} = \min\limits_{\bm{\beta}} \frac{1}{n}  \| \tilde{\bm{y}} - \tilde{\bm{X}} \bm{\beta} \| ^2
 + \lambda_1 \| \bm{\beta}  \|_1 . 
\end{align*}
Suppose that $\lambda_1$ and $\lambda_2$ are not the transition points where the active set $\mathcal{A}$ changes. The coefficient estimate is given by
\begin{align}
\hat{\bm{\beta}}_{\lambda_1,\lambda_2} & =  \left( \tilde{\bm{X}}^{\top}_{\mathcal{A}}  \tilde{\bm{X}}_{\mathcal{A}} \right)^{-1} \left( \tilde{\bm{X}}^{\top}_{\mathcal{A}} \tilde{\bm{y}} -\dfrac{\lambda_1}{2} \sign(\hat{\bm{\beta}}_{\mathcal{A}})  \right) \nonumber \\
&= \left( \bm{X}^{\top}_{\mathcal{A}} \bm{X}_{\mathcal{A}} + n\lambda_2 \bm{D}^{\top}_{\mathcal{A}} \bm{D}_{\mathcal{A}} \right)^{-1} \left( \bm{X}^{\top}_{\mathcal{A}} \bm{y} -\dfrac{\lambda_1}{2} \sign(\hat{\bm{\beta}}_{\mathcal{A}})    \right),
\label{GGENsolution}
\end{align}
and the fitted function is
\begin{align}
\hat{\bm{\mu}} ( \bm{y}) & =\bm{X}_{\mathcal{A}} \hat{\bm{\beta}}_{\lambda_1,\lambda_2} \nonumber \\
&=\bm{X}_{\mathcal{A}} \left( \bm{X}^{\top}_{\mathcal{A}} \bm{X}_{\mathcal{A}} + n\lambda_2 \bm{D}^{\top}_{\mathcal{A}} \bm{D}_{\mathcal{A}} \right)^{-1} \left( \bm{X}^{\top}_{\mathcal{A}} \bm{y} -\dfrac{\lambda_1}{2} \sign(\hat{\bm{\beta}}_{\mathcal{A}})    \right).  \label{GGENfit}
\end{align}
\end{lemma}

\begin{lemma}[Lipschitz continuity of $\hat{\bm{\mu}}$] \label{ADiff}
The GEN fitted function $\hat{\bm{\mu}}$ in \eqref{GGENfit} is 1-Lipschitz i.e., $\| \hat{\bm{\mu}}(\bm{y}+\Delta \bm{y}) - \hat{\bm{\mu}}(\bm{y})  \| \leq  \|  \Delta \bm{y} \| $ for sufficiently small $\Delta \bm{y}$.
\end{lemma}

\begin{proof}
Define a mapping $\tau: \mathbb{R}^{n} \rightarrow  \mathbb{R}^{n+m}$ such that 
\begin{align}
\tau (\bm{y}) 
&= \tilde{\bm{X}}_{\mathcal{A}} \hat{\bm{\beta}}_{\lambda_1,\lambda_2}  \nonumber \\
&=  \tilde{\bm{X}}_{\mathcal{A}}  \left( \tilde{\bm{X}}^{\top}_{\mathcal{A}}  \tilde{\bm{X}}_{\mathcal{A}} \right)^{-1} \left( \tilde{\bm{X}}^{\top}_{\mathcal{A}} \tilde{\bm{y}} -\dfrac{\lambda_1}{2} \sign(\hat{\bm{\beta}}_{\mathcal{A}})  \right) ,
\end{align}
where $\tilde{\bm{X}}_{\mathcal{A}}$ and $\tilde{\bm{y}}$ have been defined in Lemma~\ref{genSolution}. Then $\hat{\bm{\mu}} ( \bm{y})$ is the first $n$ components in $\tau (\bm{y}) $, i.e., $\hat{\bm{\mu}} ( \bm{y}) = T \cdot \tau (\bm{y})$ where $T:=[\bm{I}_{n\times n}, \bm{0}_{n\times m}]$.

Since $\lambda_1$ and $\lambda_2$ are not the transition points, the active set $\mathcal{A}$ stays constant for sufficient small $\Delta y$. Then we have
\begin{align}
\| \tau(\bm{y}+\Delta \bm{y}) - \tau(\bm{y})  \| & =\| \tilde{\bm{X}}_{\mathcal{A}}  \left( \tilde{\bm{X}}^{\top}_{\mathcal{A}}  \tilde{\bm{X}}_{\mathcal{A}} \right)^{-1} \left( \tilde{\bm{X}}^{\top}_{\mathcal{A}} \tilde{\bm{y}} -\dfrac{\lambda_1}{2} \sign(\hat{\bm{\beta}}_{\mathcal{A}})  \right) \nonumber \\ 
&-\tilde{\bm{X}}_{\mathcal{A}}  \left( \tilde{\bm{X}}^{\top}_{\mathcal{A}}  \tilde{\bm{X}}_{\mathcal{A}} \right)^{-1} \left( \tilde{\bm{X}}^{\top}_{\mathcal{A}} (\bm{y}+\Delta \bm{y}) -\dfrac{\lambda_1}{2} \sign(\hat{\bm{\beta}}_{\mathcal{A}})  \right)  \| \nonumber \\
&= \| \tilde{\bm{X}}_{\mathcal{A}}  \left( \tilde{\bm{X}}^{\top}_{\mathcal{A}}  \tilde{\bm{X}}_{\mathcal{A}} \right)^{-1} \tilde{\bm{X}}^{\top}_{\mathcal{A}} \Delta \bm{y} \| \nonumber \\
&\leq \| \Delta \bm{y} \| \label{CL2},
\end{align}
where the last relation holds since $ \tilde{\bm{X}}_{\mathcal{A}}  \left( \tilde{\bm{X}}^{\top}_{\mathcal{A}}  \tilde{\bm{X}}_{\mathcal{A}} \right)^{-1} \tilde{\bm{X}}^{\top}_{\mathcal{A}}$ is a projection matrix.
Therefore, $\tau$ is Lipschitz continuous. 

 Now, we can show that $\hat{\bm{\mu}}$ is also Lipschitz continuous:
\begin{align*}
\| \hat{\mu}(\bm{y}+\Delta \bm{y}) - \hat{\mu}(\bm{y})  \| & = \| T \cdot \left( \hat{\tau}(\bm{y}+\Delta \bm{y}) - \hat{\tau}(\bm{y}) \right) \| \\
&\leq \| T  \| \cdot \| \Delta \bm{y} \|  \\
& =  \| \Delta \bm{y} \|.
\end{align*}
\end{proof}

\begin{lemma}[Lipschitz continuity and differentiability \citep{meyer2000degrees}] \label{LCandDiff}
\textit{Any Lipschitz continuous function is almost differentiable. }
\end{lemma}

By local constancy of the active set $\mathcal{A}$ and Equation \eqref{GGENfit}, it is straighforward to show
\begin{align}
\nabla \cdot \hat{\bm{\mu}} = \sum\limits_{i=1}^n \dfrac{\partial \hat{\mu}_i(\bm{y})}{\partial \bm{y}_i} = Tr[\bm{X}_{\mathcal{A}} \left( \bm{X}^{\top}_{\mathcal{A}} \bm{X}_{\mathcal{A}} + n\lambda_2 \bm{D}^{\top}_{\mathcal{A}} \bm{D}_{\mathcal{A}} \right)^{-1} \bm{X}_{\mathcal{A}}^{\top} ]. \label{gradient}
\end{align}
The above results suggest that it is an unbiased estimator of degrees of freedom of $\hat{\bm{\mu}}$.
Therefore, Theorem~\ref{propGENdf} is established. 

\section{Degrees of freedom in GFL}
\label{subsec:supplement3}

\cite{tibshirani2012degrees} shows that the nullity of a particular penalty matrix is an unbiased estimator
of the degrees of freedom defined in Equation \eqref{df}; 
see Lemma  \ref{rtib_generallasso}. In this section, we provide an explicit form for this quantity in the setting of GFL. 

\begin{lemma}[\cite{tibshirani2012degrees}] \label{rtib_generallasso}
Suppose $\bm{y}_{n \times 1} \sim N(\bm{0},\sigma^2 \bm{I})$, where $\bm{\mu}$ is the mean vector and $\sigma^2$ is the common variance of each component. Given a design matrix $\bm{X}\in \mathbb{R}^{n\times p}$ of full column rank and a tuning parameter $\lambda$, we consider the generalized lasso problem
\begin{align}
\hat{\bm{\beta}} = \argmin\limits_{\bm{\beta}\in  \mathbb{R}^{ p}} \left\{\frac{1}{n}  \| \bm{y} -    \bm{X} \bm{\beta} \| ^2
 + \lambda \| \bm{F} \bm{\beta}  \|_1\right\},  \label{generallasso}
\end{align}
where $\bm{F}_{m\times p}$ is an arbitrary penalty matrix. Then
\begin{align}
df( \bm{X}\hat{\bm{\beta}})= \E[ nullity(\bm{F}_{-\mathcal{A}}) ], \label{Glassodf}.
\end{align}
\textit{where $\mathcal{A}=\{ i\in \{ 1,\cdots, m \}: (\bm{F} \hat{\bm{\beta}} )_i \neq 0 \}$.}
\end{lemma}

In GFL, the penalty matrix $\lambda\bm{F}$ is composed of two parts: the top part is a difference matrix defined in \eqref{Dmat} multiplied by the first tuning parameter $\lambda_1$ while the lower part is an identity matrix multiplied by another tuning parameter $\lambda_2$. Theorem \ref{gflDfT} establishes a simple expression of \eqref{Glassodf} in GFL. As row operations preserve the null space, in the following calculations the tuning parameters are not considered for the penalty matrix for simplicity. Let $\beta=\left( \beta(1)^{\top},\cdots,\beta(T)^{\top} \right)^{\top}$ where each $\beta(k)$ is a column vector of length $p$, and $\bm{F}$ be a $p(2T-1)$-by-$pT$ matrix, where the first $p(T-1)$ rows constitute a difference matrix and the second $pT$ rows constitute an identity matrix. We rewrite $\bm{F}$ as a block matrix using the $p\times p$ identity matrix $I$:  
\begin{align*}
\bm{F}=
\begin{bmatrix}
I & -I & 0 & 0 & 0 & 0 \\
0 & I & -I & 0 & 0  & 0 \\
 &  & \ddots  & \ddots  &  &  \\
0 & 0 & 0 & I & -I & 0 \\
0 & 0 & 0 & 0 & I & -I \\ \hdashline
I & 0 & 0 & 0 & 0 & 0 \\
0 & I & 0 & 0 & 0  & 0 \\
 &  &  & \ddots  &  &  \\
0 & 0 & 0 & 0 & I & 0 \\
0 & 0 & 0 & 0 & 0 & I \\
\end{bmatrix}.
\end{align*}

\begin{theorem}\label{gflDfT}

Let $\beta_j(k)$ denote the $j$th element of $\beta(k)$, $j = 1, \ldots, p, k = 1, \ldots, T$. If $\bm{F}_{-\mathcal{A}}$ denotes the submatrix of $\bm{F}$ after removing the rows indexed by $\mathcal{A}=\{ i\in \{ 1,\cdots, p(2T-1) \}: (\bm{F} \hat{\bm{\beta}} )_i \neq 0 \}$, then 
\begin{align}
nullity(\bm{F}_{-\mathcal{A}}) =\# \text{fused group}, \label{nullityFG}
\end{align}
where 
\begin{align}
\# \text{fused group}=  \sum\limits_{j=1}^{p}  \left[   \mathbbm{1}\{\hat{\beta}_{j}(1)\neq 0\} + \sum\limits_{k=2}^T \mathbbm{1}\{ \hat{\beta}_{j}(k)\neq \hat{\beta}_{j}(k-1),\ \hat{\beta}_{j}(k)\neq0\} \right]. \label{GFLdfF}
\end{align}
\end{theorem}

\begin{proof}
The right-hand side of \eqref{GFLdfF} can be written as 
\begin{align*}
 \sum\limits_{j=1}^{p}  \left[ T- \sum\limits_{k=1}^T \mathbbm{1}\{\hat{\beta}_{j}(k)=0\} - \sum\limits_{k=2}^T \mathbbm{1}\{ \hat{\beta}_{j}(k)=\hat{\beta}_{j}(k-1),\ \hat{\beta}_{j}(k)\neq0\} \right].
\end{align*}
Since $nullity(\bm{F}_{-\mathcal{A}})= pT-rank(\bm{F}_{-\mathcal{A}})$, we only need to show
\begin{align}
rank(\bm{F}_{-\mathcal{A}}) =  \sum\limits_{j=1}^{p}  \left[  \sum\limits_{k=1}^T \mathbbm{1}\{\hat{\beta}_{j}(k)=0\} + \sum\limits_{k=2}^T \mathbbm{1}\{ \hat{\beta}_{j}(k)=\hat{\beta}_{j}(k-1),\ \hat{\beta}_{j}(k)\neq0\} \right]. \label{rank}
\end{align}
To calculate the rank of $\bm{F}_{-\mathcal{A}}$, we count the maximum number of its independent rows. 

If $\hat{\beta}_{j}(k)=0$, then the $\left( (k-1)p + j \right)$th row in the lower identity matrix of $\bm{F}$ is preserved in $\bm{F}_{-\mathcal{A}}$, which serves as an independent row as this row with only one component $1$ can eliminate any other non-zero entries in the same column. 

When $\hat{\beta}_{j}(k)\neq 0$, the corresponding row in the lower identity matrix of $\bm{F}$ is removed in $\bm{F}_{-\mathcal{A}}$, which also takes away its only non-zero component $1$ in the $\left( (k-1)p +j \right)$th place (column). Whether there exists any other non-zero entries in the $\left( (k-1)p +j \right)$th column depends on the relations of the pairs $\left(\hat{\beta}_{j}(k-1),\hat{\beta}_{j}(k) \right)$ and $\left( \hat{\beta}_{j}(k),  \hat{\beta}_{j}(k+1) \right)$.
If $\hat{\beta}_{j}(k)=\hat{\beta}_{j}(k-1)$, the $\left( (k-2)p +j \right)$th row in the top difference matrix of $\bm{F}$ stays in $\bm{F}_{-\mathcal{A}}$ providing a $-1$ in the $\left( (k-1)p + j \right)$th column. We will count this row as an independent row.
Because even if $\hat{\beta}_{j}(k)=\hat{\beta}_{j}(k+1)$ which keeps the $\left( (k-1)p+j \right)$th row of $\bm{F}$ in $\bm{F}_{-\mathcal{A}}$, this row will be counted as an independent row when we consider the case of $\hat{\beta}_{j}(k+1)$ where $ \hat{\beta}_{j}(k+1)\neq 0 \text{ and }\hat{\beta}_{j}(k+1)=\hat{\beta}_{j}(k)$.

Hence, the maximum number of the independent rows in $\bm{F}_{-\mathcal{A}}$ is 
$$\sum_{i=1}^{p}  \left[  \sum_{k=1}^T \mathbbm{1}\{\hat{\beta}_{j}(k)=0\} + \sum_{k=2}^T \mathbbm{1}\{ \hat{\beta}_{j}(k)=\hat{\beta}_{j}(k-1),\ \beta_{j}(k)\neq0\} \right].
$$ Therefore, \eqref{rank} holds. Proof is completed. 
\end{proof}

\newpage 
\section{Covariance matrices in Scenario 1 }
\label{subsec:supplement0}

In this section, we illustrate how to select covariance matrices in the first scenario of the simulation. We take the number of basis functions as $S=13$. For the sake of identifiability, we design the covariance structure (or precision matrix) of the random coefficient vector corresponding to each B-spline basis function as follows.

The covariance matrix is rewritten as $\bm{\Sigma}_s = ( \bm{\Sigma}_s^{11} , \bm{\Sigma}_s^{12} ; \bm{\Sigma}_s^{21} , \bm{\Sigma}_s^{22} )$, where the four block submatrice are all diagonal. Under this design, the non-zero entries in the true precision are at the same positions as in $\sum_{s=1}^{S} B_s^2(t) \bm{\Sigma}_s$. Figure \ref{fig:heat map} depicts the heat maps of the 13 covariance matrices. Figure \ref{fig:bspline} display thirteen cubic B-spline basis functions, which indicates that each basis function is nonzero in several subintervals of [0, 1]. 

\begin{figure}[h]
     \centering
     \begin{subfigure}[b]{0.4\textwidth}
         \centering
         \includegraphics[width=1.4\textwidth]{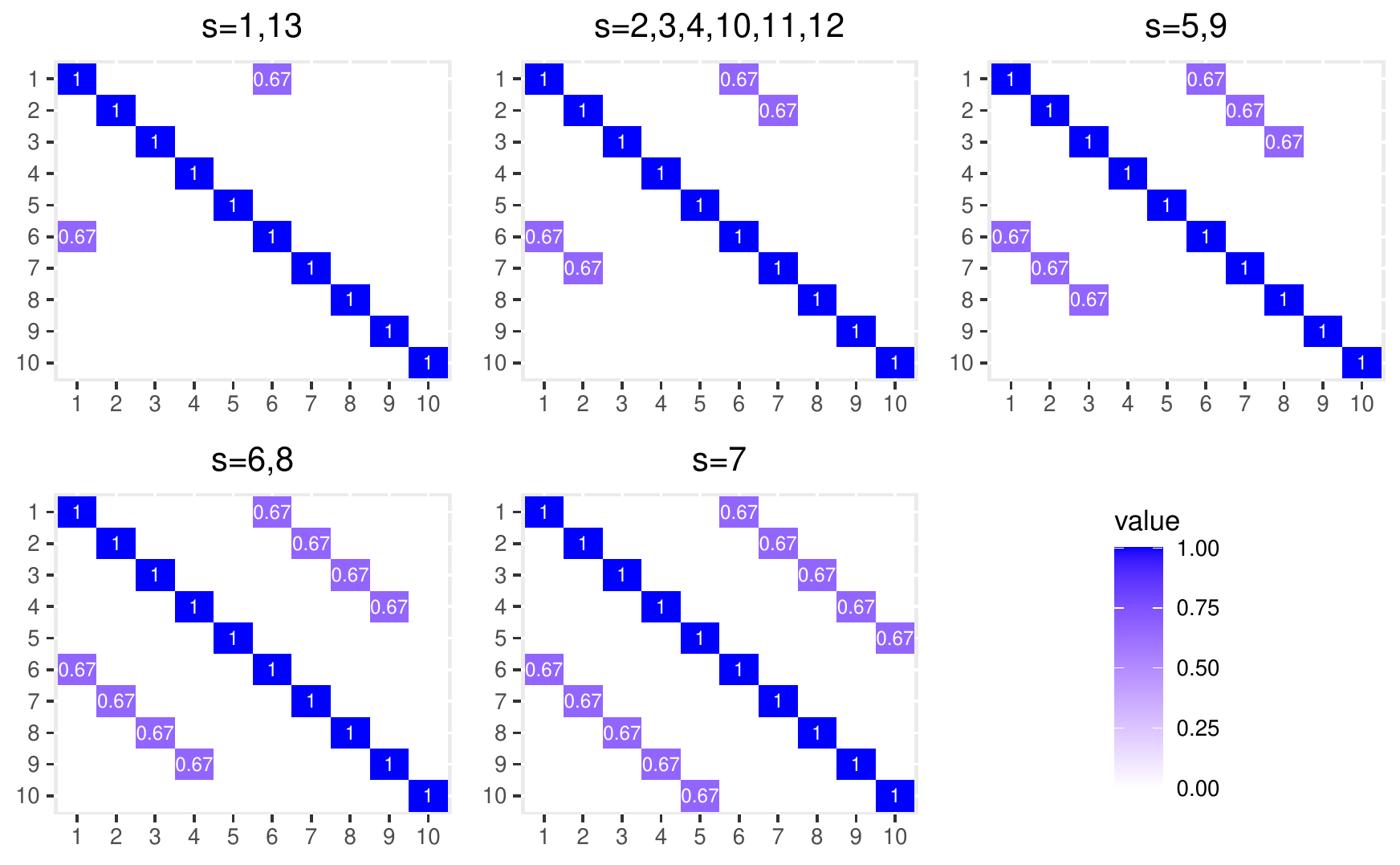}
        \caption{}
         \label{fig:heat map}
     \end{subfigure}
     \hfill
     \begin{subfigure}[b]{0.4\textwidth}
         \centering
         \includegraphics[width=\textwidth]{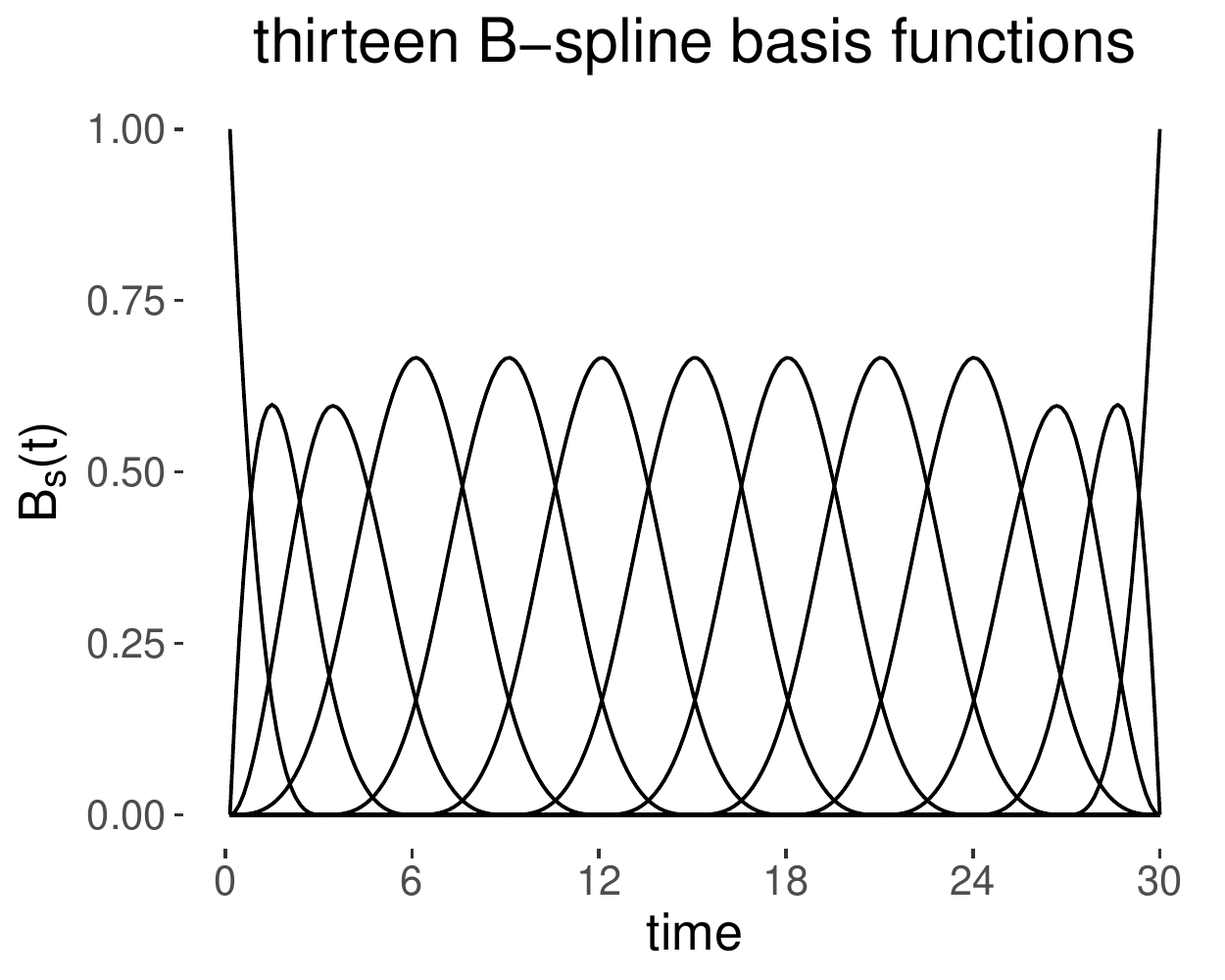}
         \caption{}
         \label{fig:bspline}
     \end{subfigure}
        \caption{(a) Heat maps of the pre-specified covariance matrices $\{ \bm{\Sigma}_s \}_{s=1}^{13}$ corresponding to 13 cubic basis functions.(b) Thirteen B-spline basis functions $\{ B_s (t) \}_{s=1}^{13}$ defined on $ [0,1]$, which is divided into 30 subintervals of equal length.} 
\end{figure}

\newpage

\section{Additional simulation results}
\label{subsec:supplement5}

\begin{figure}[H]
\centering
\begin{subfigure}[b]{\textwidth}
\centering
\includegraphics[width=0.18\linewidth,valign=t]{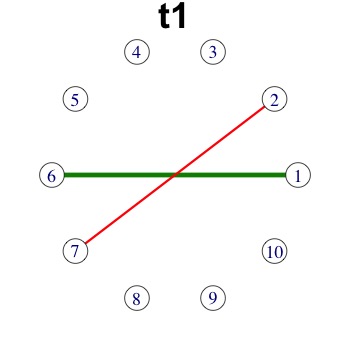}
\includegraphics[width=0.18\linewidth,valign=t]{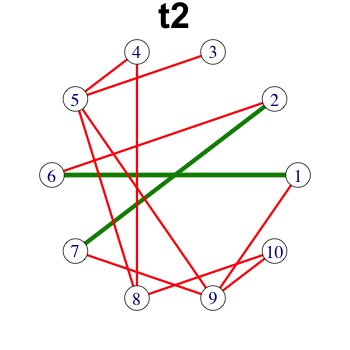}
\includegraphics[width=0.18\linewidth,valign=t]{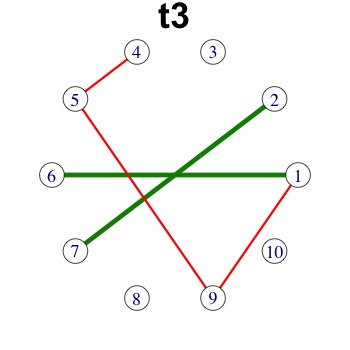}
\includegraphics[width=0.18\linewidth,valign=t]{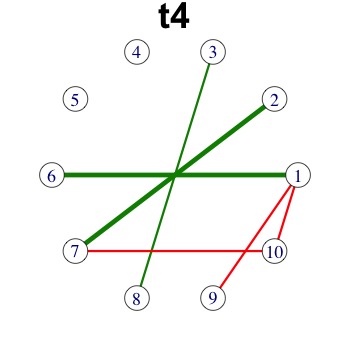}
\includegraphics[width=0.18\linewidth,valign=t]{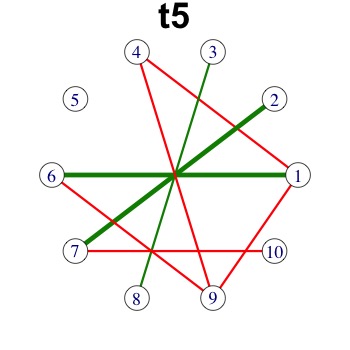}
\includegraphics[width=0.18\linewidth,valign=t]{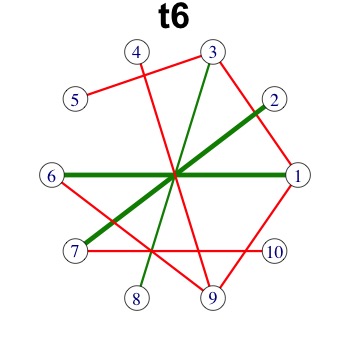}
\includegraphics[width=0.18\linewidth,valign=t]{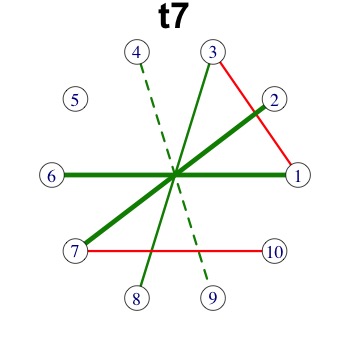}
\includegraphics[width=0.18\linewidth,valign=t]{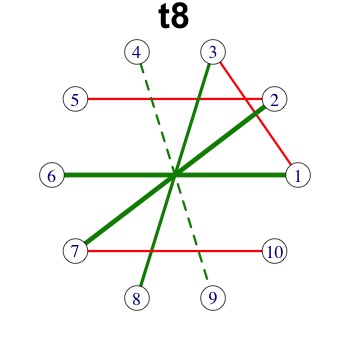}
\includegraphics[width=0.18\linewidth,valign=t]{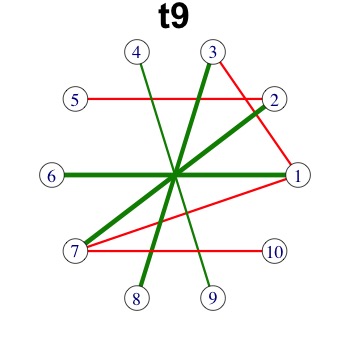}
\includegraphics[width=0.18\linewidth,valign=t]{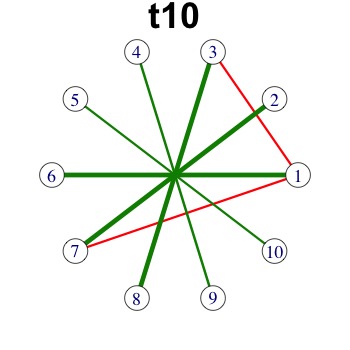}
\includegraphics[width=0.18\linewidth,valign=t]{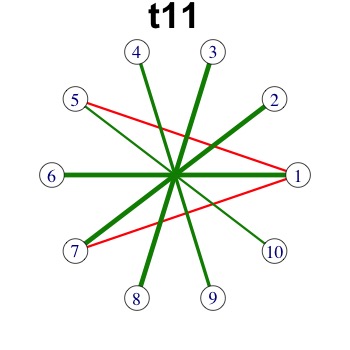}
\includegraphics[width=0.18\linewidth,valign=t]{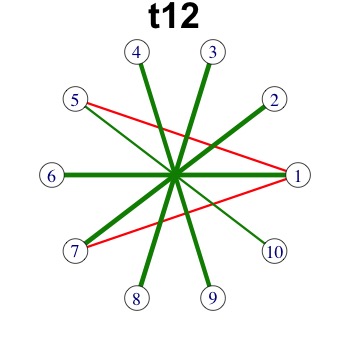}
\includegraphics[width=0.18\linewidth,valign=t]{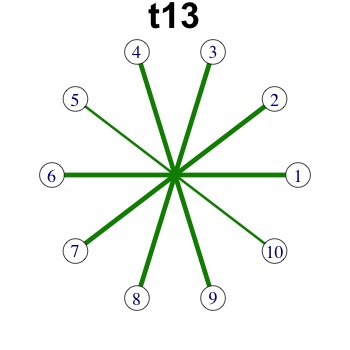}
\includegraphics[width=0.18\linewidth,valign=t]{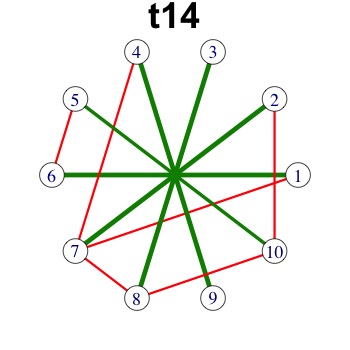}
\includegraphics[width=0.18\linewidth,valign=t]{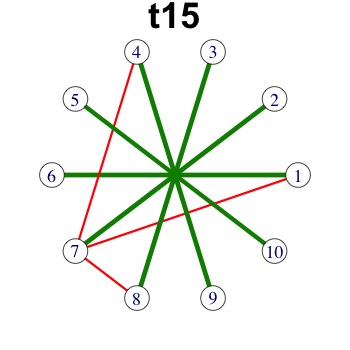}
\includegraphics[width=0.18\linewidth,valign=t]{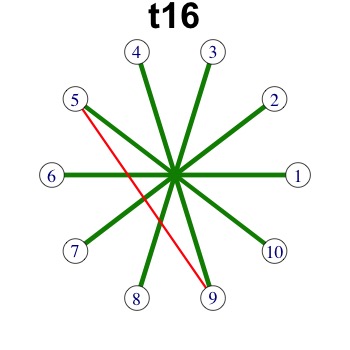}
\includegraphics[width=0.18\linewidth,valign=t]{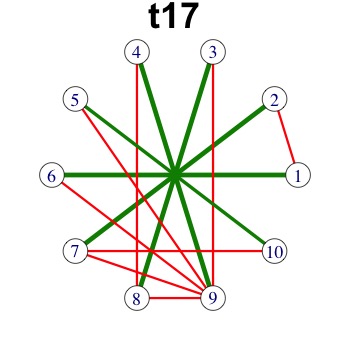}
\includegraphics[width=0.18\linewidth,valign=t]{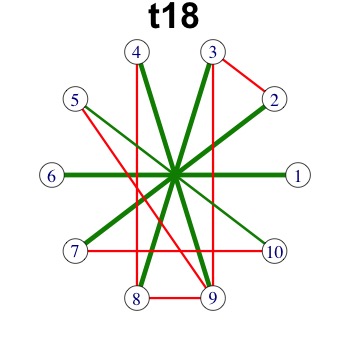}
\includegraphics[width=0.18\linewidth,valign=t]{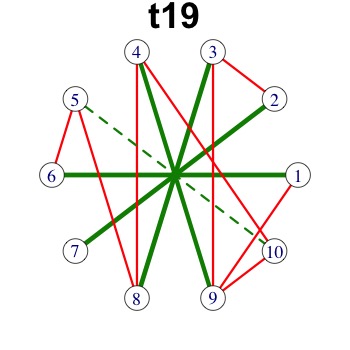}
\includegraphics[width=0.18\linewidth,valign=t]{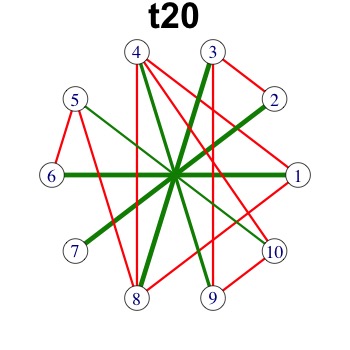}
\includegraphics[width=0.18\linewidth,valign=t]{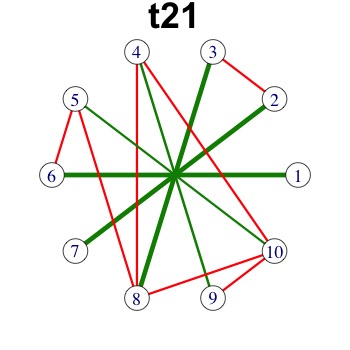}
\includegraphics[width=0.18\linewidth,valign=t]{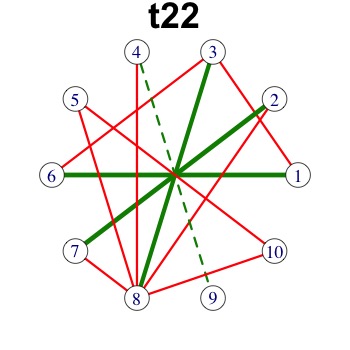}
\includegraphics[width=0.18\linewidth,valign=t]{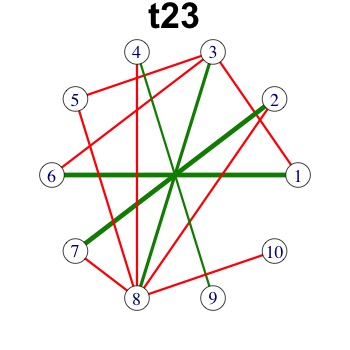}
\includegraphics[width=0.18\linewidth,valign=t]{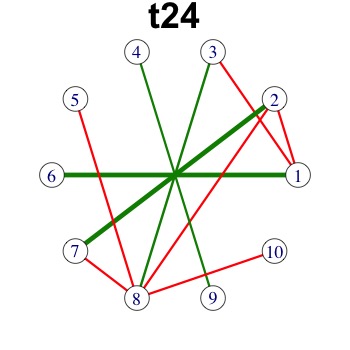}
\includegraphics[width=0.18\linewidth,valign=t]{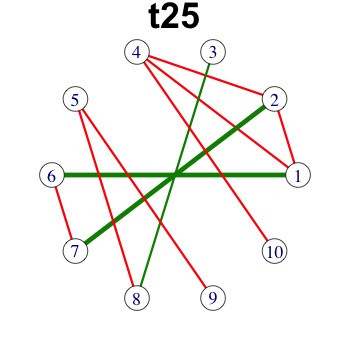}
\includegraphics[width=0.18\linewidth,valign=t]{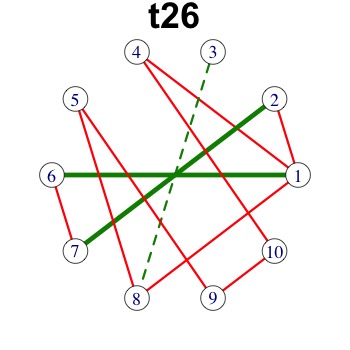}
\includegraphics[width=0.18\linewidth,valign=t]{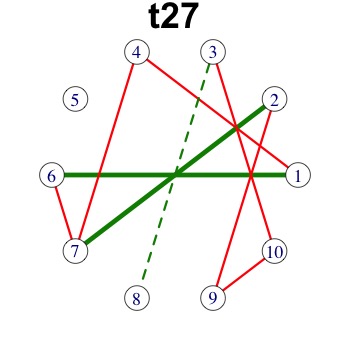}
\includegraphics[width=0.18\linewidth,valign=t]{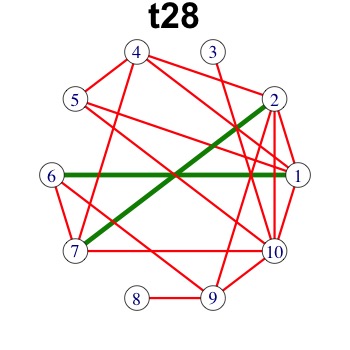}
\includegraphics[width=0.18\linewidth,valign=t]{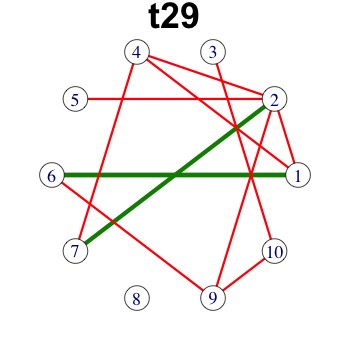}
\includegraphics[width=0.18\linewidth,valign=t]{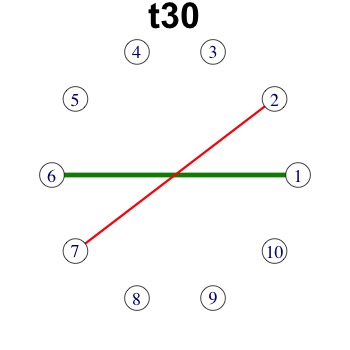}
\end{subfigure}
\caption{Estimated partial correlation networks from generalized elastic net for Scenario 1 with sample size 200. The green solid lines, the green dashed lines and the red solid lines represent the true positive connections, false negative connections and false positive connections, respectively. The thickness of each green solid line represents the magnitude of its underlying partial correlation.}
\label{fig:s1_gen_network}
\end{figure}

\begin{figure}[H]
\centering
\begin{subfigure}[b]{\textwidth}
\centering
\includegraphics[width=0.18\linewidth,valign=t]{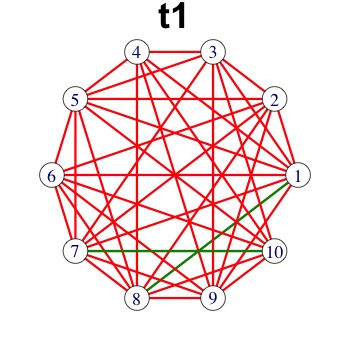}
\includegraphics[width=0.18\linewidth,valign=t]{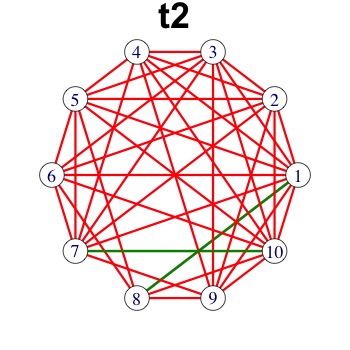}
\includegraphics[width=0.18\linewidth,valign=t]{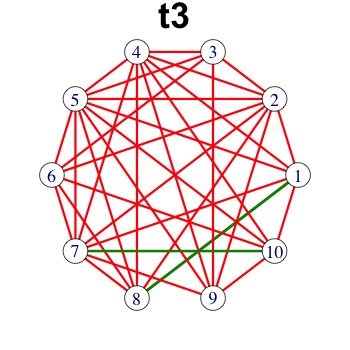}
\includegraphics[width=0.18\linewidth,valign=t]{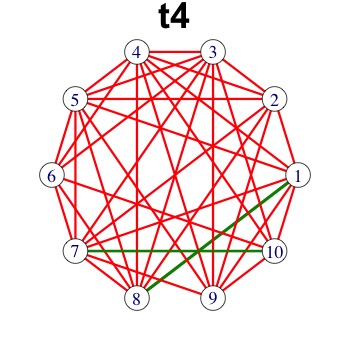}
\includegraphics[width=0.18\linewidth,valign=t]{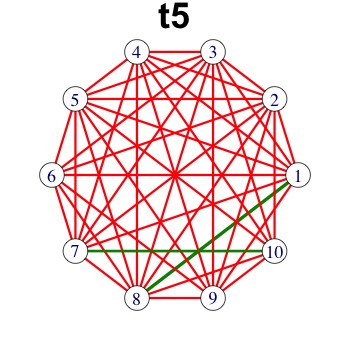}
\includegraphics[width=0.18\linewidth,valign=t]{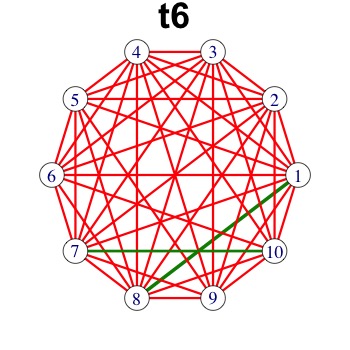}
\includegraphics[width=0.18\linewidth,valign=t]{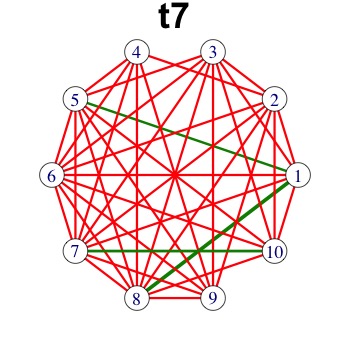}
\includegraphics[width=0.18\linewidth,valign=t]{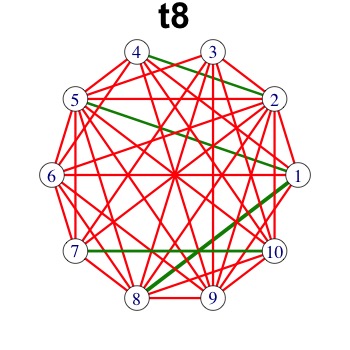}
\includegraphics[width=0.18\linewidth,valign=t]{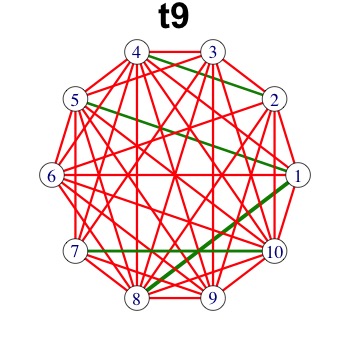}
\includegraphics[width=0.18\linewidth,valign=t]{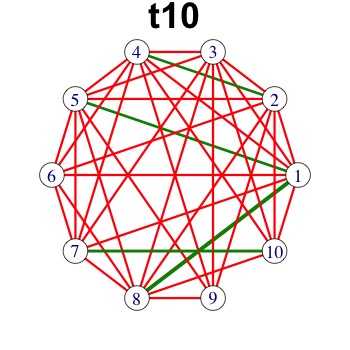}
\includegraphics[width=0.18\linewidth,valign=t]{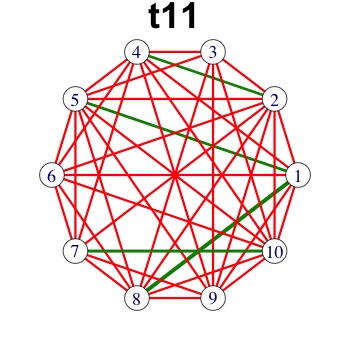}
\includegraphics[width=0.18\linewidth,valign=t]{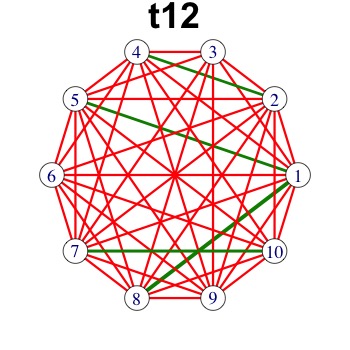}
\includegraphics[width=0.18\linewidth,valign=t]{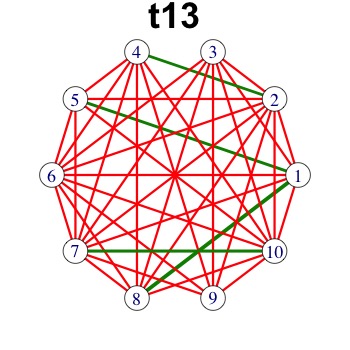}
\includegraphics[width=0.18\linewidth,valign=t]{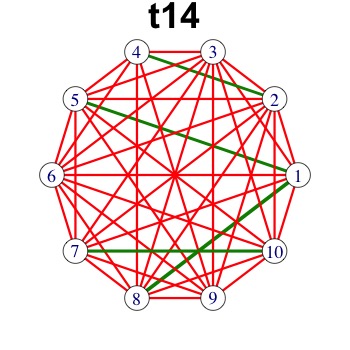}
\includegraphics[width=0.18\linewidth,valign=t]{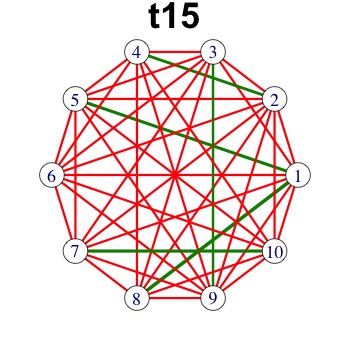}
\includegraphics[width=0.18\linewidth,valign=t]{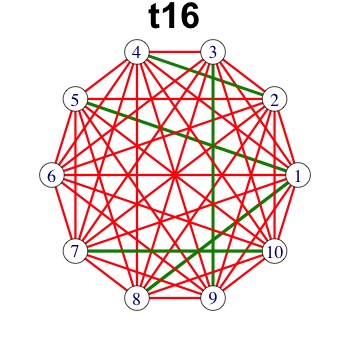}
\includegraphics[width=0.18\linewidth,valign=t]{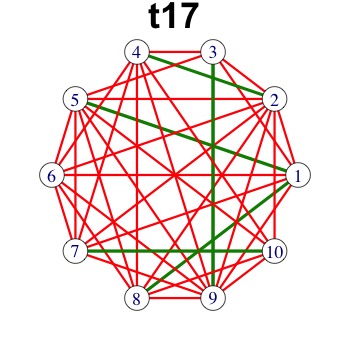}
\includegraphics[width=0.18\linewidth,valign=t]{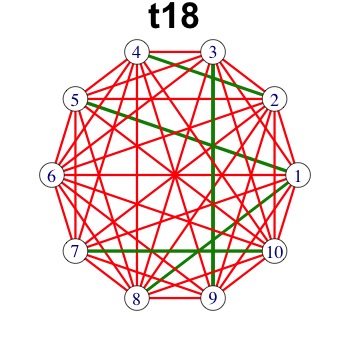}
\includegraphics[width=0.18\linewidth,valign=t]{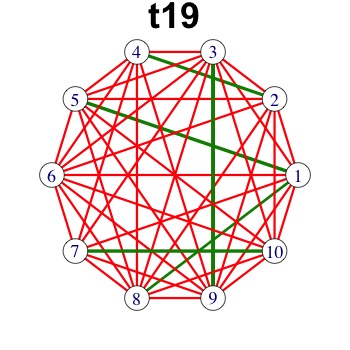}
\includegraphics[width=0.18\linewidth,valign=t]{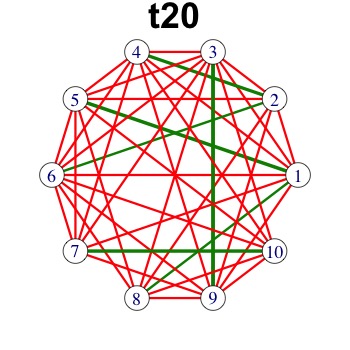}
\includegraphics[width=0.18\linewidth,valign=t]{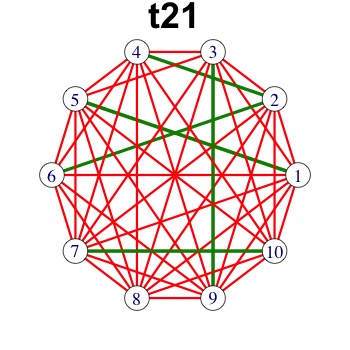}
\includegraphics[width=0.18\linewidth,valign=t]{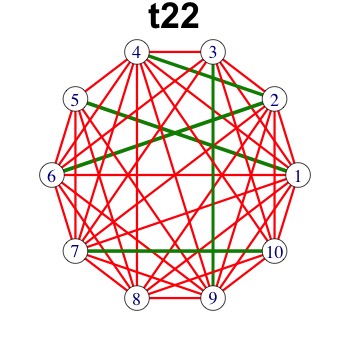}
\includegraphics[width=0.18\linewidth,valign=t]{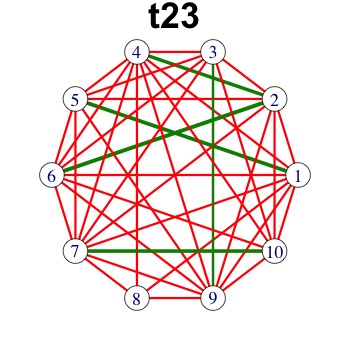}
\includegraphics[width=0.18\linewidth,valign=t]{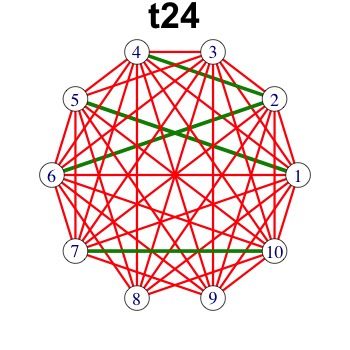}
\includegraphics[width=0.18\linewidth,valign=t]{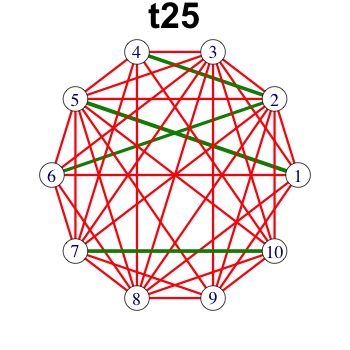}
\includegraphics[width=0.18\linewidth,valign=t]{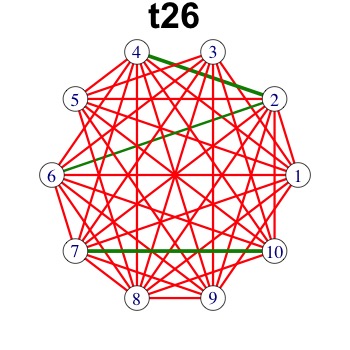}
\includegraphics[width=0.18\linewidth,valign=t]{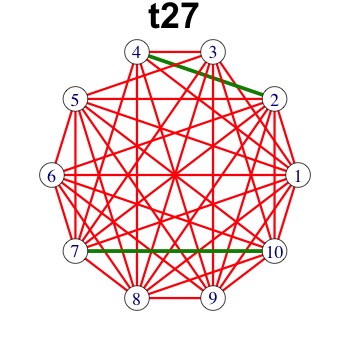}
\includegraphics[width=0.18\linewidth,valign=t]{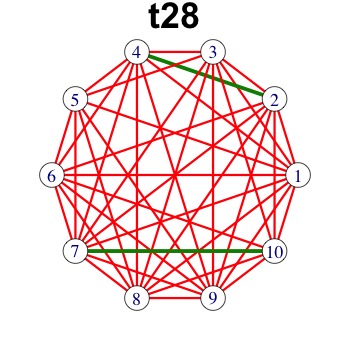}
\includegraphics[width=0.18\linewidth,valign=t]{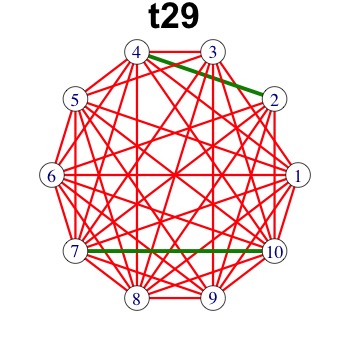}
\includegraphics[width=0.18\linewidth,valign=t]{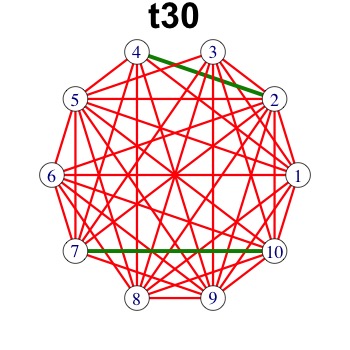}
\end{subfigure}
\caption{Estimated partial correlation networks from generalized elastic net for Scenario 2 with sample size 200. The green solid lines, the green dashed lines and the red solid lines represent the true positive connections, false negative connections and false positive connections, respectively. The thickness of each green solid line represents the magnitude of its underlying partial correlation.}
\label{fig:s2_gen_network}
\end{figure}

\begin{figure}[H]
\centering
\begin{subfigure}[b]{\textwidth}
\centering
\includegraphics[width=0.18\linewidth,valign=t]{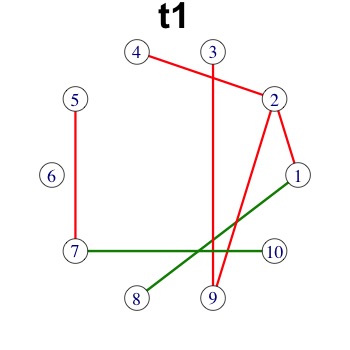}
\includegraphics[width=0.18\linewidth,valign=t]{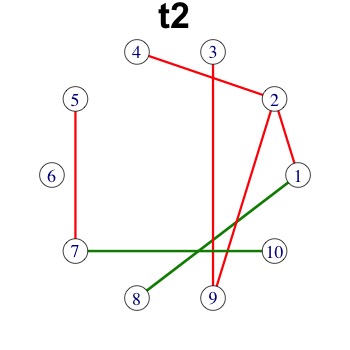}
\includegraphics[width=0.18\linewidth,valign=t]{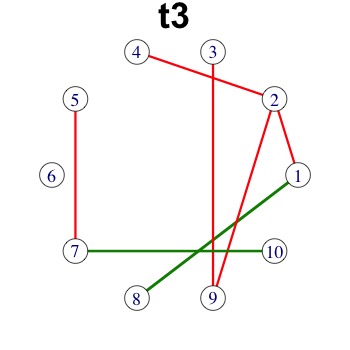}
\includegraphics[width=0.18\linewidth,valign=t]{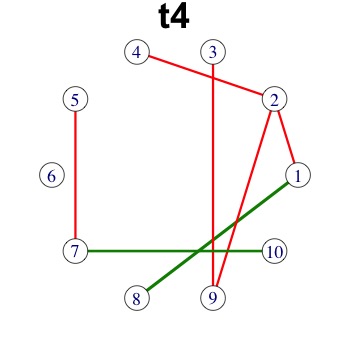}
\includegraphics[width=0.18\linewidth,valign=t]{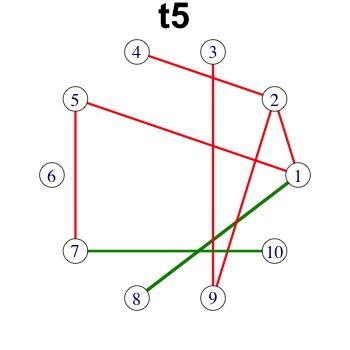}
\includegraphics[width=0.18\linewidth,valign=t]{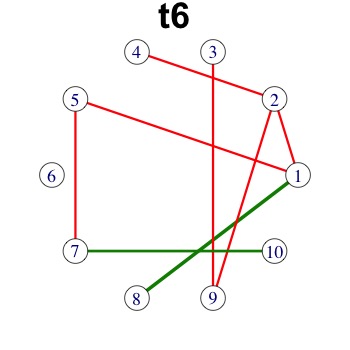}
\includegraphics[width=0.18\linewidth,valign=t]{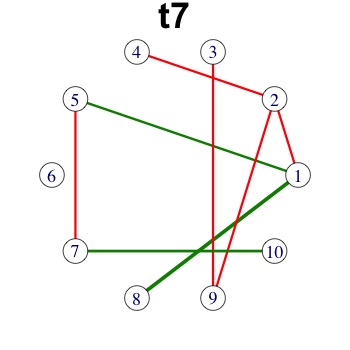}
\includegraphics[width=0.18\linewidth,valign=t]{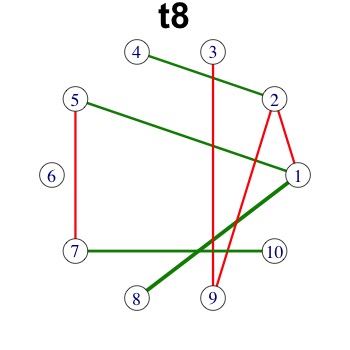}
\includegraphics[width=0.18\linewidth,valign=t]{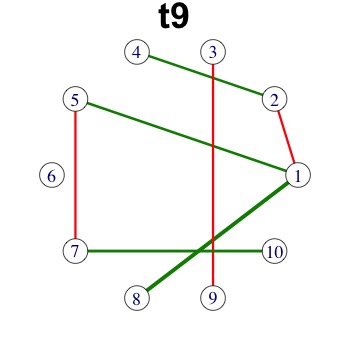}
\includegraphics[width=0.18\linewidth,valign=t]{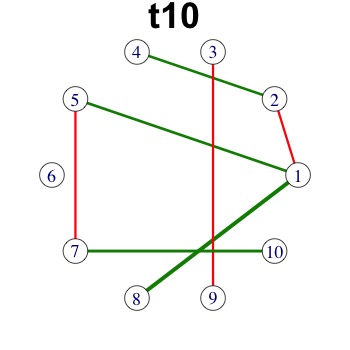}
\includegraphics[width=0.18\linewidth,valign=t]{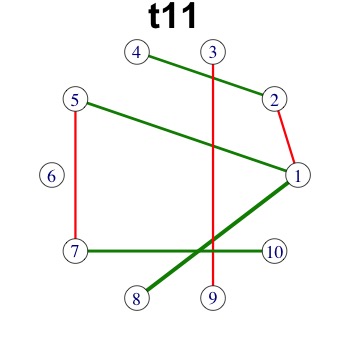}
\includegraphics[width=0.18\linewidth,valign=t]{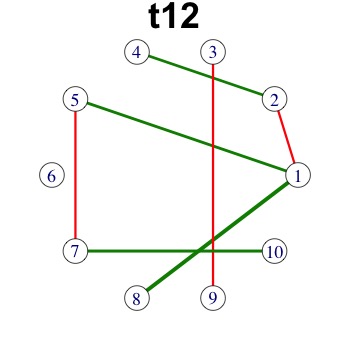}
\includegraphics[width=0.18\linewidth,valign=t]{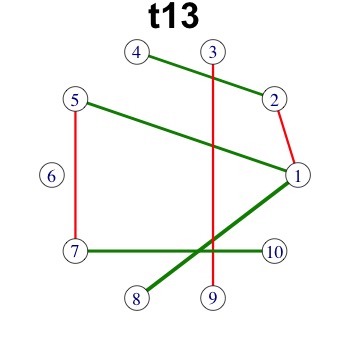}
\includegraphics[width=0.18\linewidth,valign=t]{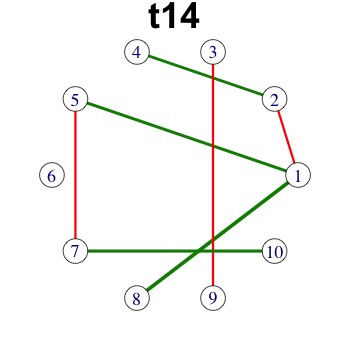}
\includegraphics[width=0.18\linewidth,valign=t]{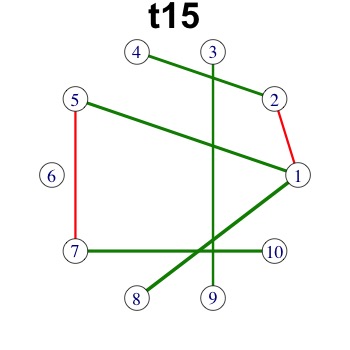}
\includegraphics[width=0.18\linewidth,valign=t]{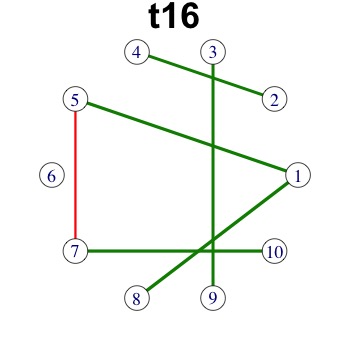}
\includegraphics[width=0.18\linewidth,valign=t]{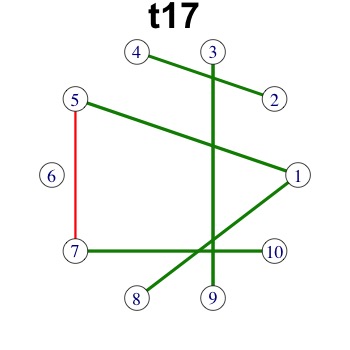}
\includegraphics[width=0.18\linewidth,valign=t]{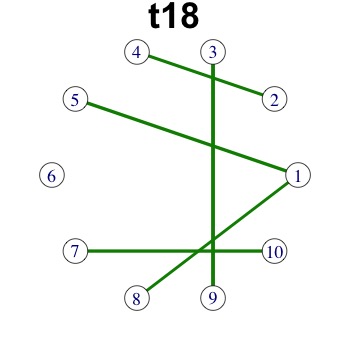}
\includegraphics[width=0.18\linewidth,valign=t]{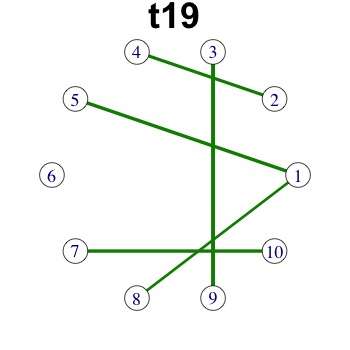}
\includegraphics[width=0.18\linewidth,valign=t]{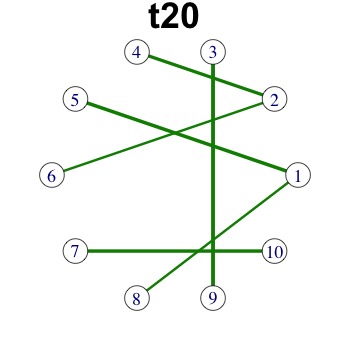}
\includegraphics[width=0.18\linewidth,valign=t]{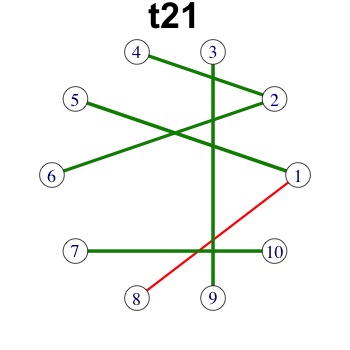}
\includegraphics[width=0.18\linewidth,valign=t]{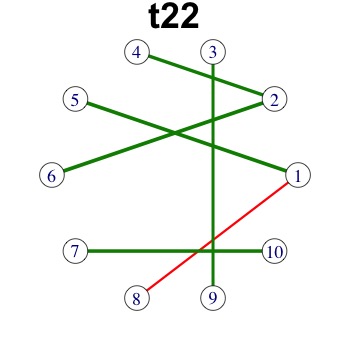}
\includegraphics[width=0.18\linewidth,valign=t]{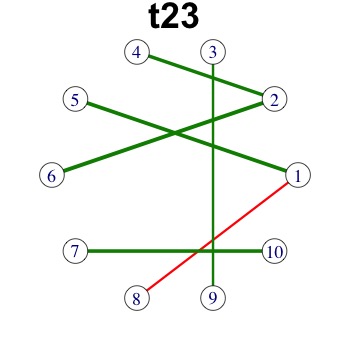}
\includegraphics[width=0.18\linewidth,valign=t]{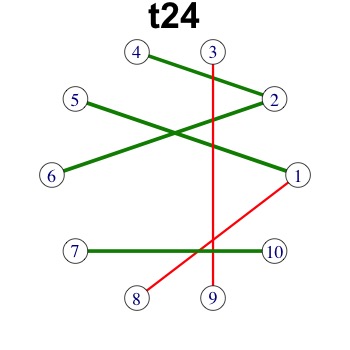}
\includegraphics[width=0.18\linewidth,valign=t]{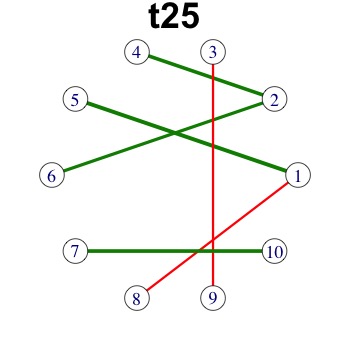}
\includegraphics[width=0.18\linewidth,valign=t]{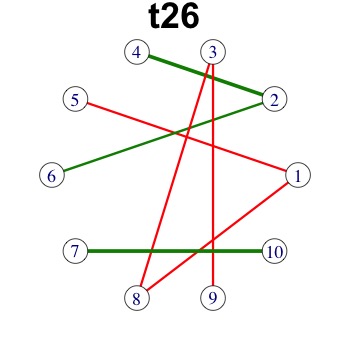}
\includegraphics[width=0.18\linewidth,valign=t]{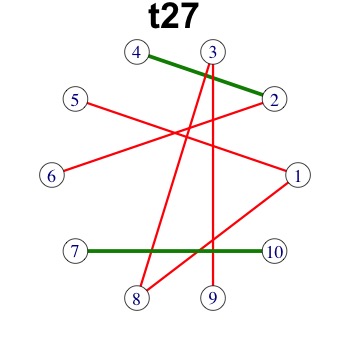}
\includegraphics[width=0.18\linewidth,valign=t]{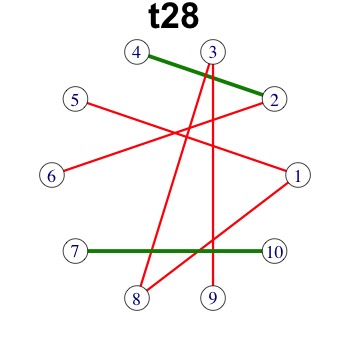}
\includegraphics[width=0.18\linewidth,valign=t]{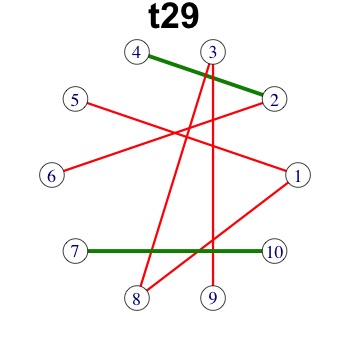}
\includegraphics[width=0.18\linewidth,valign=t]{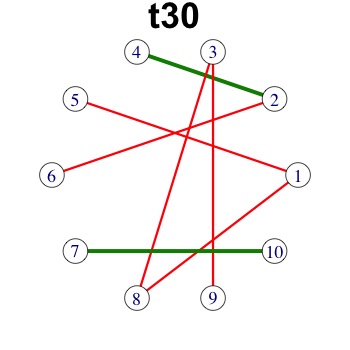}
\end{subfigure}
\caption{Estimated partial correlation networks from generalized fused lasso for Scenario 2 with sample size 200. The green solid lines, the green dashed lines and the red solid lines represent the true positive connections, false negative connections and false positive connections, respectively. The thickness of each green solid line represents the magnitude of its underlying partial correlation.}
\label{fig:s2_gfl_network}
\end{figure}

\begin{figure}[H]
\centering
\begin{subfigure}[b]{\textwidth}
\centering
\includegraphics[angle=270,width=0.3\linewidth,valign=t]{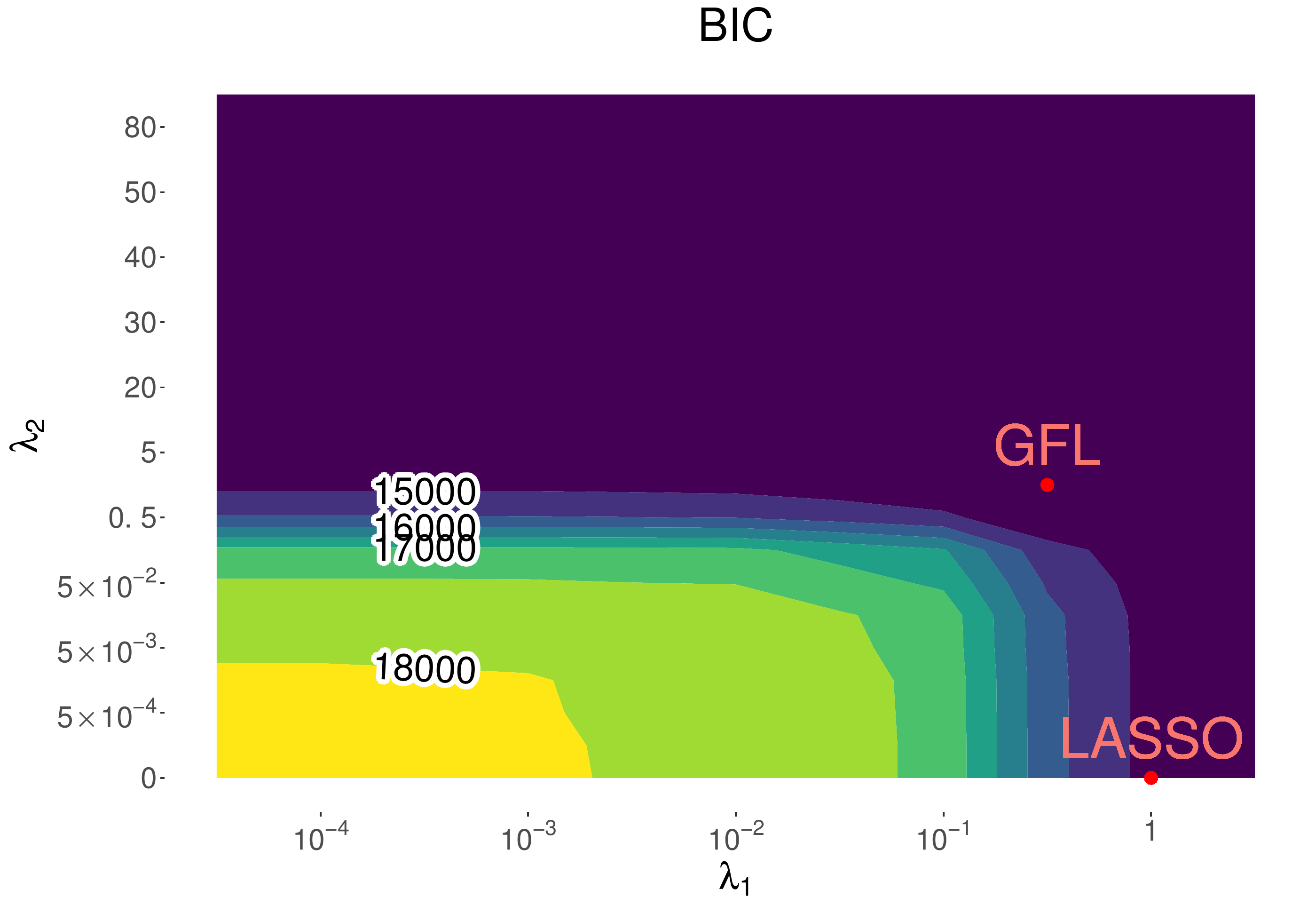}
\includegraphics[angle=270,width=0.3\linewidth,valign=t]{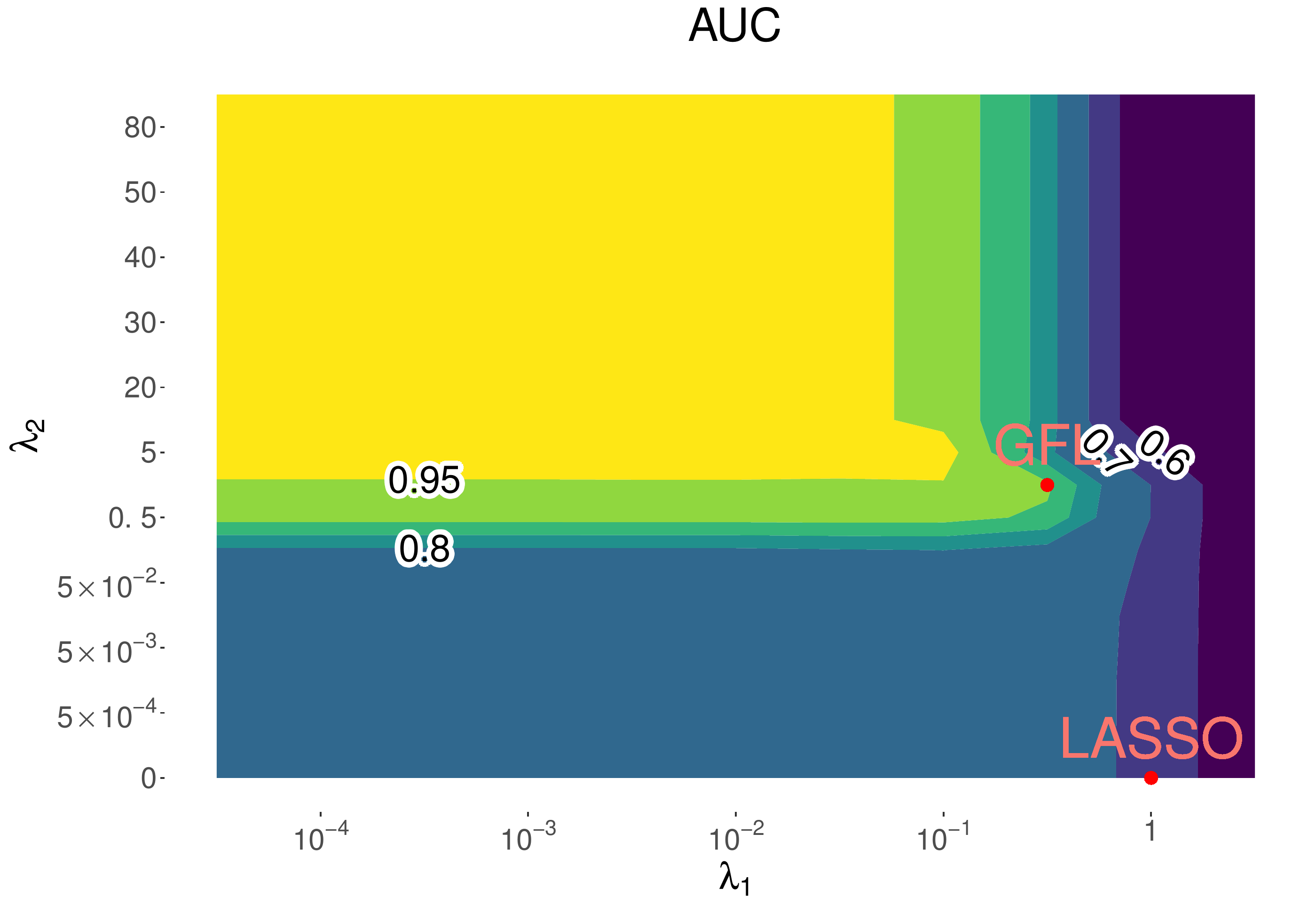}
\includegraphics[angle=270,width=0.3\linewidth,valign=t]{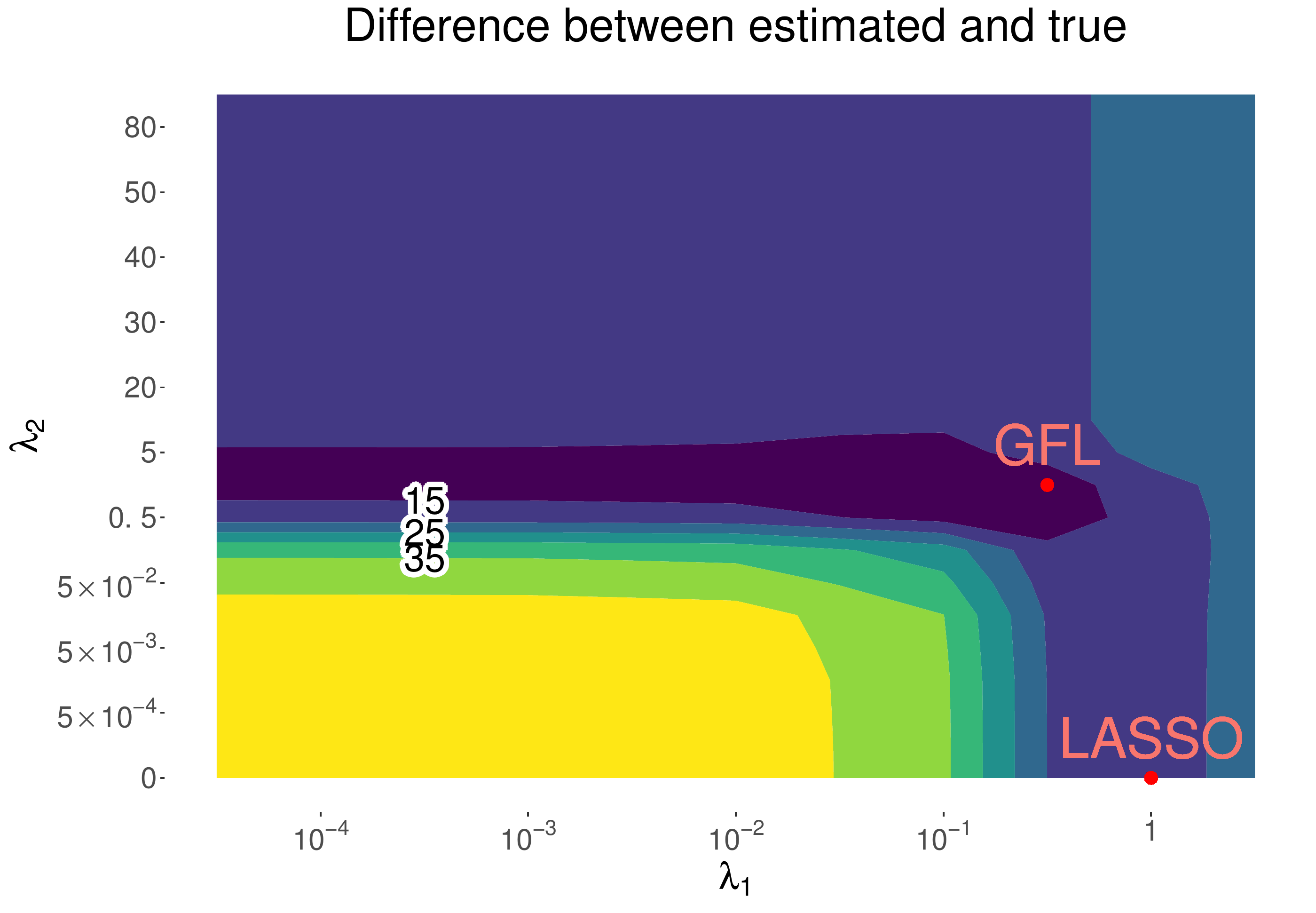}
\caption{sample size 50}
\end{subfigure}
\begin{subfigure}[b]{\textwidth}
\centering
\includegraphics[angle=270,width=0.3\linewidth,valign=t]{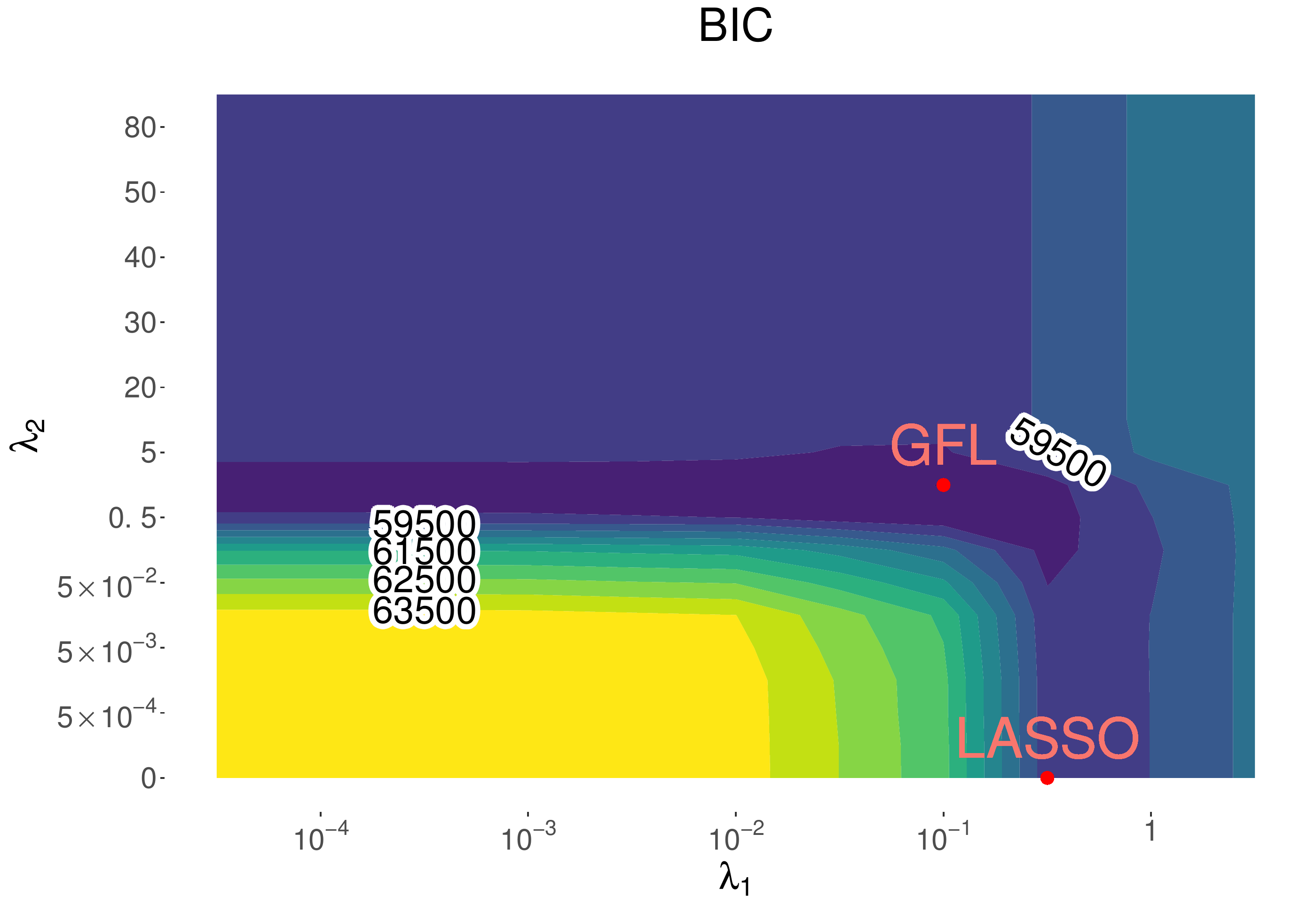}
\includegraphics[angle=270,width=0.3\linewidth,valign=t]{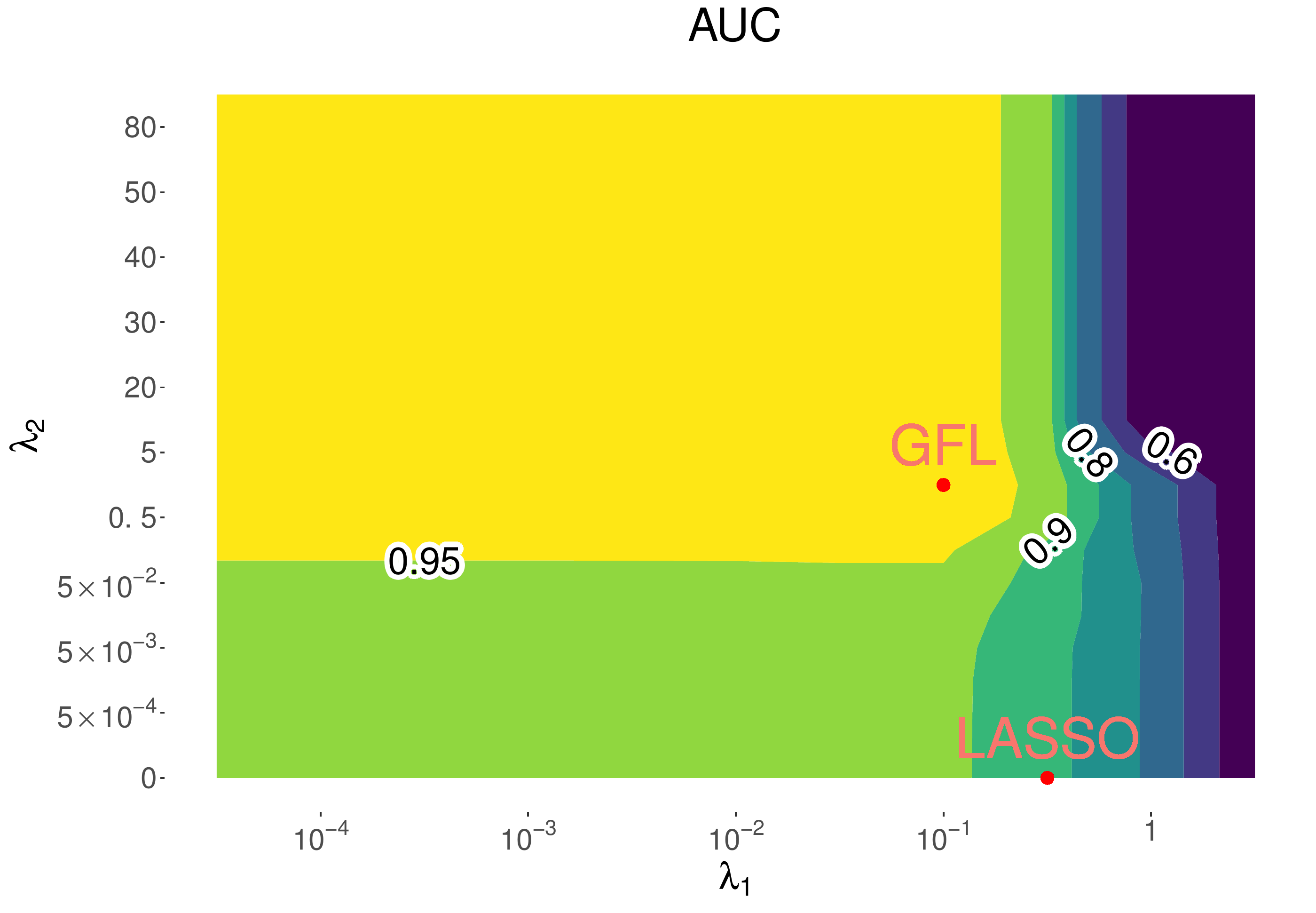}
\includegraphics[angle=270,width=0.3\linewidth,valign=t]{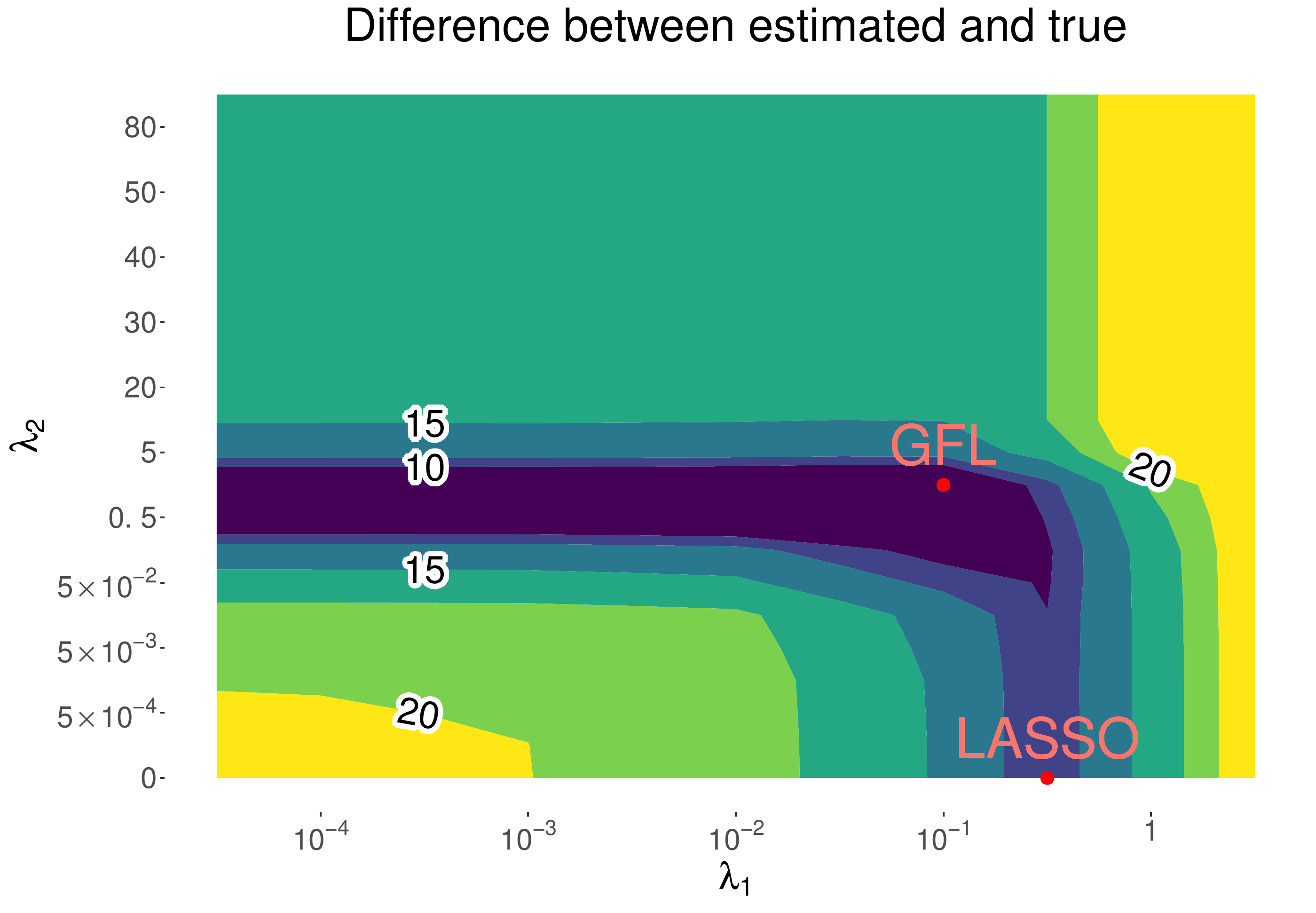}
\caption{sample size 200}
\end{subfigure}
\caption{Simulation performance measurements from generalized fused lasso for Scenario 2. Darker areas represent lower values. For reference, in the simulation with sample size 50, the sample AUC is 0.7236, comparing to the GFL AUC 0.9101 in the figure; the difference between sample and true is 43.92, comparing to the difference between GFL estimate and true as 12.15 in the figure. In the simulation with sample size 200, the sample AUC is 0.9097, comparing to the GFL AUC 0.9994 in the figure; the difference between sample and true is 20.08, comparing to the difference between GFL estimate and true as 5.23 in the figure.}
\label{fig:s2_gfl_contour}
\end{figure}

\begin{figure}[H]
\centering
\begin{subfigure}[b]{\textwidth}
\centering
\includegraphics[angle=270,width=0.3\linewidth,valign=t]{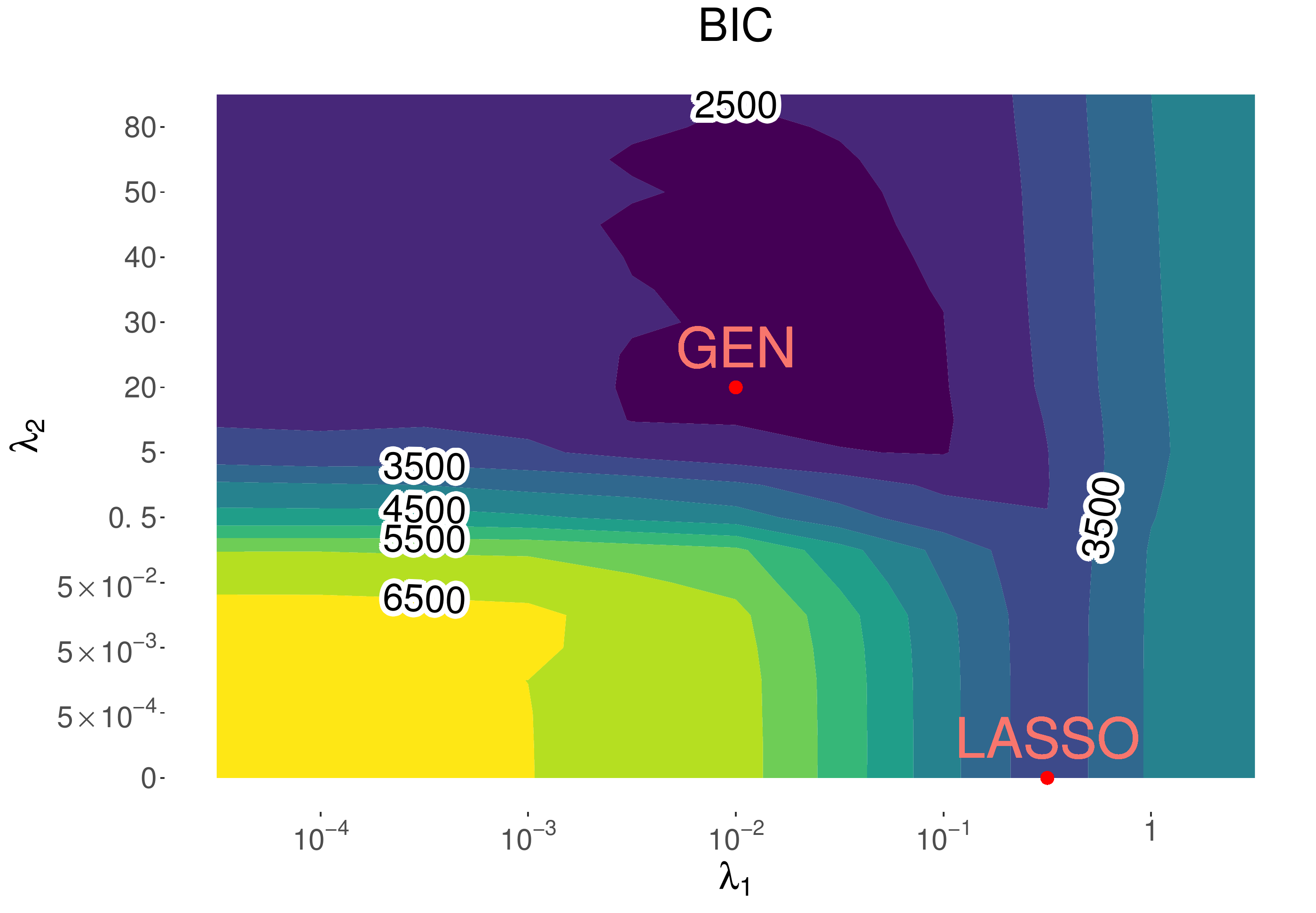}
\includegraphics[angle=270,width=0.3\linewidth,valign=t]{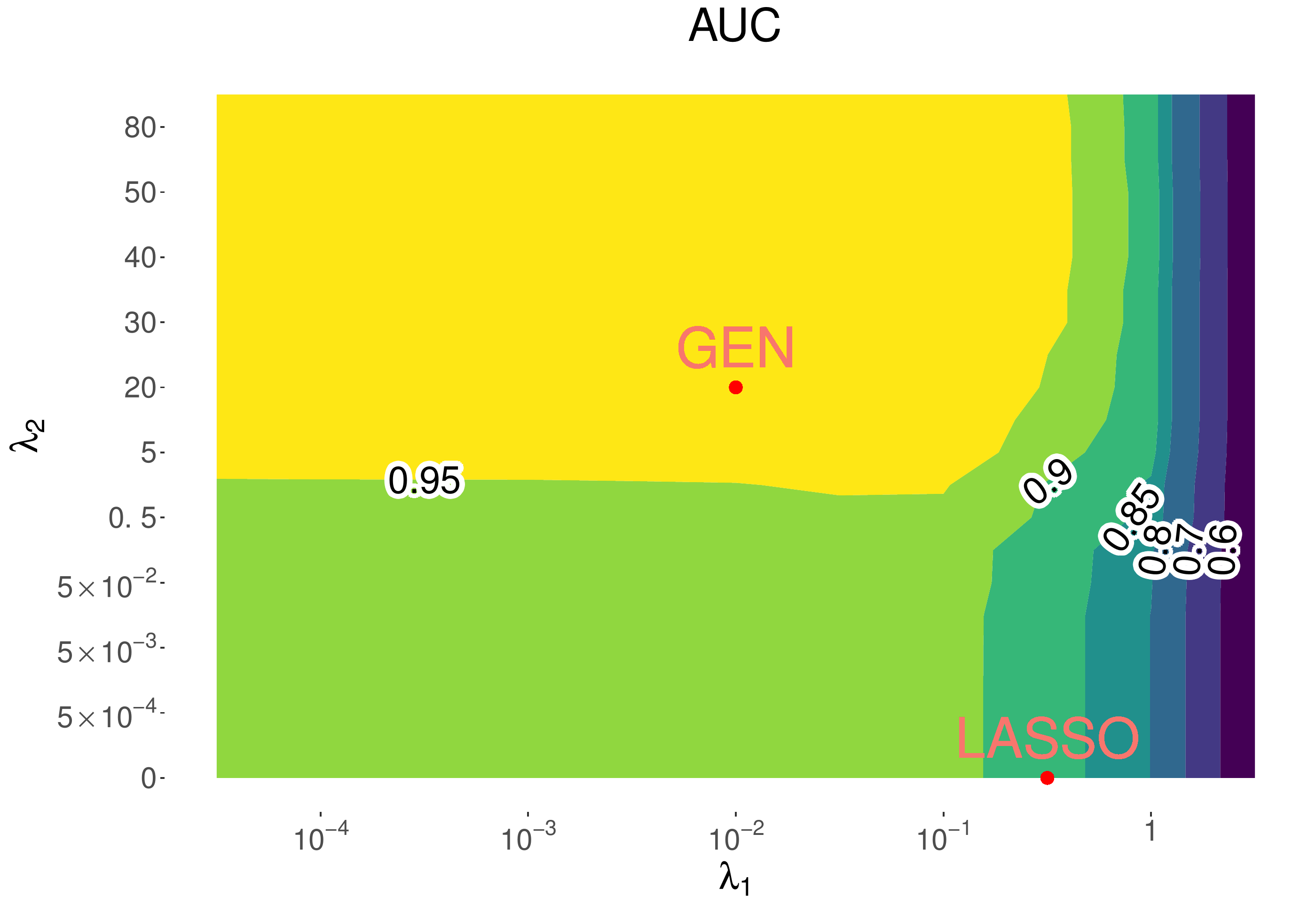}
\includegraphics[angle=270,width=0.3\linewidth,valign=t]{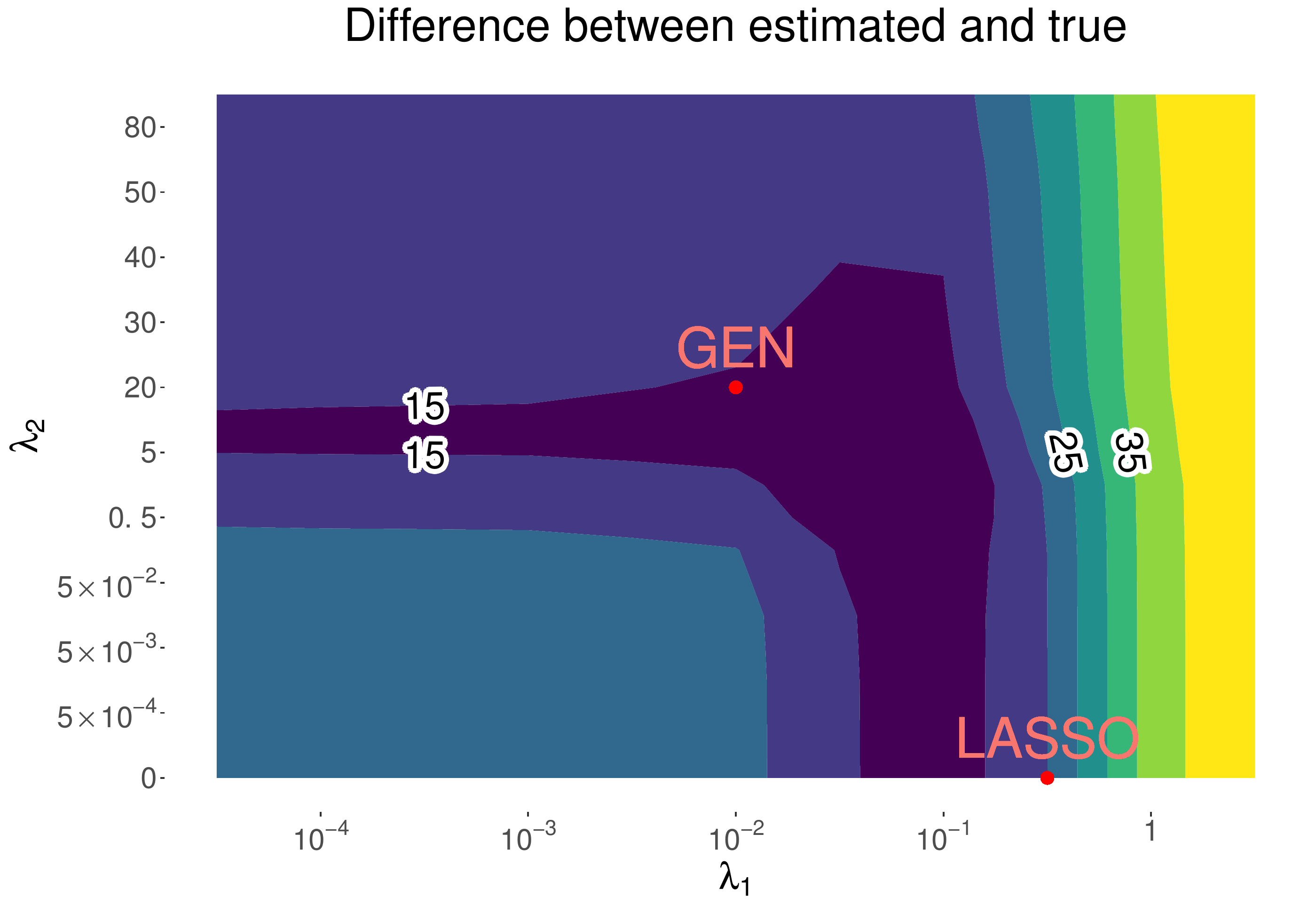}
\caption{sample size 50}
\end{subfigure}
\begin{subfigure}[b]{\textwidth}
\centering
\includegraphics[angle=270,width=0.3\linewidth,valign=t]{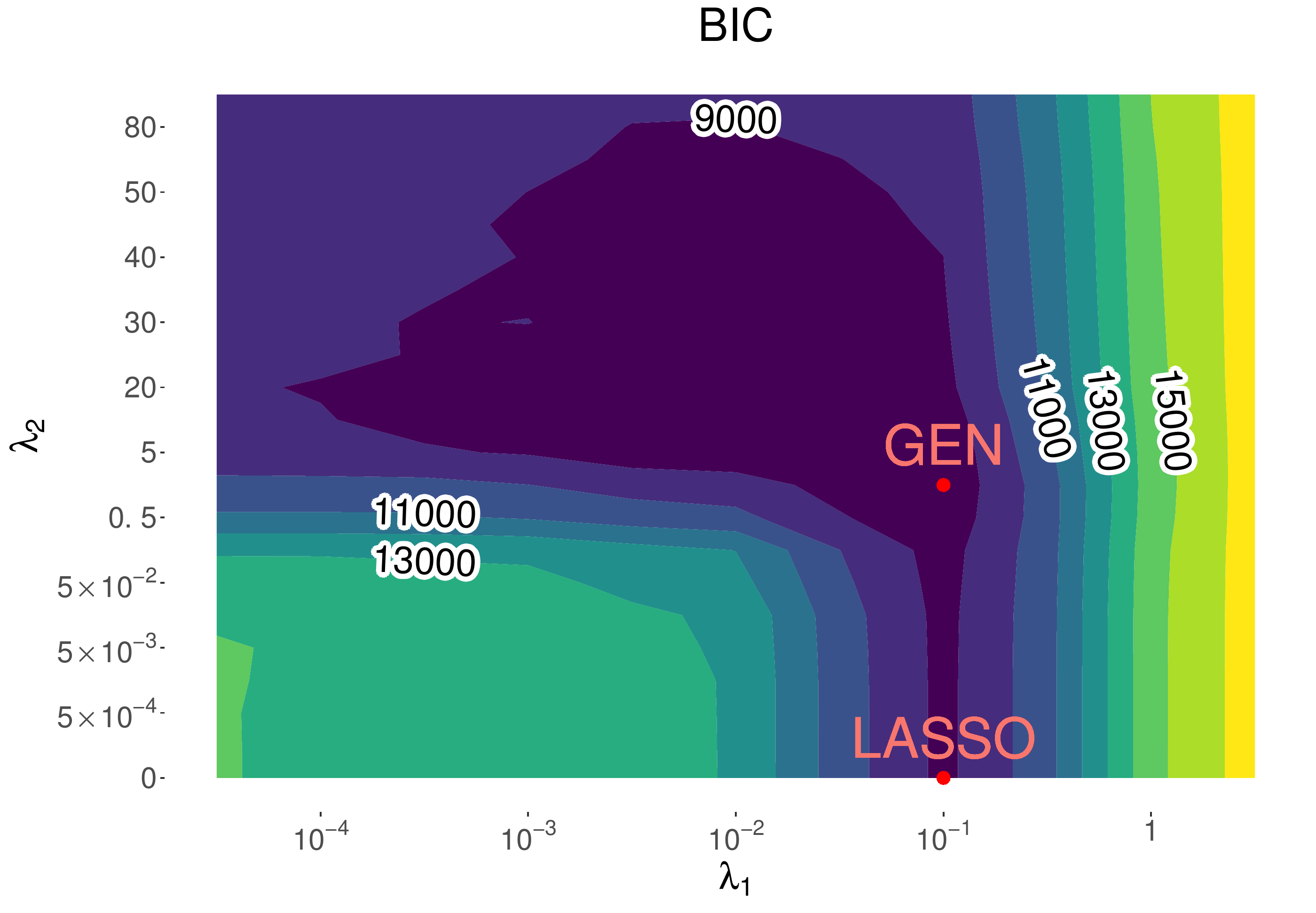}
\includegraphics[angle=270,width=0.3\linewidth,valign=t]{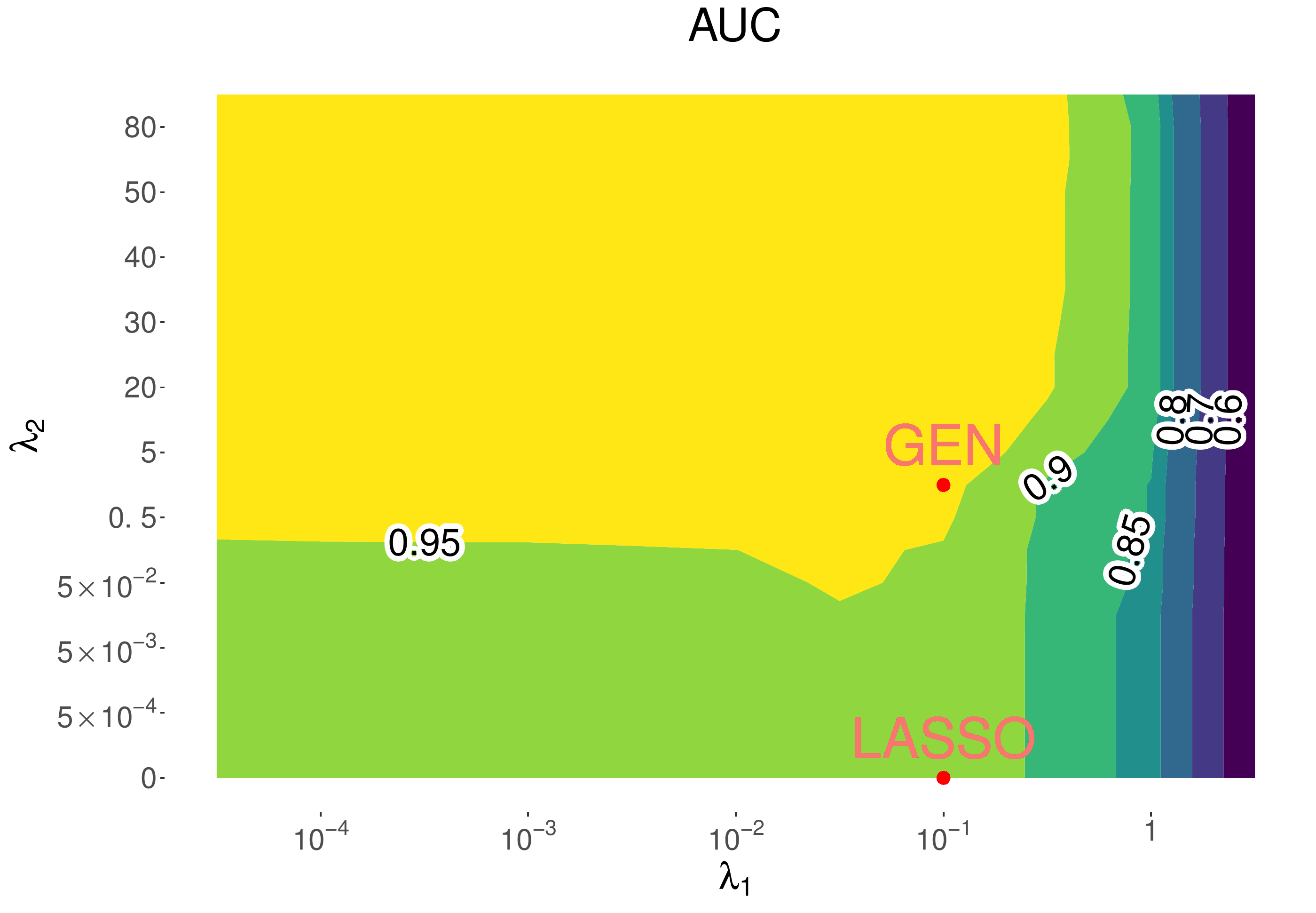}
\includegraphics[angle=270,width=0.3\linewidth,valign=t]{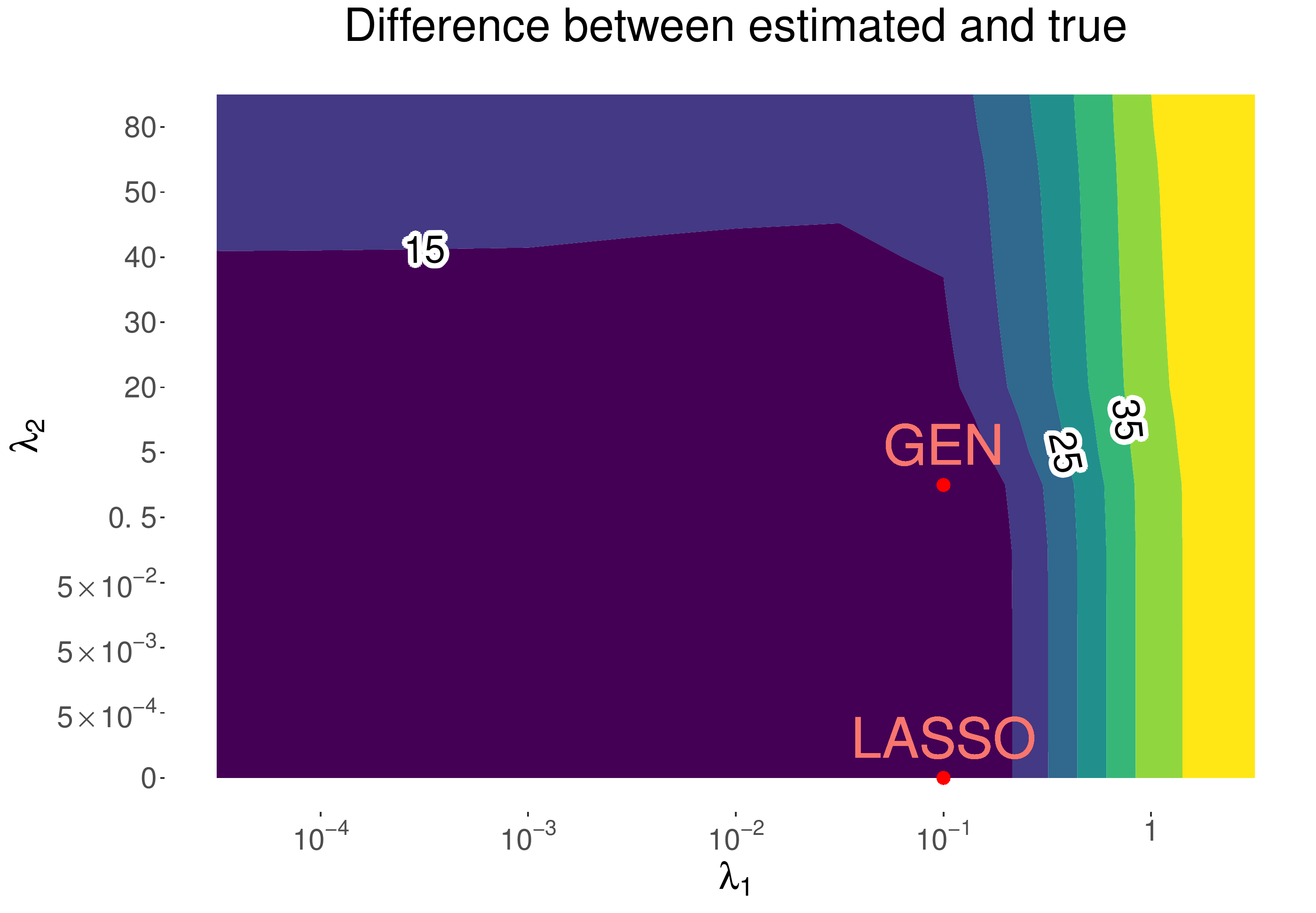}
\caption{sample size 200}
\end{subfigure}
\caption{Simulation performance measurements from generalized elastic net for Scenario 1. Darker areas represent lower values. For reference, in the simulation with sample size 50, the sample AUC is 0.9017, comparing to the GEN AUC 0.9967 in the figure; the difference between sample and true is 23.94, comparing to the difference between GEN estimate and true as 14.76 in the figure. In the simulation with sample size 200, the sample AUC is 0.9365, comparing to the GEN AUC 0.9661 in the figure; the difference between sample and true is 11.28, comparing to the difference between GEN estimate and true as 6.79 in the figure.}
\label{fig:s1_gen_contour}
\end{figure}

\begin{figure}[H]
\centering
\begin{subfigure}[b]{\textwidth}
\centering
\includegraphics[angle=270,width=0.3\linewidth,valign=t]{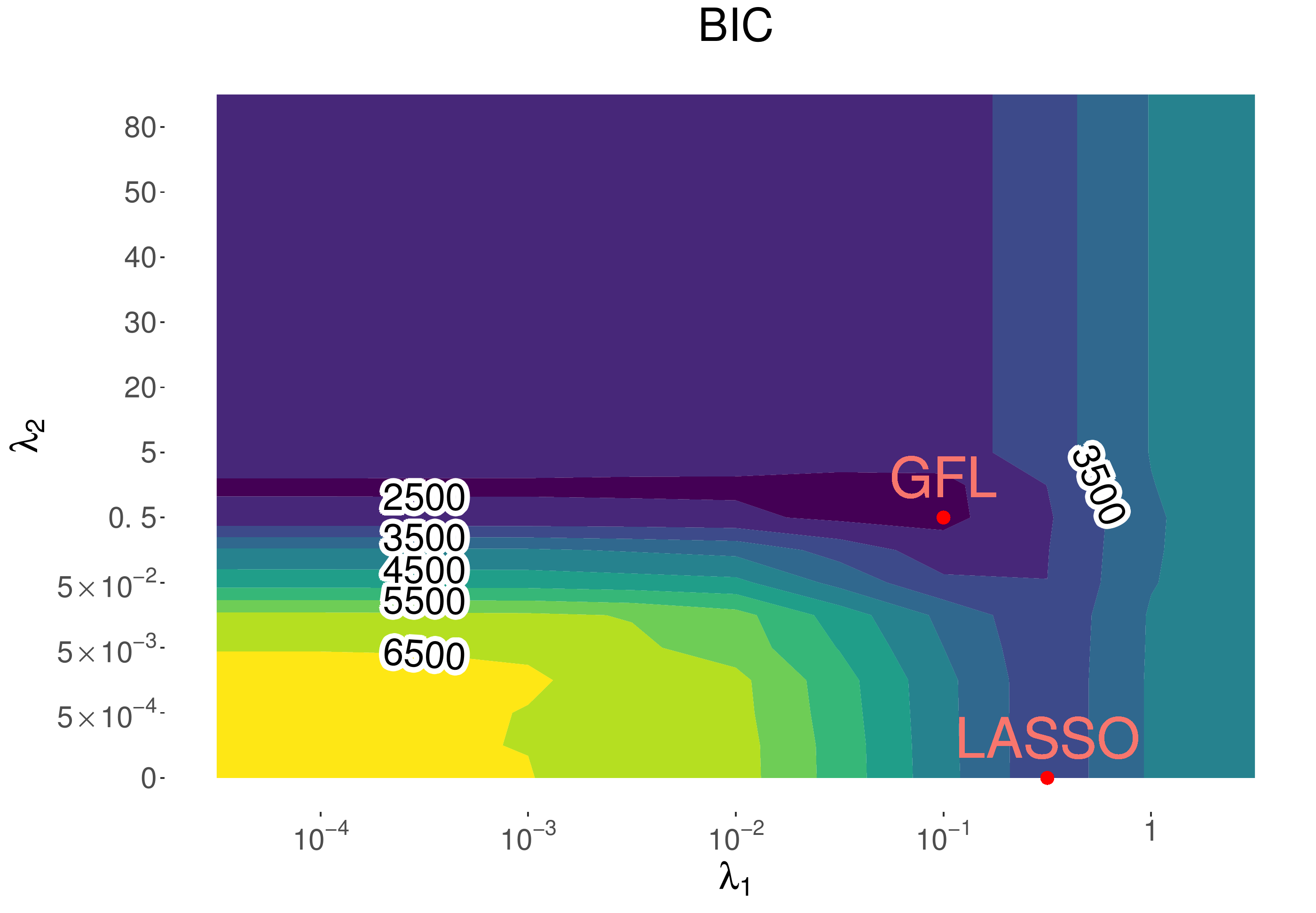}
\includegraphics[angle=270,width=0.3\linewidth,valign=t]{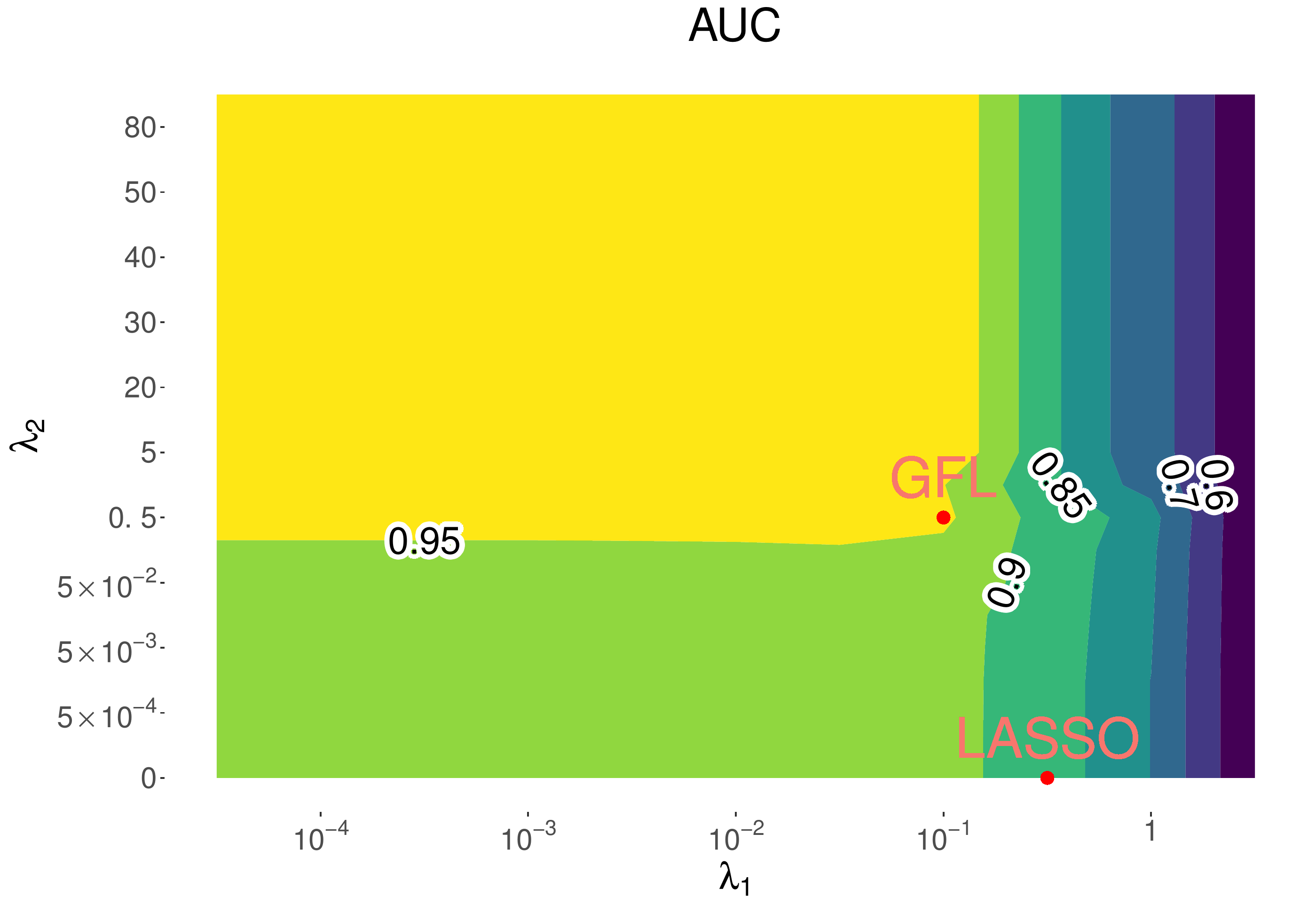}
\includegraphics[angle=270,width=0.3\linewidth,valign=t]{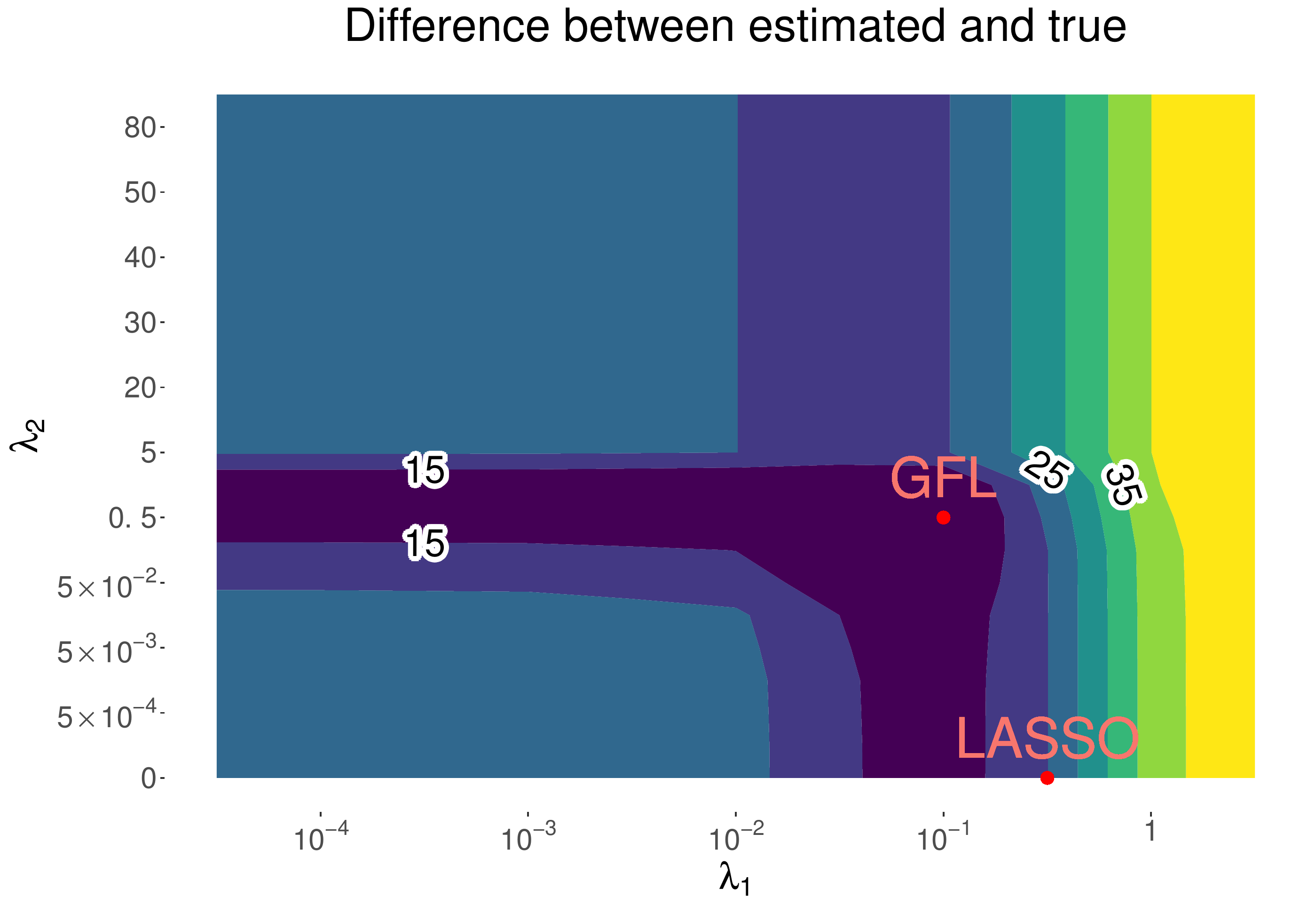}
\caption{sample size 50}
\end{subfigure}
\begin{subfigure}[b]{\textwidth}
\centering
\includegraphics[angle=270,width=0.3\linewidth,valign=t]{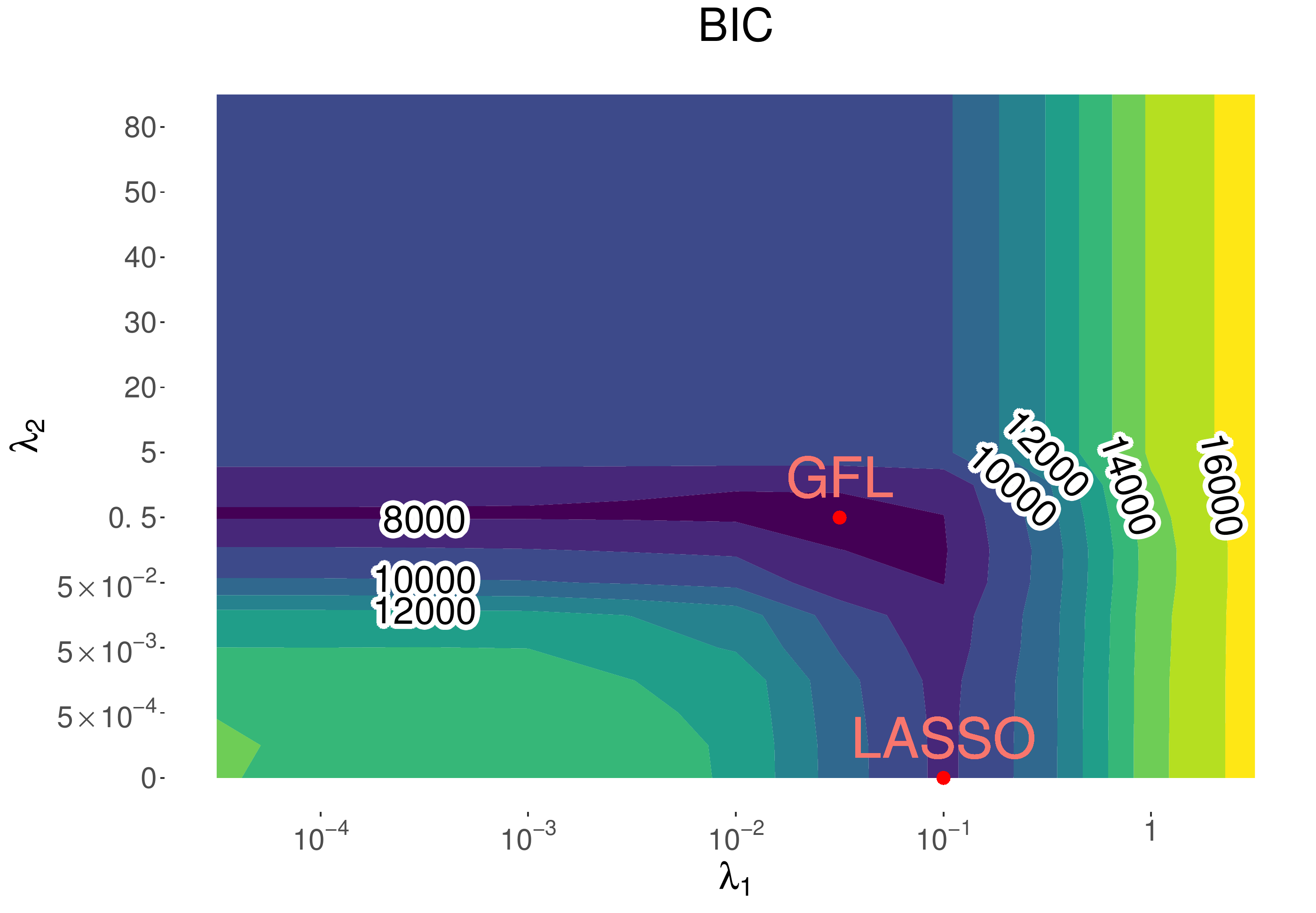}
\includegraphics[angle=270,width=0.3\linewidth,valign=t]{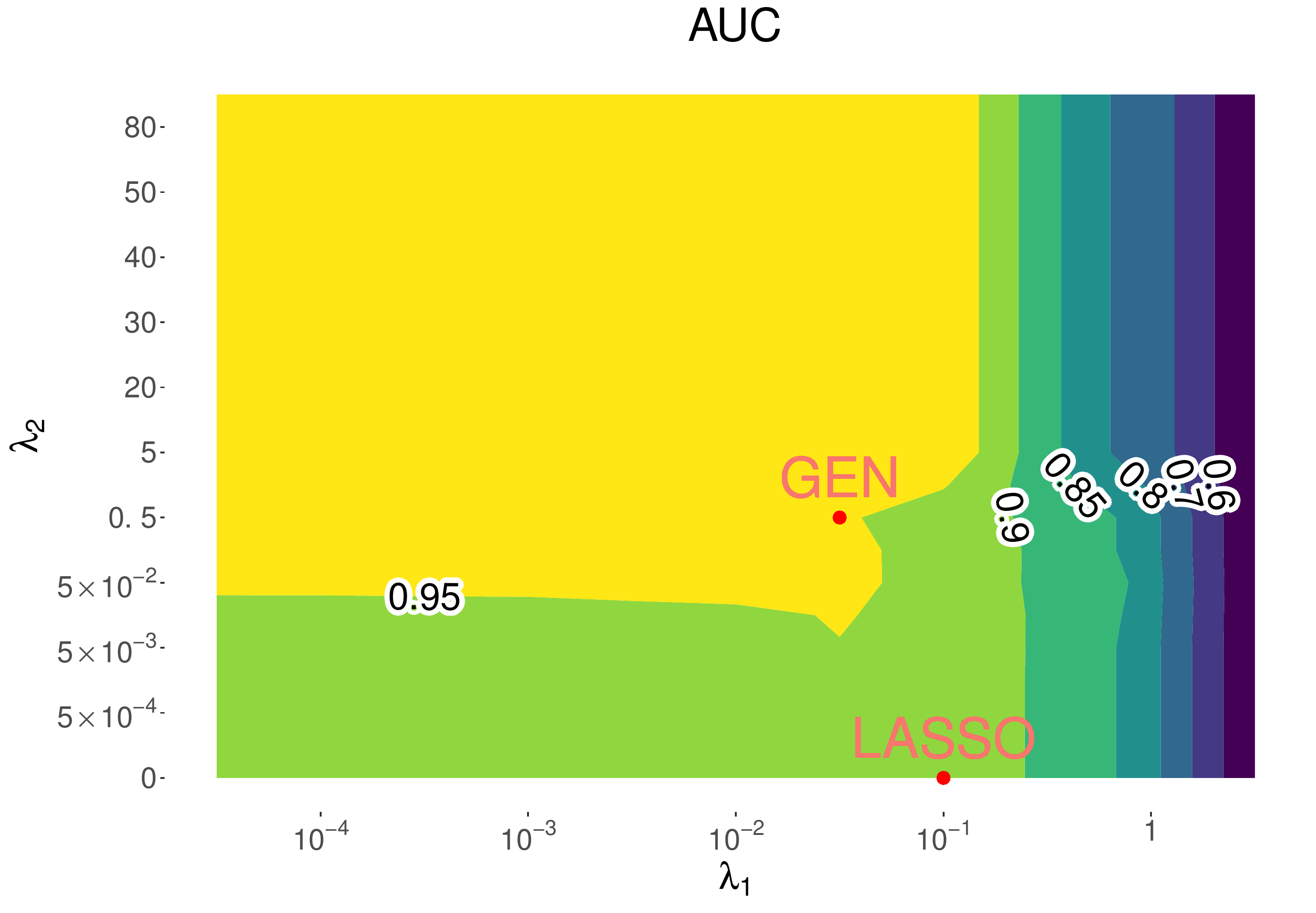}
\includegraphics[angle=270,width=0.3\linewidth,valign=t]{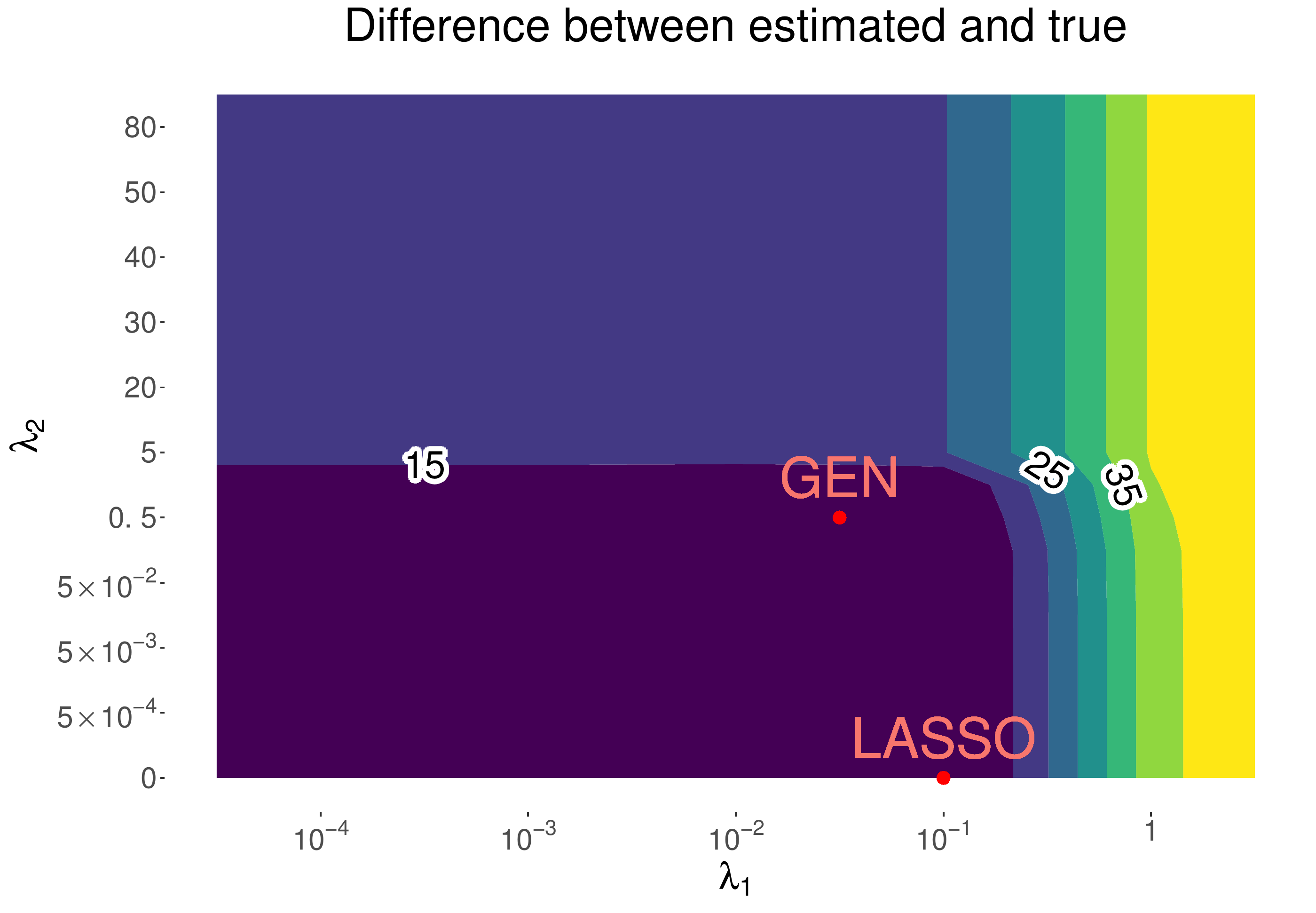}
\caption{sample size 200}
\end{subfigure}
\caption{Simulation performance measurements from generalized fused lasso for Scenario 1. Darker areas represent lower values. For reference, in the simulation with sample size 50, the sample AUC is 0.9017, comparing to the GFL AUC 0.9595 in the figure; the difference between sample and true is 23.94, comparing to the difference between GFL estimate and true as 6.66 in the figure. In the simulation with sample size 200, the sample AUC is 0.9365, comparing to the GFL AUC 0.9570 in the figure; the difference between sample and true is 11.28, comparing to the difference between GFL estimate and true as 5.20 in the figure. }
\label{fig:s1_gfl_contour}
\end{figure}

\section{Additional results from ADHD data}
\label{subsec:supplement6}

\begin{figure}[H]
\centering
\begin{subfigure}[b]{\textwidth}
\centering
\includegraphics[angle=270,width=0.45\linewidth,valign=t]{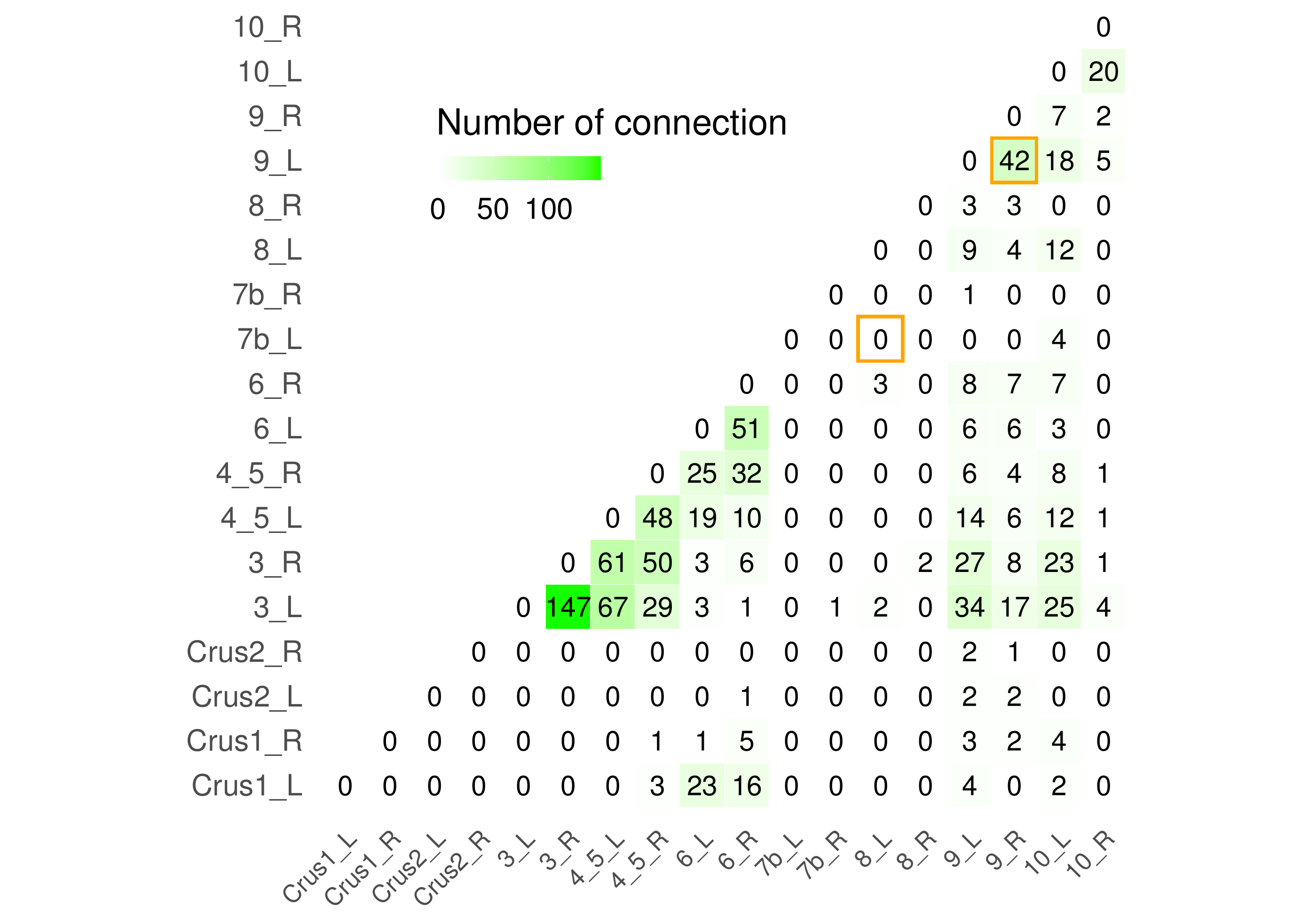}
\includegraphics[width=0.32\linewidth,valign=t]{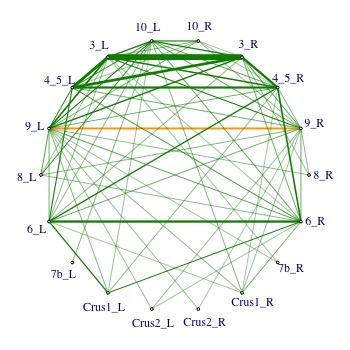}
\caption{Healthy group}
\end{subfigure}
\begin{subfigure}[b]{\textwidth}
\centering
\includegraphics[angle=270,width=0.45\linewidth,valign=t]{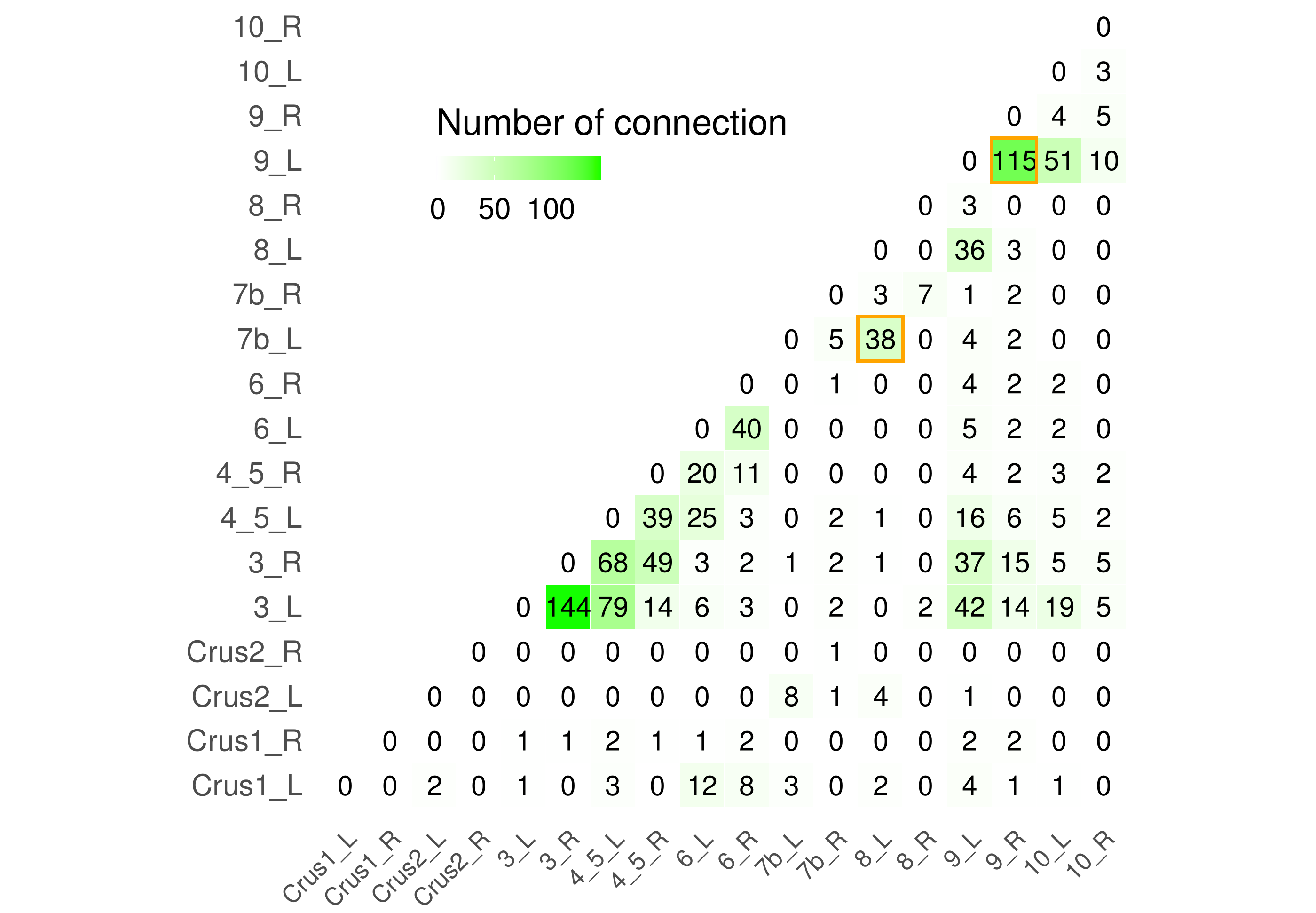}
\includegraphics[width=0.32\linewidth,valign=t]{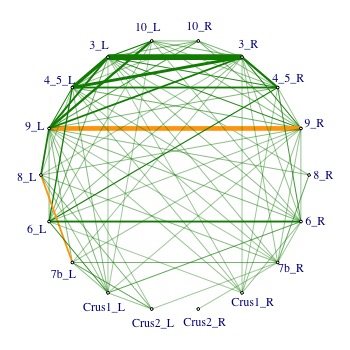}
\caption{ADHD group}
\end{subfigure}
\caption{[Result 1, GEN] Accumulated connections between different cerebellum regions estimated by generalized elastic net, showing the total connections between different cerebellum regions over 172 time points. The yellow squares on the left highlight the number of the edge 7b\_L - 8\_L and the edge 9\_L - 9\_R. }
\label{fig:app_gen11_connection}
\end{figure}

\begin{figure}[H]
\centering
\begin{subfigure}[b]{0.47\textwidth}
\centering
\includegraphics[width=\linewidth,valign=t]{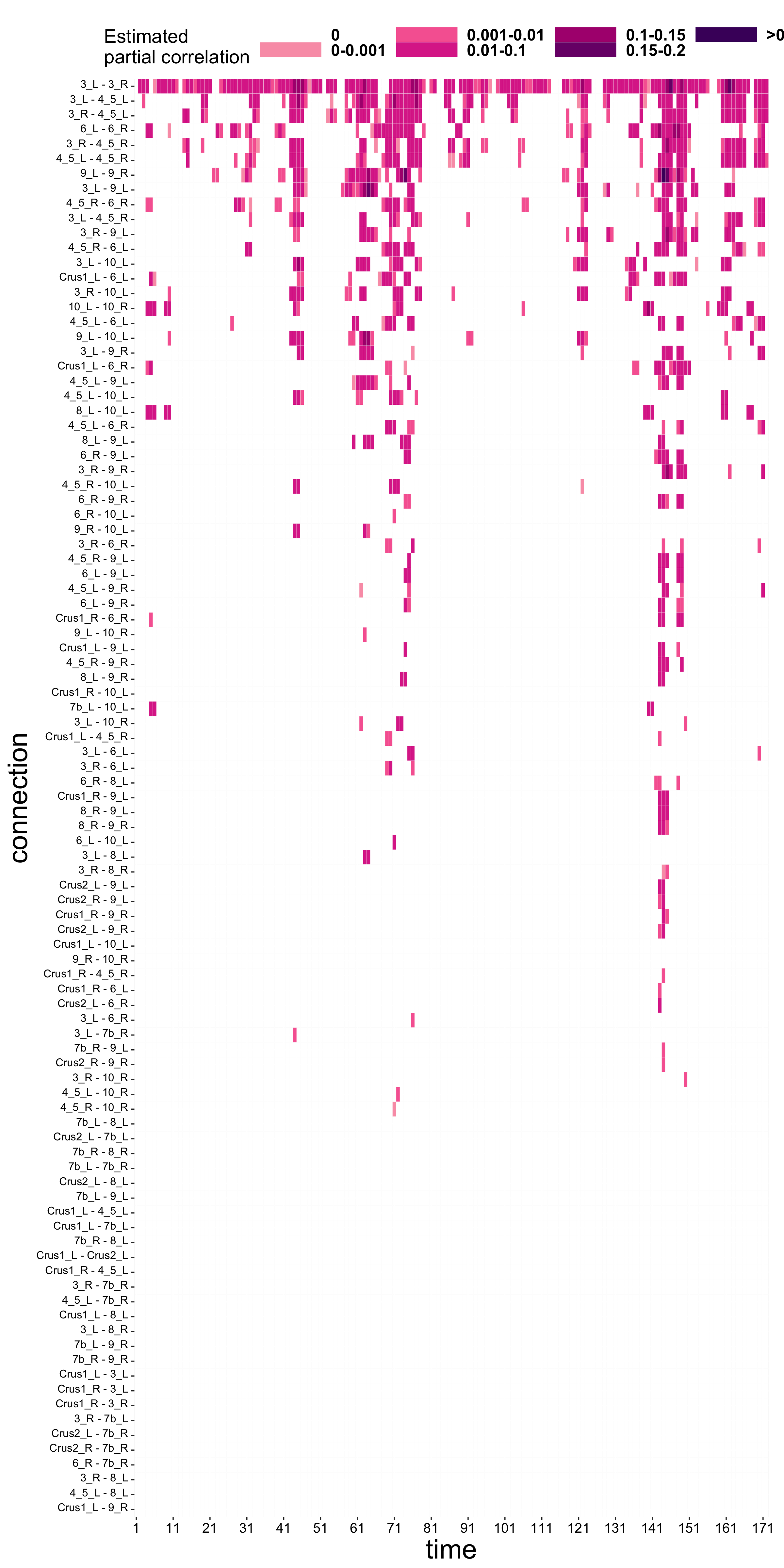}
\caption{Healthy group}
\end{subfigure}
~
\begin{subfigure}[b]{0.47\textwidth}
\centering
\includegraphics[width=\linewidth,valign=t]{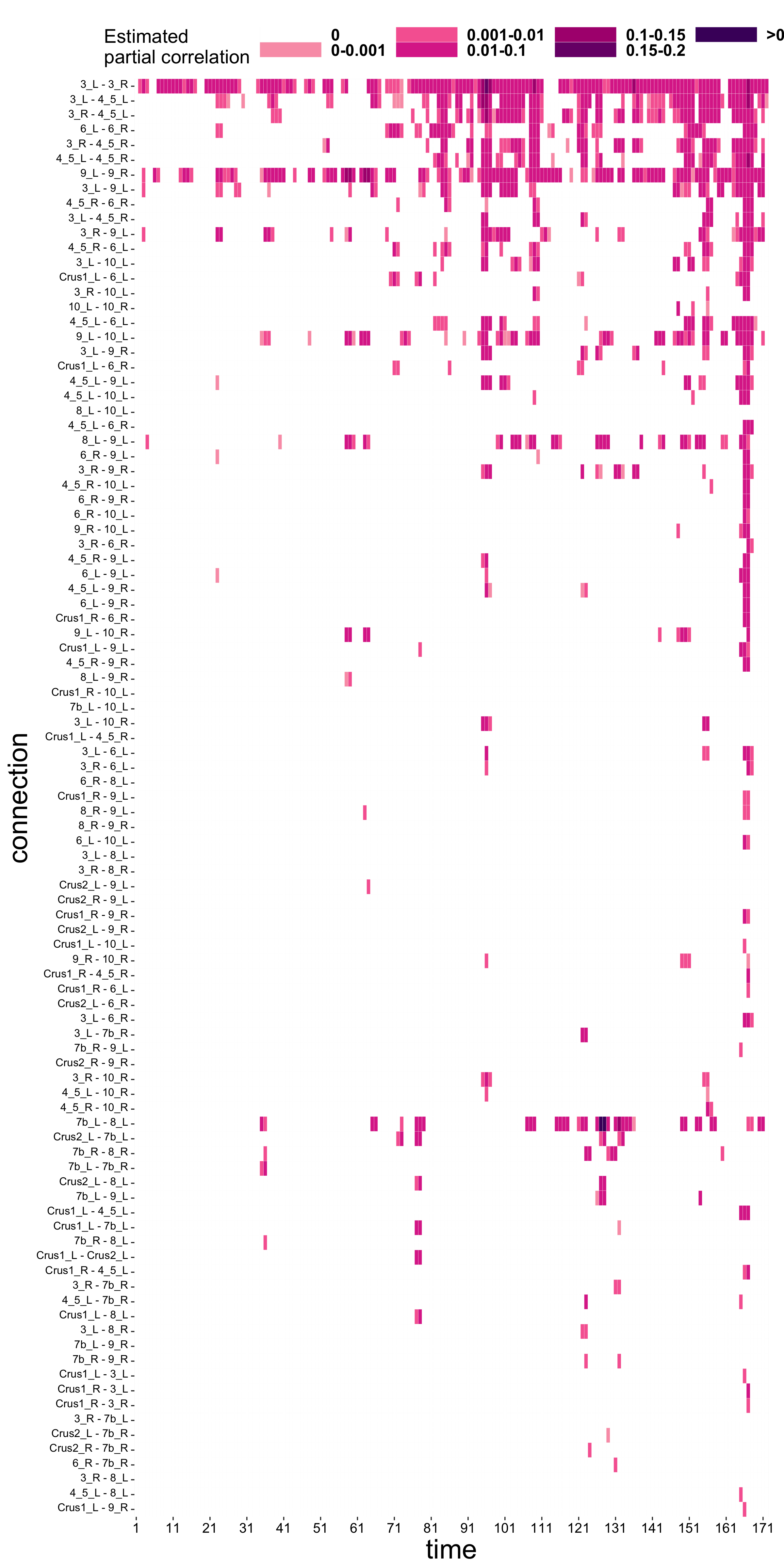}
\caption{ADHD group}
\end{subfigure}
\caption{[Result 1, GEN] Estimated partial correlations between different cerebellum regions estimated by generalized elastic net, illustrating how connections change over time. Each cell corresponds a connection at a given time point, and the color represents the magnitude of the estimated partial correlation value. }
\label{fig:app_gen11_table}
\end{figure}

\begin{figure}[H]
\centering
\begin{subfigure}[b]{\textwidth}
\centering
\includegraphics[angle=270,width=0.45\linewidth,valign=t]{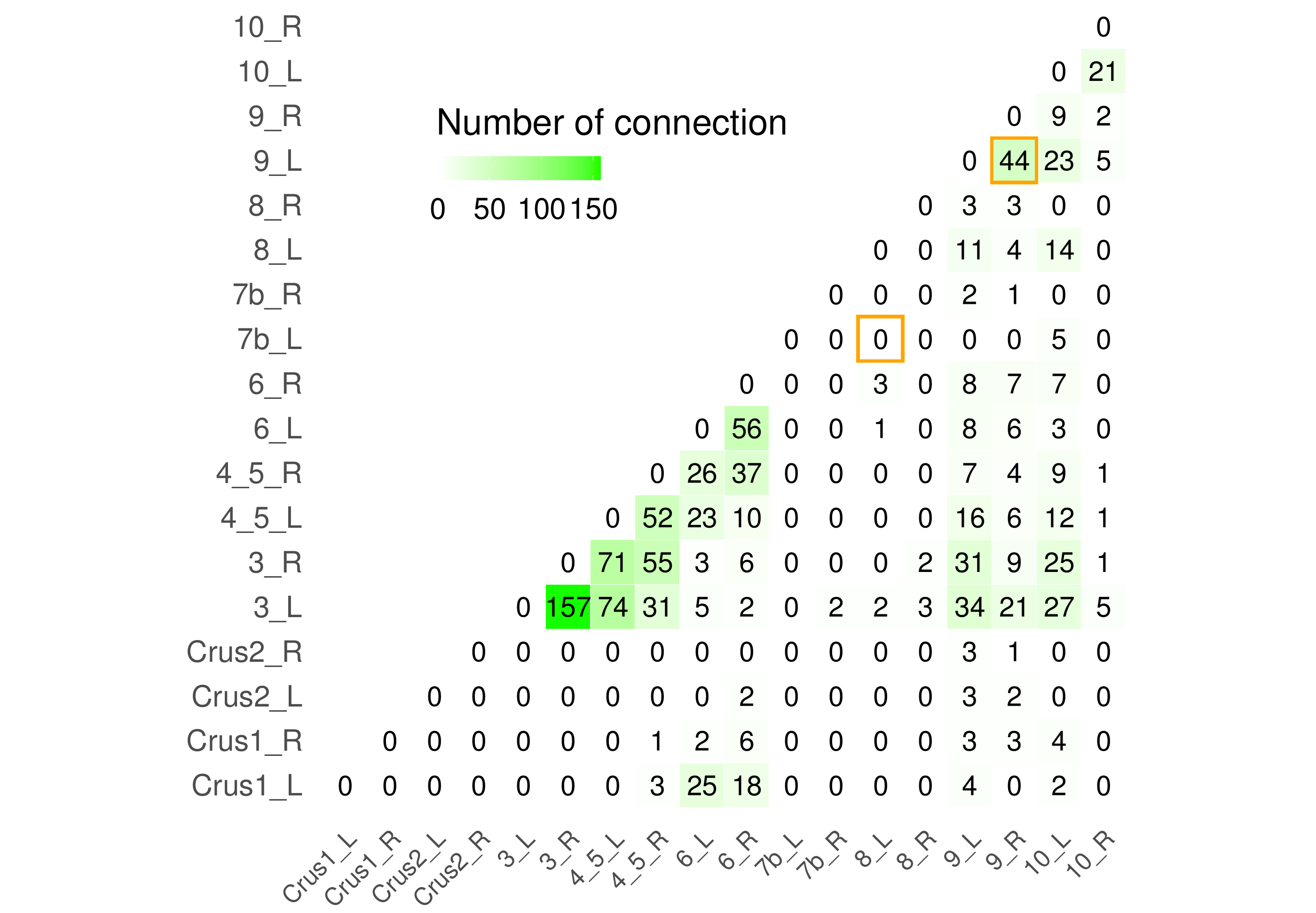}
\includegraphics[width=0.32\linewidth,valign=t]{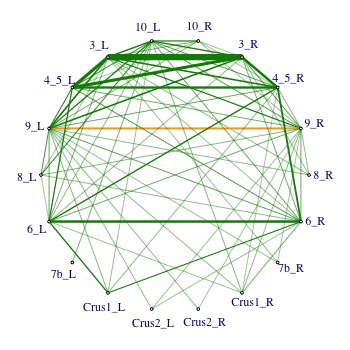}
\caption{Healthy group}
\end{subfigure}
\begin{subfigure}[b]{\textwidth}
\centering
\includegraphics[angle=270,width=0.45\linewidth,valign=t]{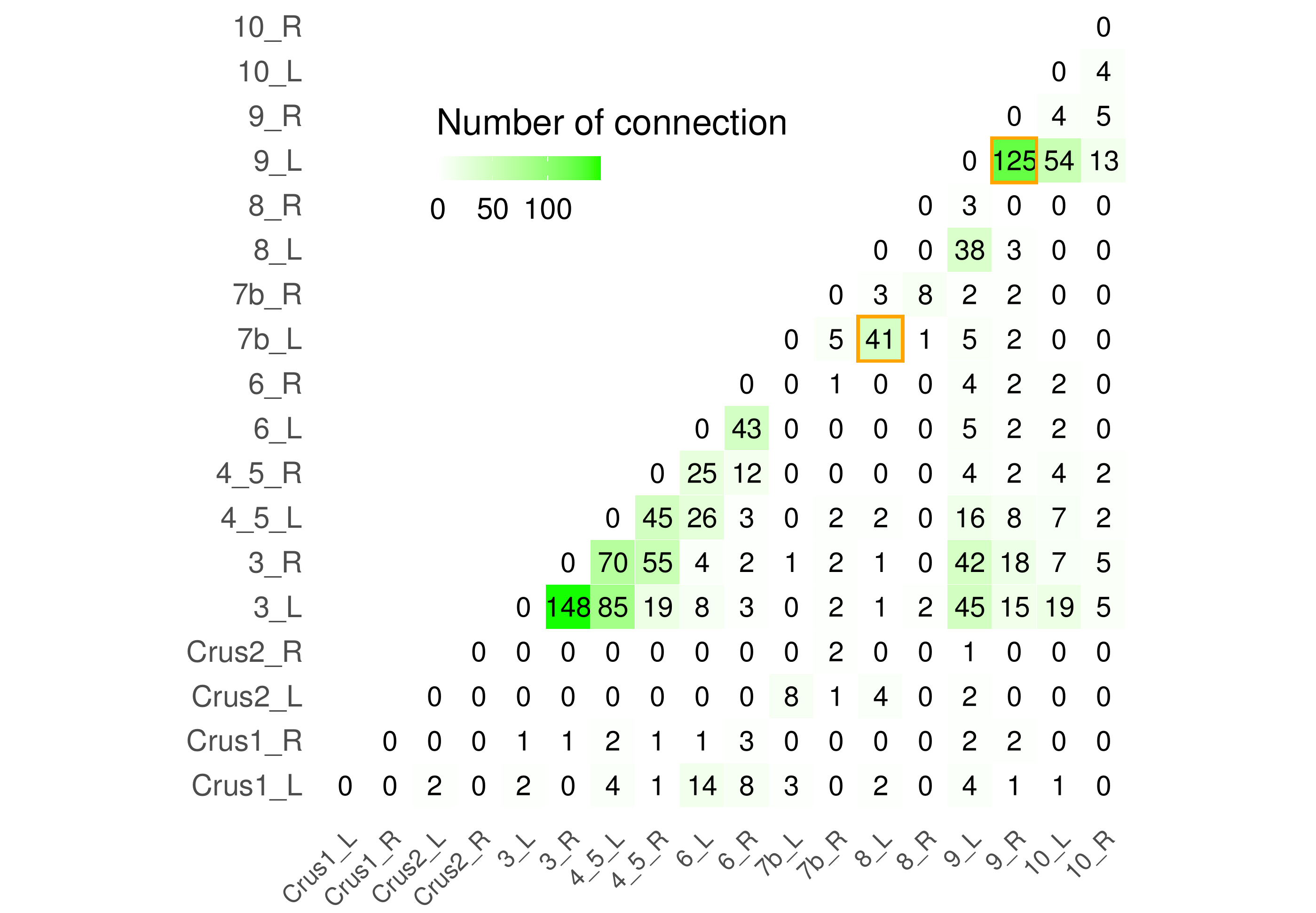}
\includegraphics[width=0.32\linewidth,valign=t]{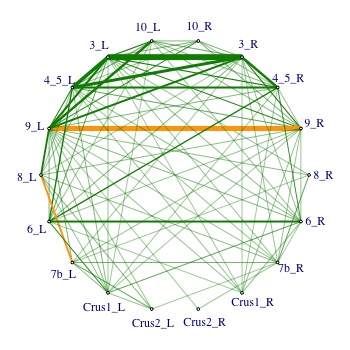}
\caption{ADHD group}
\end{subfigure}
\caption{[Result 2, GEN] Accumulated connections between different cerebellum regions estimated by generalized elastic net, showing the total connections between different cerebellum regions over 172 time points. The yellow squares on the left highlight the number of the edge 7b\_L - 8\_L and the edge 9\_L - 9\_R. }
\label{fig:app_gen13_connection}
\end{figure}

\begin{figure}[H]
\centering
\begin{subfigure}[b]{0.47\textwidth}
\centering
\includegraphics[width=\linewidth,valign=t]{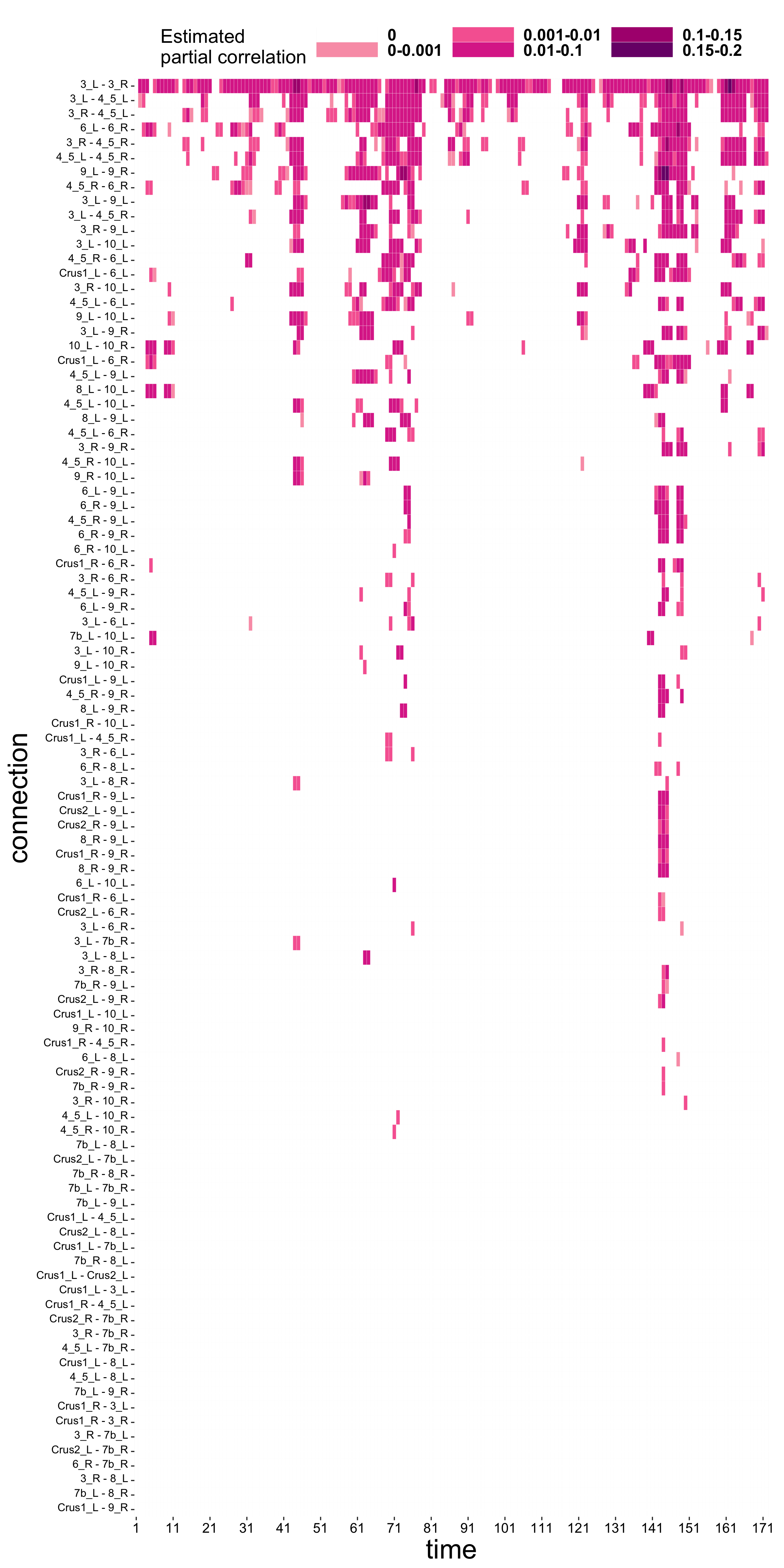}
\caption{Healthy group}
\end{subfigure}
~
\begin{subfigure}[b]{0.47\textwidth}
\centering
\includegraphics[width=\linewidth,valign=t]{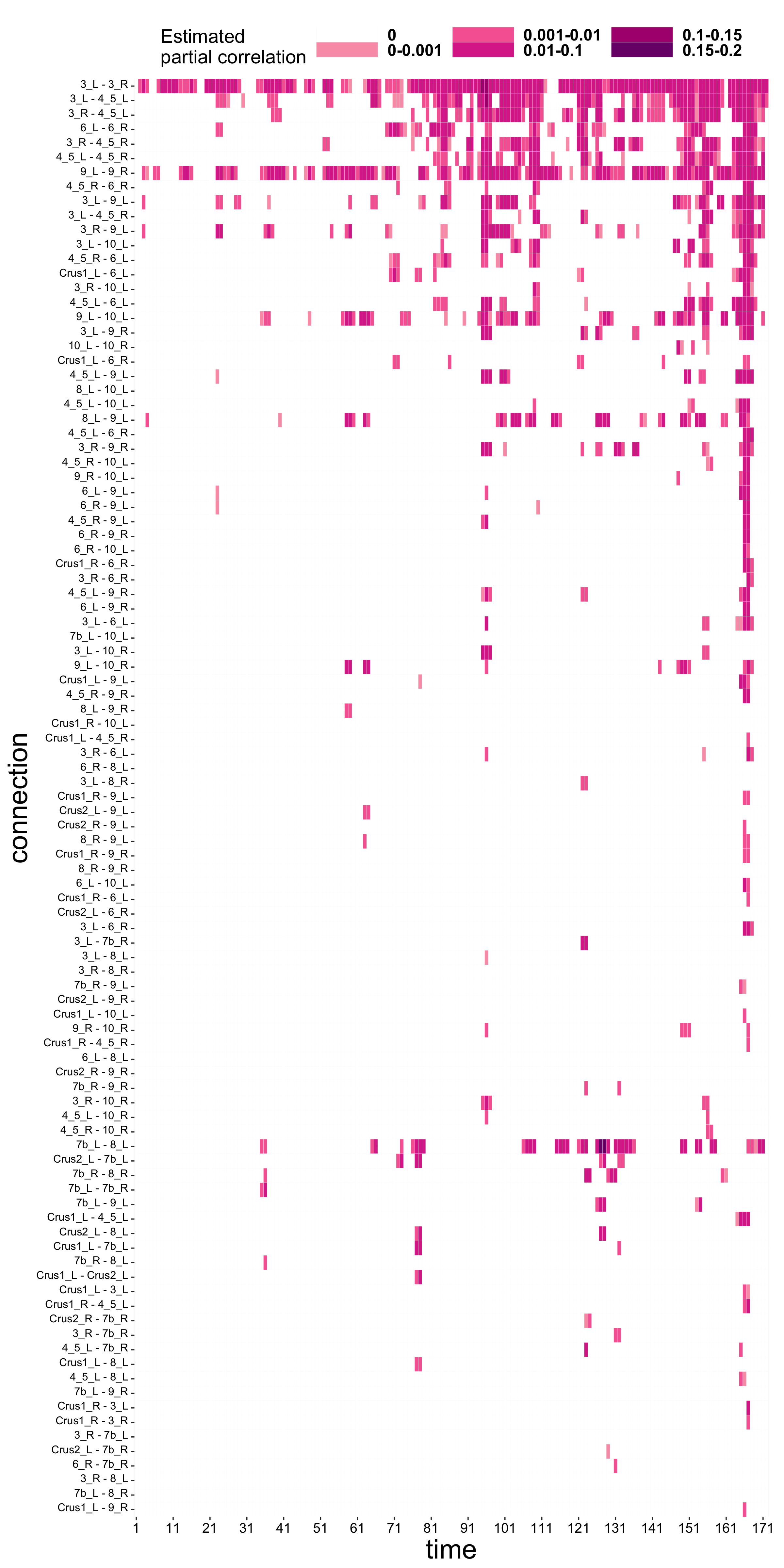}
\caption{ADHD group}
\end{subfigure}
\caption{[Result 2, GEN] Estimated partial correlations between different cerebellum regions estimated by generalized elastic net, illustrating how connections change over time. Each cell corresponds a connection at a given time point, and the color represents the magnitude of the estimated partial correlation value. }
\label{fig:app_gen13_table}
\end{figure}

\begin{figure}[H]
\centering
\begin{subfigure}[b]{\textwidth}
\centering
\includegraphics[angle=270,width=0.45\linewidth,valign=t]{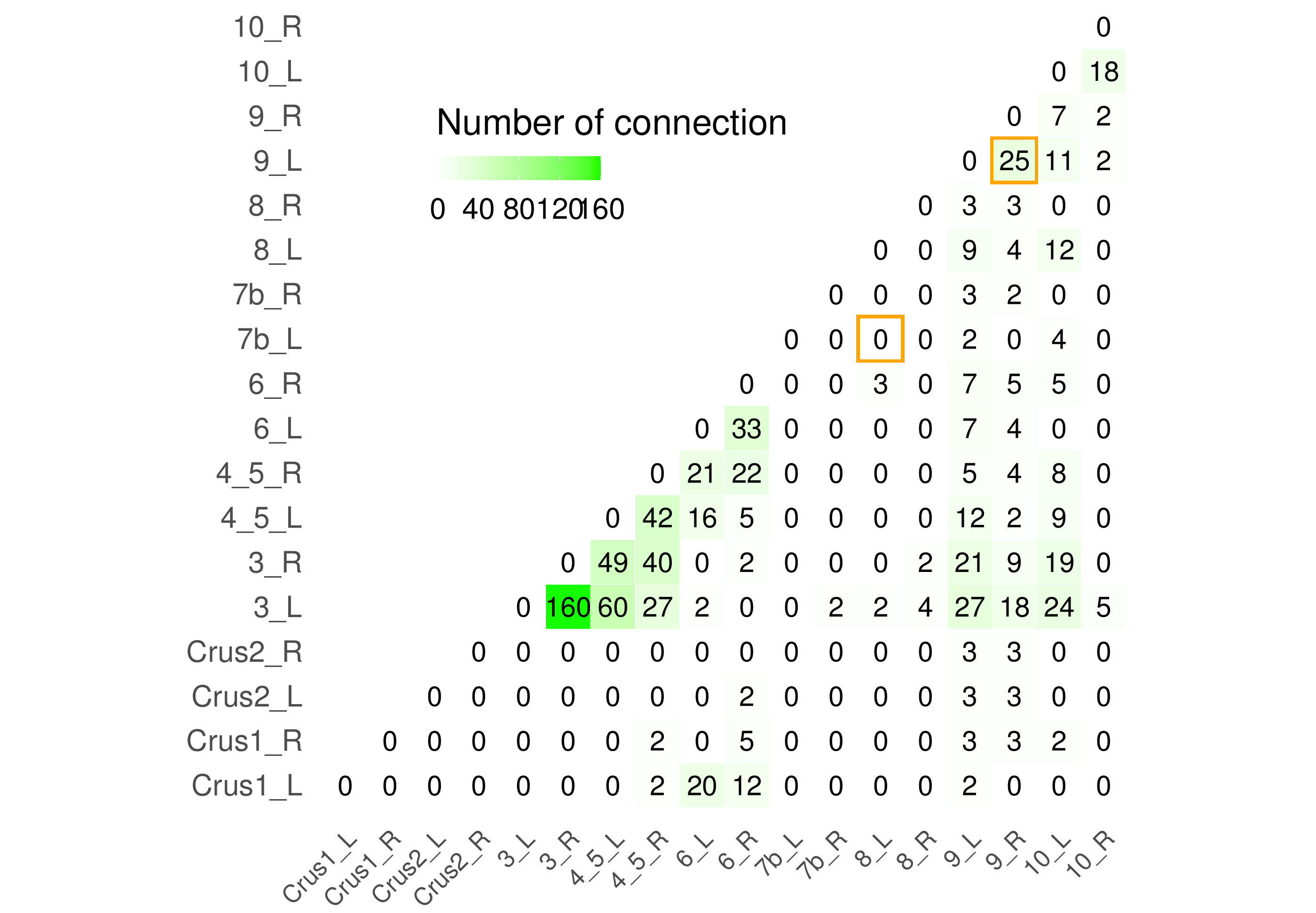}
\includegraphics[width=0.32\linewidth,valign=t]{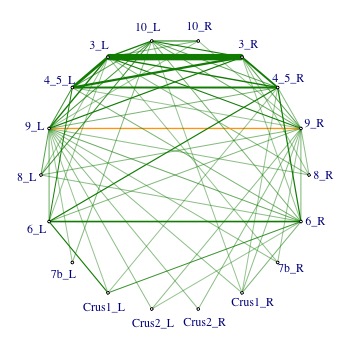}
\caption{Healthy group}
\end{subfigure}
\begin{subfigure}[b]{\textwidth}
\centering
\includegraphics[angle=270,width=0.45\linewidth,valign=t]{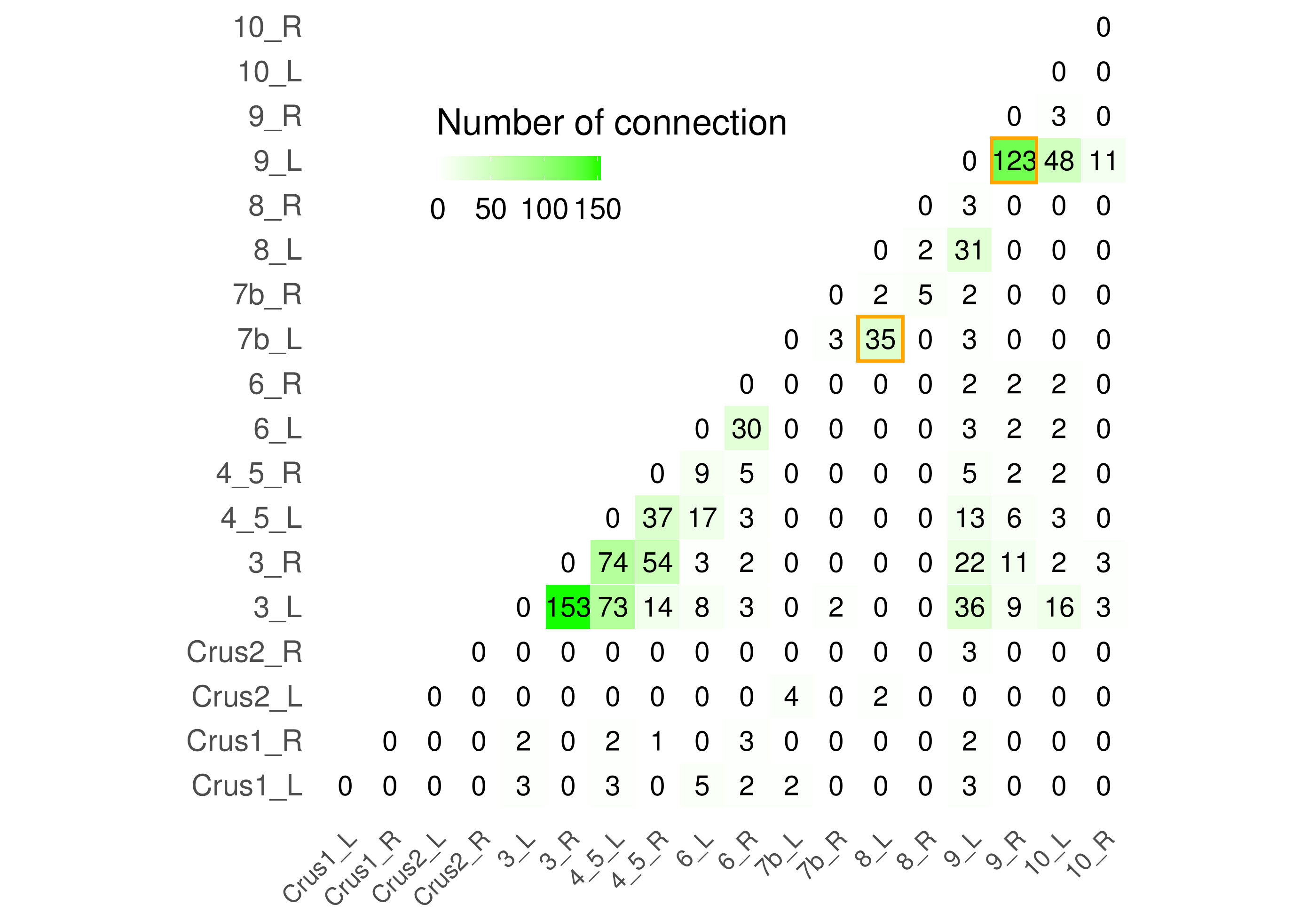}
\includegraphics[width=0.32\linewidth,valign=t]{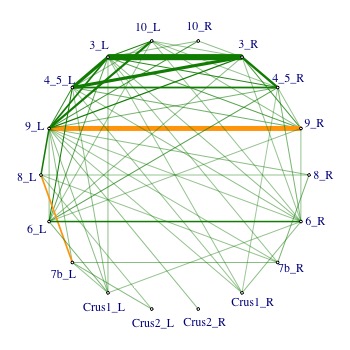}
\caption{ADHD group}
\end{subfigure}
\caption{[Result 2, GFL] Accumulated connections between different cerebellum regions estimated by generalized fused lasso. The figures show the total connections between different cerebellum regions over 172 time points. The yellow squares on the left highlight the number of the edge 7b\_L - 8\_L and the edge 9\_L - 9\_R. }
\label{fig:app_gfl16_connection}
\end{figure}

\begin{figure}[H]
\centering
\begin{subfigure}[b]{0.47\textwidth}
\centering
\includegraphics[width=\linewidth,valign=t]{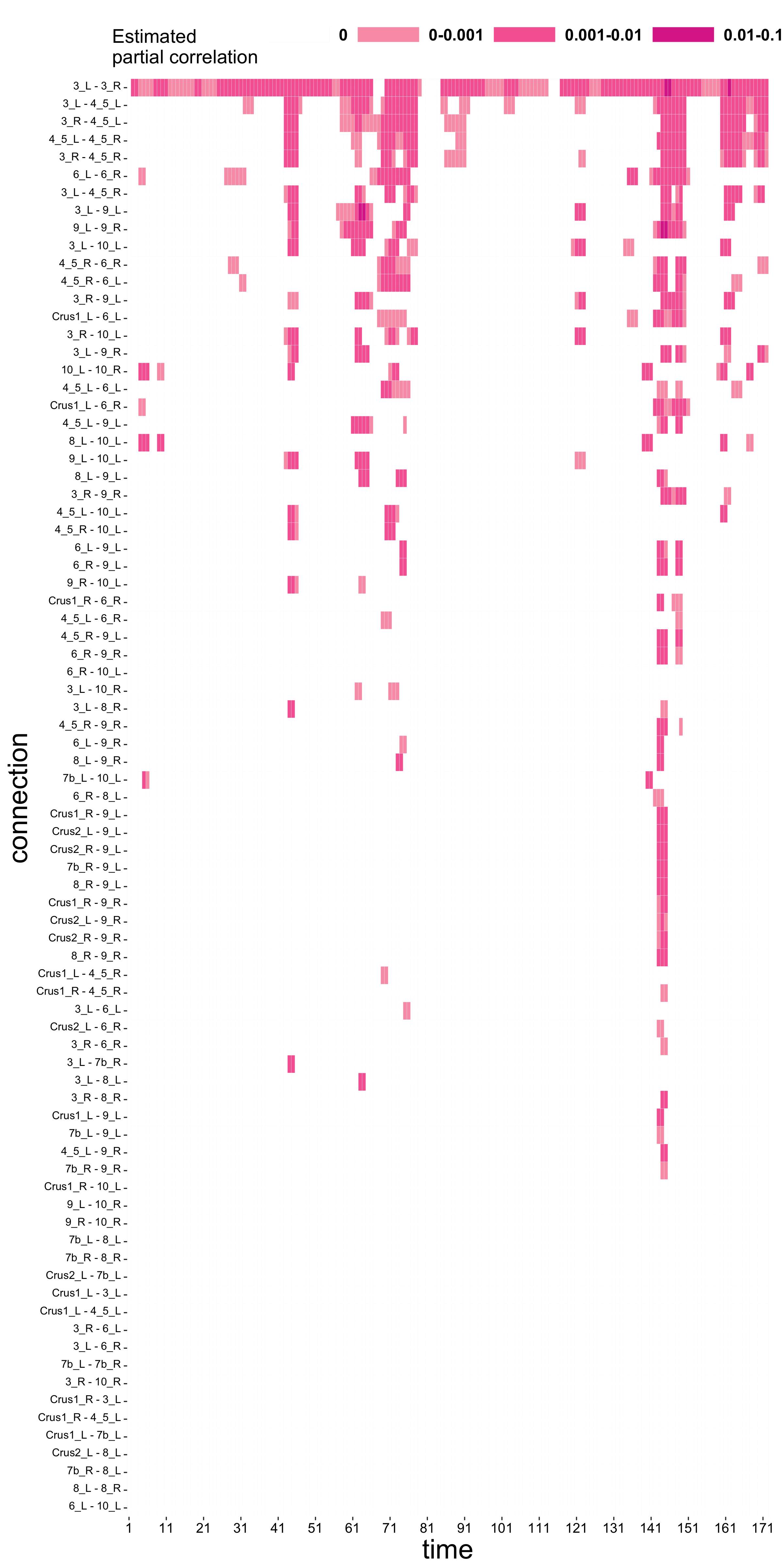}
\caption{Healthy group}
\end{subfigure}
~
\begin{subfigure}[b]{0.47\textwidth}
\centering
\includegraphics[width=\linewidth,valign=t]{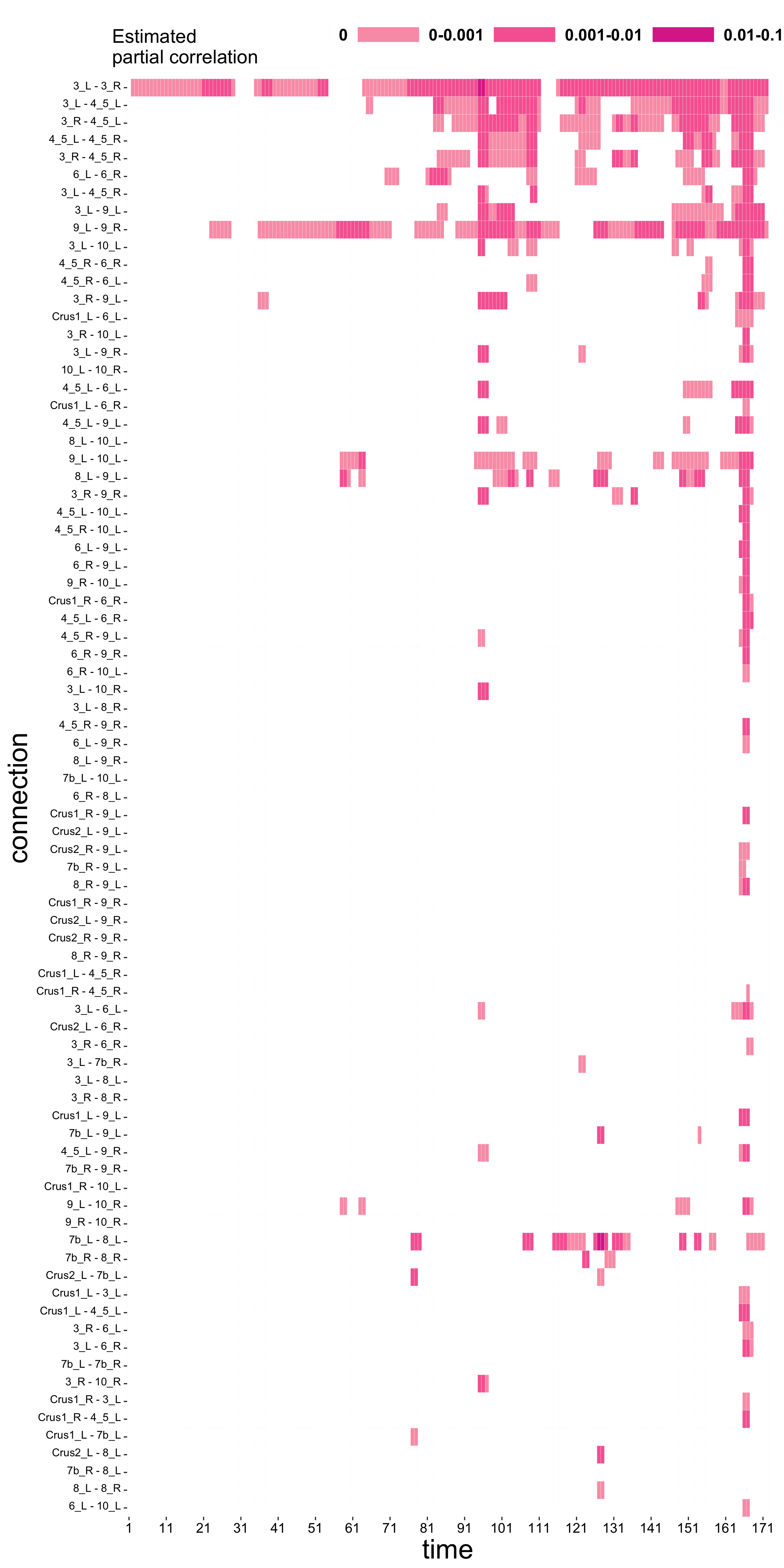}
\caption{ADHD group}
\end{subfigure}
\caption{[Result 2, GFL] Estimated partial correlations between different cerebellum regions estimated by generalized fused lasso, illustrating how connections change over time. Each cell corresponds a connection at a given time point, and the color represents the magnitude of the estimated partial correlation value. }
\label{fig:app_gfl16_table}
\end{figure}


\end{document}